\documentclass[12pt]{iopart}

\expandafter\let\csname equation*\endcsname\relax
\expandafter\let\csname endequation*\endcsname\relax

\usepackage{amsmath}

\usepackage{iopams}

\usepackage[utf8]{inputenc}
\usepackage[margin = 1in]{geometry}
\usepackage{amsfonts}
\usepackage{amssymb}
\usepackage{overpic}
\usepackage{epstopdf}
\usepackage{graphicx}
\usepackage{bm}
\usepackage[font = small]{caption}
\usepackage[font = small]{subcaption}
\usepackage[hidelinks]{hyperref}

\usepackage{mathtools}
\usepackage{physics}
\usepackage{cancel}
\newcommand{\mathds}{\mathbb}
\usepackage{amsthm}
\usepackage{tikz}
\newcommand*\circled[1]{\tikz[baseline=(char.base)]{
            \node[shape=circle,draw,inner sep=2pt] (char) {#1};}}

\renewcommand{\bs}{\boldsymbol}
\newcommand{\ol}{\overline}
\newcommand{\mbf}{\mathbf}
\newcommand{\s}{\\}

\newcommand{\eps}{\varepsilon}

\newcommand{\jac}{J_{\mbf u} \, }

\begin{document}

\title[Wave bifurcation in reaction-diffusion systems]
{Amplitude equations for wave bifurcations in reaction-diffusion systems}

\author{Edgardo Villar-Sep\'{u}lveda and Alan Champneys}

\address{Engineering Mathematics, University of Bristol, Bristol BS8 1TR, UK}
\ead{edgardo.villar-sepulveda@bristol.ac.uk}
\ead{a.r.champneys@bristol.ac.uk}
\vspace{10pt}
\begin{indented}
\item Submitted for publication: \today
\end{indented}

\begin{abstract}
    A wave bifurcation is the counterpart to a Turing instability in reaction-diffusion systems, but where the critical wavenumber corresponds to a pure imaginary pair rather than a zero temporal eigenvalue. Such bifurcations require at least three components and give rise to patterns that are periodic in both space and time. Depending on boundary conditions, these patterns can comprise either rotating or standing waves.
    
    Restricting to systems in one spatial dimension, complete formulae are derived for the evaluation of the coefficients of the weakly nonlinear normal form of the bifurcation up to order five, including those that determine the criticality of both rotating and standing waves. The formulae apply to arbitrary $n$-component systems ($n\geq 3$) and their evaluation is implemented in software which is made available as supplementary material.
    
    The theory is illustrated on two different versions of three-component reaction-diffusion models of excitable media that were previously shown to feature super- and subcritical wave instabilities and on a five-component model of two-layer chemical reaction. In each case, two-parameter bifurcation diagrams are produced to illustrate the connection between complex dispersion relations and different types of Hopf, Turing, and wave bifurcations, including the existence of several codimension-two bifurcations. 
\end{abstract}

%

\section{Introduction}
    Turing instability in reaction-diffusion systems whose diffusion rates differ in magnitude represents a fundamental mechanism for the breakdown of homogeneity and formation of patterns, across life and physical sciences, see \cite{Murray2, Meron} and references therein. The theory and application of such instabilities remain an active area of research, see e.g.~\cite{Krause} for a recent overview.

    In a recent paper \cite{Villar1}, we developed necessary and sufficient conditions for an $n$-component reaction-diffusion system to undergo a {\em Turing bifurcation}, which corresponds to the transition through a parameter value at which there is a critical wavenumber with a zero temporal eigenvalue. An estimate for the amplitude of the ensuing spatial pattern, and indeed whether it is stable or not, can be captured by weakly nonlinear normal-form analysis. So, in another recent paper \cite{TOMS}, we derived expressions for the coefficients of such normal forms for the Turing bifurcation, for arbitrary $n$-component systems of reaction-diffusion eqautions.

    The paper \cite{Villar1} points to another possible Turing-like instability, a {\em (Turing) wave bifurcation}, which would occur when the critical wavenumber corresponds to a pure-imaginary pair of eigenvalues. Whereas, when projected to an appropriate finite-dimensional subspace, the Turing bifurcation is typically a pitchfork bifurcation, the wave bifurcation is essentially a Hopf bifurcation.
        
    There is, however, the subtle question of boundary conditions. On a finite domain, allowable wavenumbers for a pitchfork-like instability are quantised. A Turing bifurcation arising from the continuous dispersion relation in the infinite-domain limit can be considered as an accumulation point of infinitely many such pitchfork bifurcations of different wave numbers \cite{Brena}. For the wave bifurcation, the same situation occurs for standing waves, namely a Turing wave bifurcation on an infinite domain, which would represent the infinite accumulation point of Hopf bifurcation points of systems posed on long finite domains with, say, Neumann boundary conditions. But, if we allow periodic boundary conditions, then the wave bifurcation can also give rise to travelling waves (rotating waves). Moreover, the Lyapunov coefficients that determine the criticality of the two different types of spatio-temporal patterns are not equal, in general. Thus, it is perfectly possible for standing waves to bifurcate supercritically whereas travelling waves are subcritical; or vice versa (see examples in Sec.~\ref{sec:examples} below). Also, at most one such wave can locally be stable, even if both bifurcate supercritical (see Fig.~\ref{fig:knobloch2par} below).

    There are a number of examples in the literature of Turing wave bifurcations, often in the context of models for excitable media, see \cite{Knobloch21,Vanag,Yang2,Yang,YochelisKnobloch}. The analog of such bifurcations in fluid mechanics is also well known (see e.g.~\cite{Knobloch86,Couette}).  But, it would be fair to say that this kind of pattern-forming instability has not received such a comprehensive treatment as the case of the stationary Turing patterns. One reason for this lack may be the fact that the minimum number of components for a reaction-diffusion system to feature a Turing instability is two, whereas at least three are required in the case of a wave bifurcation, see \cite{Villar1}. Thus, most canonical models of reaction-diffusion systems, such as the Brusselator, Gray-Scott model, etc.~having only two components, are incapable of developing wave instabilities.
        
    The purpose of this paper is to derive the general amplitude equation for a wave bifurcation in systems with $n$ components, for any $n\geq 3$. We will provide explicit formulae for the first Lyapunov coefficient, the coefficient of the appropriate cubic term that characterises the criticality of the bifurcation, both for standing-wave and rotating-wave solutions. The amplitude equations are equivalent to the normal form of a Hopf bifurcation with $O(2)$ symmetry, see \cite{GoStSc:88}. This normal form was first derived by Knobloch \cite{Knobloch86} at least up to third-order terms, in the context of convection in binary mixtures. That work also contains a partial unfolding of the dynamics, showing the bifurcation behaviour of standing and travelling waves, which we reproduce here.
     
    Our method of derivation of amplitude equations is quite standard; based on the multiple-scale asymptotic methodology developed by Elphick \emph{et al} \cite{tirapegui}, an extension to our earlier equivalent derivation for Turing bifurcations \cite{TOMS}. The details in the wave bifurcation case also borrow from the treatment in the book by Kuznetsov\cite[sec.~5.6]{Kuznetsov} for computation of Hopf bifurcation normal forms in reaction-diffusion systems. See also \cite{Ahamadi}. The novelty here is our presentation of general expressions for normal-form coefficients for arbitrary reaction-diffusion systems. Also, we go beyond the third-order form, to consider fifth-order terms, which come into play at a codimension-two bifurcation point (referred to as a Bautin bifurcation \cite{Kuznetsov} for a finite-dimensional Hopf bifurcation) where the appropriate third-order coefficient vanishes. The nature of such a normal form relies on the sign of an appropriate fifth-order term. Hence, we will also seek general formulae for these coefficients, both for the cases of standing and travelling waves. Moreover, all of the algorithms have been implemented in software, both {\tt Python} and {\tt Mathematica},    which is made freely available \cite{criticality-wave} for users to perform the necessary calculations on their own examples.
        	
    The rest of the paper is outlined as follows. Section \ref{sec:main} introduces the normal form and its interpretation in terms of bifurcations of standing and travelling waves. Detailed formulae are presented for the normal form coefficients from an arbitrary $n$-component reaction-diffusion system. Section \ref{sec:coefs} outlines the steps taken to derive these formulae, with algebraic details for the higher-order terms relegated to Appendix \ref{ap:order5}. Section \ref{sec:examples} considers several example systems that both illustrate the theory and point to further interactions between wave bifurcations and other instabilities. Finally, Section \ref{sec:concl} draws conclusions and suggests avenues for future work.

\section{The wave bifurcation normal form} \label{sec:main}
    Consider a general $n$-component reaction-diffusion system written in the form:
	\begin{align}
		\partial_t \mbf u &= \mbf f(\mbf u, \mu) + \mathbb D(\mu) \, \partial_{xx} \mbf u, \label{geneq}
	\end{align}
	where $\mbf u = \left(u_1, \ldots, u_n\right)^\intercal$ is a vector of variables in $\mathbb R^n$, $n\geq 3$, $\mbf f = \left(f_1, \ldots, f_n\right)^\intercal\in \mathcal C^5$ is a real-valued vector of smooth functions, $\mu\in \mathbb R$ is a parameter of the system, and $\mathbb D(\mu)$ is a positive-definite diffusion matrix, whose entries may also depend on $\mu$.
	
	Next, let us assume that $\mbf P = \mbf P(\mu)$ is a stable, isolated homogeneous steady state of system \eqref{geneq}, that is $\mbf f(\mbf P(\mu), \mu) = 0$, and if we define the Jacobian matrix of the diffusion-free system at $\mbf P$ by
	\begin{align*}
		\jac \mbf f(\mbf P(\mu), \mu) = \left(\frac{\partial f_i}{\partial u_j}(\mbf P(\mu), \mu)
		\right)_{1\leq i, j\leq n}, \qquad \mbox{then} \quad 
        \det \left(\jac \mbf f(\mbf P(\mu), \mu) \right) \neq 0,
    \end{align*}
    for all $\mu$-values of interest.

    To find eigenfunctions of the full system linearised around $\mbf P$, we seek functions of the form $\exp(ikx + \lambda t) \, \bs \phi_1^{[1]}$ with $\bs \phi_1^{[1]} \in \mathbb C^n$, $k \in \mathbb R$, $\lambda \in \mathbb C$ such that
    \begin{align}
        \det\left(\jac \mbf f(\mbf P(\mu), \mu) - k^2 \, \mathbb D(\mu) - \lambda \, I\right) &= 0. \label{eq:lambdadef}
    \end{align}
    Moreover, as is usual in the theory of nonlinear waves, we define the multivalued \emph{dispersion relation} for each of the branches $\lambda_j(k; \mu)$, $j = 1, \ldots, n$ as $k$ varies from $0$ to $+\infty$. We then define a local bifurcation by the existence of a critical parameter value $\mu^*$ for which there exists $1\leq j \leq n$ and $k^*>0$
    \begin{equation}
        \Re\left[\lambda_j(k^*; \mu^*)\right] = 0, \qquad \frac{\dd}{\dd k} \Re\left[\lambda_j(k^*; \mu^*)\right] = 0 \quad \mbox{ and } \quad \frac{\dd^2}{\dd k^2} \mbox{Re}\left[\lambda_j(k^*; \mu^*\right] < 0. \label{eq:critical}
    \end{equation}
    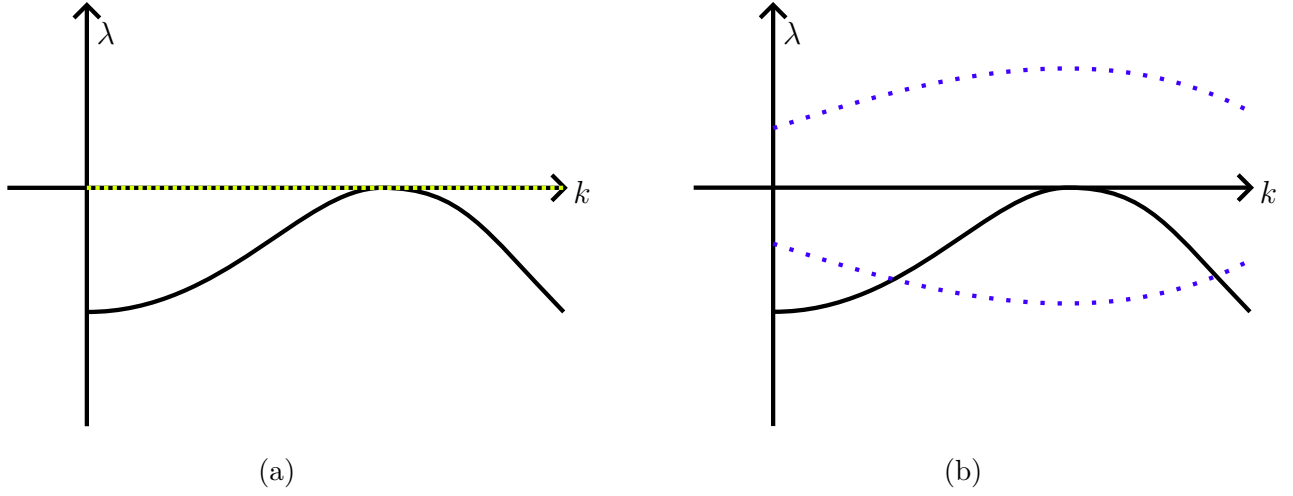
\begin{figure}
         \centering
         \begin{subfigure}[b]{0.45\textwidth}
             \centering
             \begin{tikzpicture}
	            \node (image) at (0,0) {
                    \includegraphics[width=\textwidth]{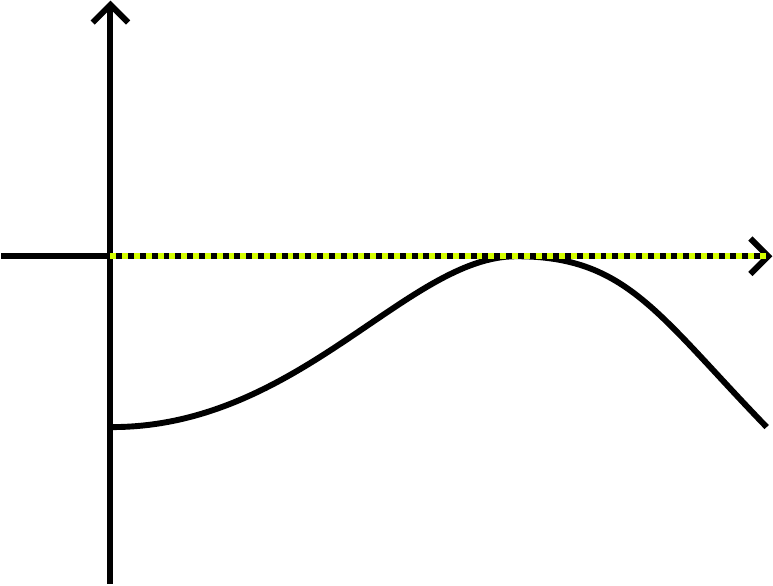}
                };
                \node (text) at (3.9, 0.3) {$k$};
                \node (text) at (-2.4, 2.4) {$\lambda$};
	        \end{tikzpicture}
            \caption{}
         \end{subfigure}
         \hfill
         \begin{subfigure}[b]{0.45\textwidth}
             \centering
             \begin{tikzpicture}
	            \node (image) at (0,0) {
                    \includegraphics[width=\textwidth]{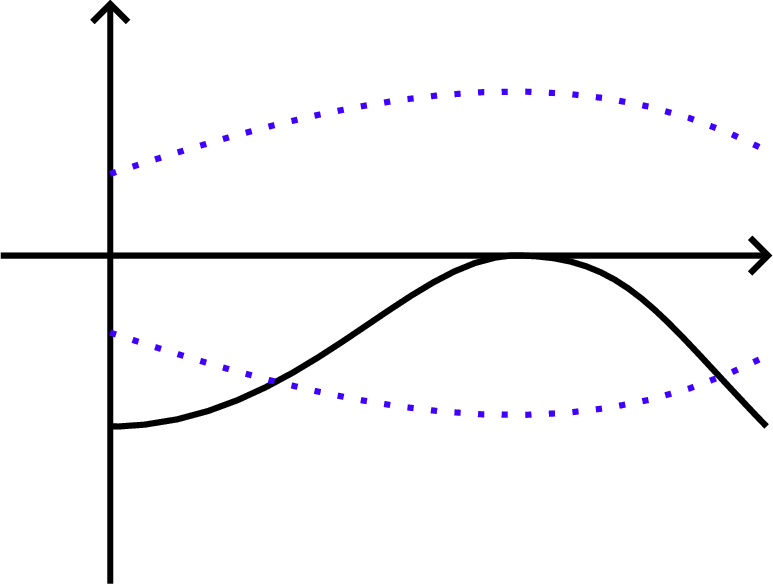}
                };
                \node (text) at (3.9, 0.3) {$k$};
                \node (text) at (-2.4, 2.4) {$\lambda$};
	        \end{tikzpicture}
             \caption{}
         \end{subfigure}
        \caption{One branch of dispersion relations corresponding to (a) Turing bifurcation and (b) wave bifurcation. In (a), the eigenvalue plotted is real whilst in (b), two complex-conjugate eigenvalues are plotted. In each plot, solid lines represent the real part of an eigenvalue, and dotted lines correspond to the imaginary part.}
        \label{fig:dispersion}
    \end{figure}

    We introduce a small parameter  $\eps = \mu - \mu^*$ so that the critical parameter value is $\eps = 0$, and despite the abuse of notation, we rewrite \eqref{geneq} as
    \begin{align}
		\partial_t \mbf u &= \mbf f(\mbf u, \eps) + \mathbb D(\eps) \, \partial_{xx} \mbf u. \label{eq:geneq}
	\end{align}
    Note there are two generic ways in which \eqref{eq:critical} can be satisfied, see Fig.~\ref{fig:dispersion}. If the critical $\lambda_j$ is real, then we have a pitchfork-like-instability which in this context is typically referred to as a Turing bifurcation; a complete derivation of the corresponding normal form in that case was given in our earlier work \cite{TOMS}. In this work, we are interested in the case where the critical $\lambda_j$ has a non-zero imaginary part. Owing to the fact that \eqref{eq:geneq} is real, then the complex conjugate $\ol \lambda_j$ must also be an eigenvalue. Hence, we have a form of spatially-extended Hopf bifurcation, which we have termed a \emph{wave bifurcation}.
        
    The normal form of the wave bifurcation has appeared in a number of publications, where it has also been termed a Hopf bifurcation with $O(2)$-symmetry \cite[Ch.XVII]{GoStSc:88}. One of the earliest treatments was in the pioneering work of Knobloch \cite{Knobloch86}, in the context of convection in a binary fluid mixture. See also the independent work by van Gils and Malet-Paret \cite{VanGils}.
    
    Under appropriate non-degeneracy conditions, we build the normal form via the ansatz 
	\begin{align}
		\mbf u = \left(A \, e^{ikx} \, \bs \phi_1^{[1]} + \bar A \, e^{-ikx} \, \overline{\bs \phi_1^{[1]}}\right) + \left(B \, e^{ikx} \, \ol{\bs \phi_1^{[1]}} + \bar B \, e^{-ikx} \, \bs \phi_1^{[1]}\right) + \mathcal O\left(A, B\right)^2, \label{ansatz}
	\end{align}
	where $\bs \phi_1^{[1]}\in \mathbb C^n$, $\norm{\bs \phi_1^{[1]}} = 1$, $(A, B) = (A(t), B(t))\in \mathbb C^2$ are two complex variables related to the amplitude of a spatiotemporal pattern that arises at the bifurcation point, and $\mathcal O(A, B)^q$ means terms of order $q$ in $A$ and $B$.

    Then, the normal form can be written in the form of equations for $A$ and $B$ as:
	\begin{align}
        \partial_t A &= \begin{multlined}[t][10cm]
            \left(C_1^{[1,0,1]} \, \eps + i\omega\right) A + C_1^{[3,0]} \, |A|^2 \, A + C_1^{[1,2]} \, A \, |B|^2
            \s 
            + C_1^{[5,0]} \, |A|^4 \, A + C_1^{[1,4]} \, |B|^4 \, A + C_1^{[3,2]} \, |A|^2 \, |B|^2 \, A + \mathcal{O}(A,B)^7 + \eps \, \mathcal O(A, B)^3, 
        \end{multlined} \label{eq:nf1}
        \s 
        \partial_t B &= \begin{multlined}[t][10cm]
            \left(\ol{C_1^{[1,0,1]}} \, \eps - i\omega\right) B + \ol{C_1^{[3,0]}} \, |B|^2 \, B + \ol{C_1^{[1,2]}} \, |A|^2 \, B
            \s 
            + \ol{C_1^{[5,0]}} \, |B|^4 \, B + \ol{C_1^{[1,4]}} \, |A|^4 \, B + \ol{C_1^{[3,2]}} \, |A|^2 \, |B|^2 \, B + \mathcal{O}(A,B)^7 + \eps \, \mathcal O(A, B)^3, 
        \end{multlined} \label{eq:nf2}
    \end{align}    
    where 
    $$
        C_1^{[1, 0, 1]} = \partial_\mu \, \lambda_j\left(k^*, \mu^*\right),
    $$
    with $\lambda_j$ being the eigenvalue going through the imaginary axis at the bifurcation point.

    \subsection{The dynamics of the normal form}
        The full dynamics of the normal form are not known, but there is evidence of complex spatio-temporal dynamics including mechanisms that generate chaotic dynamics. It can be viewed as an extension of the normal form for the Takens-Bogdanov bifurcation with $O(2)$-symmetry (which applies in the limit that $\omega \to 0$), for which many complex bifurcation scenarios are known, see \cite{Dangelmayr,Rucklidge}. Our interest here is to be able to find closed-form expressions for the normal-form coefficients to determine whether the simplest kinds of solution bifurcate either super- or subcritically, together with the associated stability of the solutions. Fig.~\ref{fig:knobloch2par} shows a summary of what is known if we impose periodic boundary conditions, so that both travelling and standing waves are allowed in the system. In this case, when both kinds of waves bifurcate supercritically, only one of the associated solutions is stable. On the other hand, with homogeneous Neumann boundary conditions, we can make the reduction $A = \bar B$, and traveling waves are not allowed. In that case, whenever the standing waves bifurcate supercritically (resp.~subcritically), then they are stable (resp.~unstable) right after the bifurcation.
        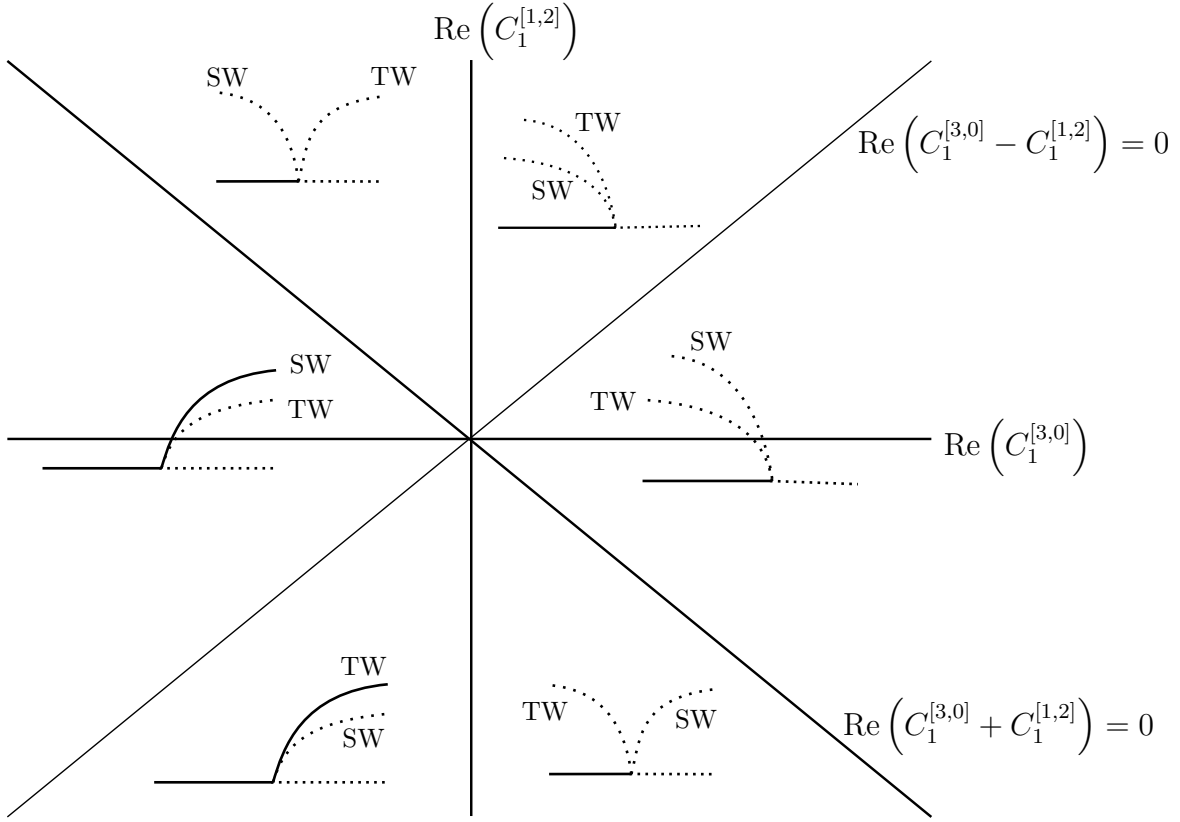
\begin{figure}
            \begin{center}
            \begin{tikzpicture}
                \node (image) at (0,0) {
                    \includegraphics[scale = 0.7]{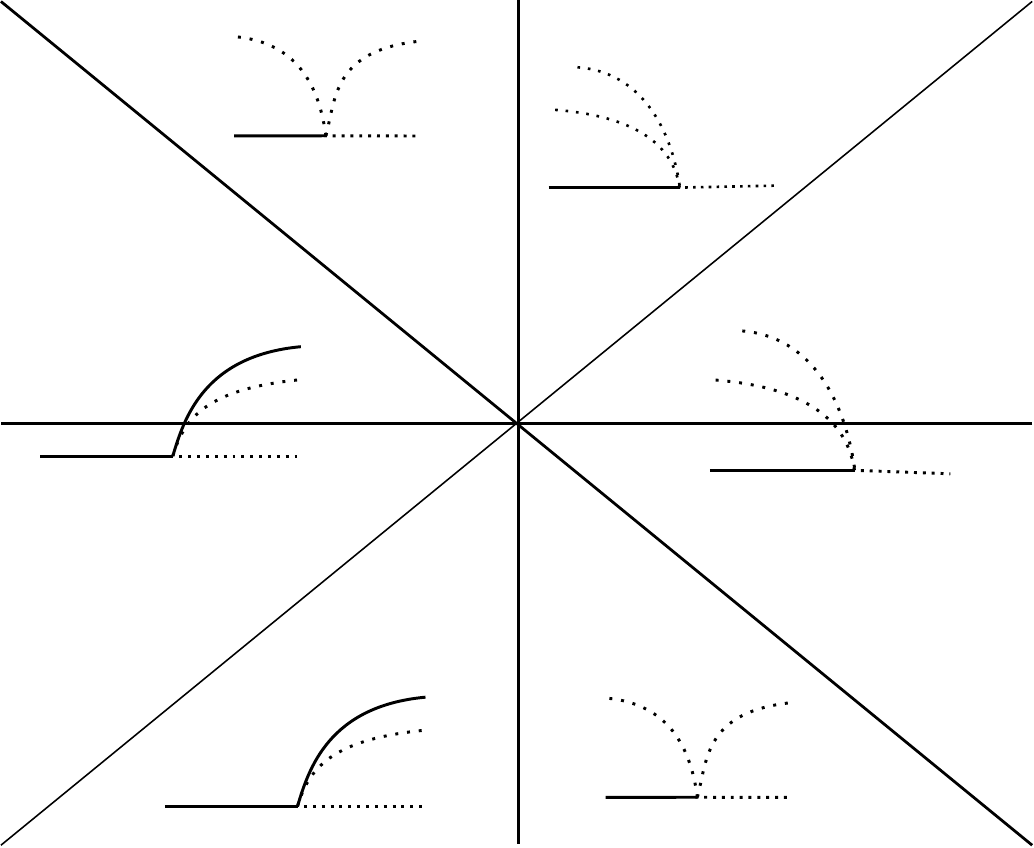}
                };
                \node (text) at (7.25, -0.1) {$\Re\left(C_1^{[3, 0]}\right)$};
                \node (text) at (0.5, 5.4) {$\Re\left(C_1^{[1, 2]}\right)$};
                
                \node (text) at (7.0, -3.8) {$\Re\left(C_1^{[3, 0]} + C_1^{[1, 2]}\right) = 0$};
                \node (text) at (7.2, 3.9) {$\Re\left(C_1^{[3, 0]} - C_1^{[1, 2]}\right) = 0$};
                \node (text) at (-1.0, 4.8) {\footnotesize{TW}};
                \node (text) at (-3.2, 4.8) {\footnotesize{SW}};
                
                \node (text) at (-2.1, 0.4) {\footnotesize{TW}};
                \node (text) at (-2.1, 1.0) {\footnotesize{SW}};
                
                \node (text) at (-1.4, -3.0) {\footnotesize{TW}};
                \node (text) at (-1.4, -3.95) {\footnotesize{SW}};
                
                \node (text) at (1.0, -3.6) {\footnotesize{TW}};
                \node (text) at (3.0, -3.7) {\footnotesize{SW}};
                
                \node (text) at (1.9, 0.5) {\footnotesize{TW}};
                \node (text) at (3.2, 1.3) {\footnotesize{SW}};
                
                \node (text) at (1.7, 4.2) {\footnotesize{TW}};
                \node (text) at (1.1, 3.3) {\footnotesize{SW}};
            \end{tikzpicture}
            \caption{The region of parameter space in which standing waves (SW) and travelling waves (TW) bifurcate super- or subcriticality as a function of the cubic coefficients of the normal form \eqref{eq:nf1} and \eqref{eq:nf1}. This figure is based on \cite[Fig.1]{Knobloch86}.}
            \label{fig:knobloch2par}
            \end{center}
        \end{figure}
        
        To understand more about the bifurcation of travelling and standing waves, observe that the change of variables $(A, B) \to \left(A \, e^{i \omega t}, B \, e^{- i \omega t}\right)$ transforms the system \eqref{eq:nf1} and \eqref{eq:nf2} into
        \begin{align*}
            \partial_t A &= \begin{multlined}[t][10cm]
                C_1^{[1,0,1]} \, \varepsilon A + C_1^{[3,0]} \, |A|^2 \, A + C_1^{[1,2]} \, A \, |B|^2
                \s 
                + C_1^{[5,0]} \, |A|^4 \, A + C_1^{[1,4]} \, |B|^4 \, A + C_1^{[3,2]} \, |A|^2 \, |B|^2 \, A + \mathcal O(A, B)^7 + \eps \, \mathcal O(A, B)^3, 
            \end{multlined}
            \s 
            \partial_t B &= \begin{multlined}[t][10cm]
                \ol{C_1^{[1,0,1]}} \, \varepsilon B + \ol{C_1^{[3,0]}} \, |B|^2 \, B + \ol{C_1^{[1,2]}} \, |A|^2 \, B
                \s 
                + \ol{C_1^{[5,0]}} \, |B|^4 \, B + \ol{C_1^{[1,4]}} \, |A|^4 \, B + \ol{C_1^{[3,2]}} \, |A|^2 \, |B|^2 \, B + \mathcal O(A, B)^7 + \eps \, \mathcal O(A, B)^3. 
            \end{multlined}
        \end{align*}
        Furthermore, the ansatz \eqref{ansatz} becomes
        \begin{align*}
            \mbf u = \left(A \, e^{i\left(kx + \omega t\right)} \, \bs \phi_1^{[1]} + \bar A \, e^{-i\left(kx + \omega t\right)} \, \overline{\bs \phi_1^{[1]}}\right) + \left(B \, e^{i\left(kx - \omega t\right)} \, \ol{\bs \phi_1^{[1]}} + \bar B \, e^{-i\left(kx - \omega t\right)} \, \bs \phi_1^{[1]}\right) + \mathcal O\left(A, B\right)^2.
        \end{align*}
        Note that the subspace $B = 0$ (equivalently $A = 0$) is invariant, and solutions in this subspace represent travelling waves $\sim A \, e^{i \left(k x + \omega t\right)} + c.c.$ (equivalently $\sim B \, e^{i \left(k x - \omega t\right)} + c.c.$), where $c.c.$ stands for \textit{complex-conjugate}. Within this subspace, the truncated form of \eqref{eq:nf1} and \eqref{eq:nf2} becomes
        \begin{align}
            \partial_t A &= \left(C_1^{[1,0,1]} \, \varepsilon + i\omega\right) A + C_1^{[3,0]} \, |A|^2 \, A, + C_1^{[5,0]} |A|^4 A \label{B=0}
        \end{align}
        (and the complex conjugate of this equation for $B$). From this, we see that the criticality of travelling waves is determined by the sign of the coefficient
        $$
            \ell_{1r} := \Re \left(C_1^{[3, 0]}\right),
        $$
        supercritical if negative and subcritical if positive. 
        
        On the other hand, $B = \pm \bar A$ also define invariant subspaces. These correspond to standing waves, of the form $u \sim A \, \cos(kx) \, e^{i \omega t} + c.c.$. Within such subspaces, we have the reduced equation
        \begin{align}
            \partial_t A &=\left(C_1^{[1, 0, 1]} \, \varepsilon + i\omega \right) A + \left(C_1^{[3, 0]} + C_1^{[1, 2]}\right) \, |A|^2 \, A + \left(C_1^{[5, 0]} + C_1^{[1, 4]} + C_1^{[3, 2]}\right) \, |A|^4 \, A. \label{B=Abar}
        \end{align}
        Therefore, the criticality of the bifurcation of standing waves is determined by the sign of
        $$
            \ell_{1s} := \Re\left(C_1^{[3, 0]} + C_1^{[1, 2]}\right),
        $$
        supercritical if negative and subcritical otherwise.
        
        Now, as with the normal form of a Hopf bifurcation in finite dimensions, if either of the third-order coefficients $\ell_{1r}$ or $\ell_{1s}$ is zero, then one has a codimension-two degenerate Hopf bifurcation, also known as a \emph{Bautin bifurcation} \cite{Kuznetsov}. Then, via restriction to the appropriate invariant subspaces, it is the sign of the appropriate fifth-order coefficients
        $$
            \ell_{2r}: = \Re \left(C_1^{[5, 0]}\right), \qquad \ell_{2s}:= \Re\left(C_1^{[5, 0]} + C_1^{[1, 4]} + C_1^{[3, 2]}\right)
        $$
        that determines the nature of the Bautin bifurcation for travelling and standing waves, respectively.
    
        It is useful to also explain a few details of what happens close to such Bautin points, in particular to find leading-order expressions for the fold curve of standing or travelling waves in a two-parameter plot.
        
        When neglecting higher-order terms, the amplitude equations look like
        \begin{align*}
            \partial_t A &= \begin{multlined}[t][10cm]
                \left(C_1^{[1,0,1]} \, \eps + i\omega\right) A + C_1^{[3,0]} \, |A|^2 \, A + C_1^{[1,2]} \, A \, |B|^2
                \s 
                + C_1^{[5,0]} \, |A|^4 \, A + C_1^{[1,4]} \, |B|^4 \, A + C_1^{[3,2]} \, |A|^2 \, |B|^2 \, A, 
            \end{multlined}
            \s 
            \partial_t B &= \begin{multlined}[t][10cm]
                \left(\ol{C_1^{[1,0,1]}} \, \eps - i\omega\right) B + \ol{C_1^{[3,0]}} \, |B|^2 \, B + \ol{C_1^{[1,2]}} \, |A|^2 \, B
                \s 
                + \ol{C_1^{[5,0]}} \, |B|^4 \, B + \ol{C_1^{[1,4]}} \, |A|^4 \, B + \ol{C_1^{[3,2]}} \, |A|^2 \, |B|^2 \, B. 
            \end{multlined}
        \end{align*}
        As usual, if we consider the following change of variables:
        \begin{align*}
            \left(A, \bar A, B, \bar B\right) = \left(R_1 \, e^{i \varphi_1}, R_1 \, e^{- i \varphi_1}, R_2 \, e^{i \varphi_2}, R_2 \, e^{-i \varphi_2}\right),
        \end{align*}
        we get
        \begin{align*}
            \begin{pmatrix}
                \partial_t A
                \\
                \partial_t \bar A
                \\
                \partial_t B
                \\
                \partial_t \bar B
            \end{pmatrix} = \begin{pmatrix}
                e^{i \varphi_1} & i R_1 \, e^{i \varphi_1} & 0 & 0
                \\
                e^{- i \varphi_1} & - i R_1 \, e^{- i \varphi_1} & 0 & 0
                \\
                0 & 0 & e^{i \varphi_2} & i R_2 \, e^{i \varphi_2}
                \\
                0 & 0 & e^{- i \varphi_2} & - i R_2 \, e^{- i \varphi_2}
            \end{pmatrix} \begin{pmatrix}
                \partial_t R_1
                \\
                \partial_t \varphi_1
                \\
                \partial_t R_2
                \\
                \partial_t \varphi_2
            \end{pmatrix},
        \end{align*}
        which implies
        \begin{align*}
            \begin{pmatrix}
                \partial_t R_1
                \\
                \partial_t \varphi_1
                \\
                \partial_t R_2
                \\
               \partial_t \varphi_2
            \end{pmatrix} &= \frac{1}{2 \, R_1 \, R_2} \begin{pmatrix}
                R_1 \, R_2 \, e^{- i \varphi_1} & R_1 \, R_2 \, e^{i \varphi_1} & 0 & 0
                \\
                - i R_2 \, e^{-i \varphi_1} & i R_2 \, e^{i \varphi_1} & 0 & 0
                \\
                0 & 0 & R_1 \, R_2 \, e^{- i \varphi_2} & R_1 \, R_2 \, e^{i \varphi_2}
                \\
                0 & 0 & - i R_1 \, e^{- i \varphi_2} & i R_1 \, e^{i \varphi_2}
            \end{pmatrix}\begin{pmatrix}
                \partial_t A
                \\
                \partial_t \bar A
                \\
                \partial_t B
                \\
                \partial_t \bar B
            \end{pmatrix}.
        \end{align*}
        Therefore,
        \begin{align*}
            \partial_t R_1 &= \Re\left(C_1^{[1,0,1]}\right) \, \eps \, R_1 + \Re\left(C_1^{[3,0]}\right) \, R_1^3 + \Re\left(C_1^{[1,2]}\right) \, R_1 \, R_2^2 + \Re\left(C_1^{[5,0]}\right) \, R_1^5
            \\
            & \quad + \Re\left(C_1^{[1,4]}\right) \, R_1 \, R_2^4 + \Re\left(C_1^{[3,2]}\right) \, R_1^3 \, R_2^2,
        \end{align*}
        and
        \begin{align*}
            \partial \varphi_1 &= \Im\left(C_1^{[1,0,1]}\right) \, \eps + \omega + \Im\left(C_1^{[3,0]}\right) \, R_1^2 + \Im\left(C_1^{[1,2]}\right) \, R_2^2 + \Im\left(C_1^{[5,0]}\right) \, R_1^4
            \\
            & \quad + \Im\left(C_1^{[1,4]}\right) \, R_2^4 + \Im\left(C_1^{[3,2]}\right) \, R_1^2 \, R_2^2,
        \end{align*}
        together with
        \begin{align*}
            \partial_t R_2 &= \Re\left(C_1^{[1,0,1]}\right) \, \eps \, R_2 + \Re\left(C_1^{[3,0]}\right) \, R_2^3 + \Re\left(C_1^{[1,2]}\right) \, R_1^2 \, R_2 + \Re\left(C_1^{[5,0]}\right) \, R_2^5
            \\
            & \quad + \Re\left(C_1^{[1,4]}\right) \, R_1^4 \, R_2 + \Re\left(C_1^{[3,2]}\right) \, R_1^2 \, R_2^3,
        \end{align*}
        and
        \begin{align*}
            \partial \varphi_2 &= - \Im\left(C_1^{[1,0,1]}\right) \, \eps - \omega - \Im\left(C_1^{[3,0]}\right) \, R_2^2 - \Im\left(C_1^{[1,2]}\right) \, R_1^2 - \Im\left(C_1^{[5,0]}\right) \, R_2^4
            \\
            & \quad - \Im\left(C_1^{[1,4]}\right) \, R_1^4 - \Im\left(C_1^{[3,2]}\right) \, R_1^2 \, R_2^2.
        \end{align*}
        Note that the equations for $R_1$ and $R_2$ decouple from the phases $\varphi_1$ and $\varphi_2$, something that also holds for the equations up to third order \cite{Knobloch86}. Therefore, we analyze steady states of the equations for $R_1$ and $R_2$, which can be written as
        \begin{align*}
            \partial_t R_1 &= R_1 \left(\Re\left(C_1^{[1,0,1]}\right) \, \eps + \Re\left(C_1^{[3,0]}\right) \, R_1^2 + \Re\left(C_1^{[1,2]}\right) \, R_2^2 + \Re\left(C_1^{[5,0]}\right) \, R_1^4\right.
            \\
            & \quad \left. + \Re\left(C_1^{[1,4]}\right) \, R_2^4 + \Re\left(C_1^{[3,2]}\right) \, R_1^2 \, R_2^2\right),
            \\
            \partial_t R_2 &= R_2 \left(\Re\left(C_1^{[1,0,1]}\right) \, \eps + \Re\left(C_1^{[3,0]}\right) \, R_2^2 + \Re\left(C_1^{[1,2]}\right) \, R_1^2 + \Re\left(C_1^{[5,0]}\right) \, R_2^4\right.
            \\
            & \quad \left. + \Re\left(C_1^{[1,4]}\right) \, R_1^4 + \Re\left(C_1^{[3,2]}\right) \, R_1^2 \, R_2^2\right).
        \end{align*}
        Here, there are four types of steady states we are interested in: $(0, 0)$, $\left(R_1, 0\right)$, $\left(0, R_2\right)$, $\left(R_1, R_1\right)$. The first one exists for every $\varepsilon \in \mathbb R$, while the second one is a root of the equation 
        \begin{align}
           \Re\left(C_1^{[1,0,1]}\right) \, \eps + \Re\left(C_1^{[3,0]}\right) \, R_1^2 + \Re\left(C_1^{[5,0]}\right) \, R_1^4,
           \label{eq:rootR1}
        \end{align}
        which is a quadratic in $R_1^2$ that has real nontrivial steady-states only if
        \begin{align*}
            \Re\left(C_1^{[3, 0]}\right)^2 - 4 \, \Re\left(C_1^{[1,0,1]}\right) \, \eps \, \Re\left(C_1^{[5,0]}\right) > 0.
        \end{align*}
        This condition defines the leading-order expression for a fold of travelling waves
        \begin{align}
            \varepsilon = \frac{\Re\left(C_1^{[3, 0]}\right)^2}{4 \, \Re\left(C_1^{[1,0,1]}\right) \, \Re\left(C_1^{[5,0]}\right)}.
            \label{eq:standingfold}
        \end{align}
        We require $R_1^2$ to be positive at the fold, which leads to the condition
        \begin{align*}
            - \frac{\Re\left(C_1^{[3, 0]}\right)}{2 \, \Re\left(C_1^{[5, 0]}\right)} > 0,
        \end{align*}
        which explains the side of the bifurcation point at which the fold appears. The analysis for the steady state $\left(0, R_2\right)$ is analogous, so the details are omitted. 
        
        Finally, for standing waves, $\left(R_1, R_1\right)$, we find steady states in that invariant manifold to be defined by
        \begin{align*}
            0 &= R_1 \left(\Re\left(C_1^{[1,0,1]}\right) \, \eps + \left(\Re\left(C_1^{[3,0]}\right) + \Re\left(C_1^{[1,2]}\right)\right) R_1^2\right.
            \\
            & \quad \left. + \left(\Re\left(C_1^{[5,0]}\right) + \Re\left(C_1^{[1,4]}\right) + \Re\left(C_1^{[3,2]}\right)\right) R_1^4\right).
        \end{align*}
        This equation a real nontrivial solution only if
        \begin{align*}
            \left(\Re\left(C_1^{[3,0]}\right) + \Re\left(C_1^{[1,2]}\right)\right)^2 - 4 \, \Re\left(C_1^{[1,0,1]}\right) \, \eps \left(\Re\left(C_1^{[5,0]}\right) + \Re\left(C_1^{[1,4]}\right) + \Re\left(C_1^{[3,2]}\right)\right) > 0.
        \end{align*}
        Looking for double roots, we find the leading-order expression for fold bifurcations of standing waves:
        \begin{align*}
            \eps = \frac{\left(\Re\left(C_1^{[3,0]}\right) + \Re\left(C_1^{[1,2]}\right)\right)^2}{4 \, \Re\left(C_1^{[1,0,1]}\right) \left(\Re\left(C_1^{[5,0]}\right) + \Re\left(C_1^{[1,4]}\right) + \Re\left(C_1^{[3,2]}\right)\right)}.
        \end{align*}

    \subsection{The main result}
        The main result in this paper is to derive closed-form expressions for each of the above coefficients $C_1^{[\cdot]}$ up to order five. To that end, let us define the following multilinear, symmetric, vector functions:
    	\begin{align*}
    		\mbf F_{1, 1} (\mbf a) &= \left(\sum_{1\leq p\leq n} a_p \, \left.\frac{\partial}{\partial \varepsilon}\left(\frac{\partial \, \mbf f}{\partial u_p}\left(\mbf P(\varepsilon), \varepsilon\right)\right)\right|_{\varepsilon = 0}\right),
    		\s
    		\mbf F_2(\mbf a, \mbf b) &= \frac{1}{2!}\left(\sum_{1\leq p, q\leq n} a_p \, b_q \, \frac{\partial^2 \, \mbf f}{\partial u_p \partial u_q}\left(\mbf P(0), 0\right)\right),
    		\s 
    		\mbf F_3(\mbf a,\mbf b, \mbf c) &= \frac{1}{3!}\left(\sum_{1\leq p, q, r\leq n} a_p \, b_q \, c_r \, \frac{\partial^3 \, \mbf f}{\partial u_p \partial u_q \partial u_r}\left(\mbf P(0), 0\right)\right),
    		\s 
    		\mbf F_4(\mbf a,\mbf b,\mbf c,\mbf d) &= \frac{1}{4!} \left(\sum_{1\leq p, q, r, s\leq n} a_p \, b_q \, c_r \, d_s \, \frac{\partial^4 \, \mbf f}{\partial u_p \partial u_q \partial u_r \partial u_s}\left(\mbf P(0), 0\right)\right),
    		\s 
    		\mbf F_5(\mbf a,\mbf b,\mbf c,\mbf d,\mbf e) &= \frac{1}{5!}\left(\sum_{1\leq p, q, r, s, t\leq n} a_p \, b_q \, c_r \, d_s \, e_t \, \frac{\partial^5 \, \mbf f}{\partial u_p \partial u_q \partial u_r \partial u_s \partial u_t}\left(\mbf P(0),0\right)\right),
    	\end{align*}
    	for $\mbf a = \left(a_1, \ldots, a_n\right)^\intercal, \mbf b = \left(b_1, \ldots, b_n\right)^\intercal, \mbf c = \left(c_1, \ldots, c_n\right)^\intercal, \mbf d = \left(d_1, \ldots, d_n\right)^\intercal, \mbf e = \left(e_1, \ldots, e_n\right)^\intercal\in \mathbb C^n$. Also, we define
        {\small
        \begin{multline}
            C_1^{[5,0]} = \frac{1}{\bs \phi_1^{[1]}\cdot \bs \psi_1^{[1]}} \, \bs \psi_1^{[1]}\cdot\left( - \left(2 \, C_1^{[3,0]} + \ol{C_1^{[3,0]}}\right) \mbf W_1^{[3]} + 2 \, \mbf F_2\left(\bs \phi_1^{[1]}, \mbf W_0^{[4]} + \ol{\mbf W_0^{[4]}}\right)\right.
            \s 
            + 2 \, \mbf F_2\left(\ol{\bs \phi_1^{[1]}}, \mbf W_2^{[4]}\right) + 4 \, \mbf F_2\left(\mbf W_0^{[2]}, \mbf W_1^{[3]}\right) + 2 \, \mbf F_2\left(\mbf W_2^{[2]}, \ol{\mbf W_1^{[3]}}\right) + 2 \, \mbf F_2\left(\ol{\mbf W_2^{[2]}}, \mbf W_3^{[3]}\right)
            \s 
            + 12 \, \mbf F_3\left(\bs \phi_1^{[1]}, \mbf W_0^{[2]}, \mbf W_0^{[2]}\right) + 6 \, \mbf F_3\left(\bs \phi_1^{[1]}, \mbf W_2^{[2]}, \ol{\mbf W_2^{[2]}}\right) + 12 \, \mbf F_3\left(\ol{\bs \phi_1^{[1]}}, \mbf W_0^{[2]}, \mbf W_2^{[2]}\right)
            \s 
            + 3 \, \mbf F_3\left(\bs \phi_1^{[1]}, \bs \phi_1^{[1]}, \ol{\mbf W_1^{[3]}}\right) + 6 \, \mbf F_3\left(\bs \phi_1^{[1]}, \ol{\bs \phi_1^{[1]}}, \mbf W_1^{[3]}\right) + 3 \, \mbf F_3\left(\ol{\bs \phi_1^{[1]}}, \ol{\bs \phi_1^{[1]}}, \mbf W_3^{[3]}\right)
            \s 
            + 4 \, \mbf F_4\left(\bs \phi_1^{[1]}, \bs \phi_1^{[1]}, \bs \phi_1^{[1]}, \ol{\mbf W_2^{[2]}}\right) + 24 \, \mbf F_4\left(\bs \phi_1^{[1]}, \bs \phi_1^{[1]}, \ol{\bs \phi_1^{[1]}}, \mbf W_0^{[2]}\right)
            \s 
            \left. + 12 \, \mbf F_4\left(\bs \phi_1^{[1]}, \ol{\bs \phi_1^{[1]}}, \ol{\bs \phi_1^{[1]}}, \mbf W_2^{[2]}\right) + 10 \, \mbf F_5\left(\bs \phi_1^{[1]}, \bs \phi_1^{[1]}, \bs \phi_1^{[1]}, \ol{\bs \phi_1^{[1]}}, \ol{\bs \phi_1^{[1]}}\right)\right), \label{C150}
        \end{multline}
        \begin{multline}
            C_1^{[1, 4]} = \frac{1}{\bs \phi_1^{[1]} \cdot \bs \psi_1^{[1]}} \, \bs \psi_1^{[1]} \cdot \left(- \left(C_1^{[1,2]} + C_1^{[3,0]} + \ol{C_1^{[3,0]}}\right) \mbf W_1^{[2,1]} \vphantom{\frac{1}{4}}\right.
            \s 
            2 \, \mbf F_2\left(\bs \phi_1^{[1]}, \mbf W_0^{[4]} + \ol{\mbf W_0^{[4]}} + \ol{\mbf W_2^{[3,1]}}\right) + 2 \, \mbf F_2\left(\ol{\bs \phi_1^{[1]}}, \mbf W_0^{[3,1]}\right) + 4 \, \mbf F_2\left(\mbf W_0^{[2]}, \mbf W_1^{[2,1]}\right)
            \s 
            + 2 \, \mbf F_2\left(\ol{\mbf W_2^{[2]}}, \mbf W_{1, 2}^{[2, 1]}\right) + 2 \, \mbf F_2\left(\mbf W_2^{[2]}, \ol{\mbf W_3^{[2,1]}}\right) + 2 \, \mbf F_2\left(\mbf W_1^{[3]}, \mbf W_2^{[1,1]}\right) + 2 \, \mbf F_2\left(\ol{\mbf W_1^{[3]}}, \mbf W_0^{[1,1]}\right)
            \s 
            + 12 \, \mbf F_3\left(\bs \phi_1^{[1]}, \mbf W_0^{[2]}, \mbf W_0^{[2]} + \mbf W_2^{[1,1]}\right) + 6 \, \mbf F_3\left(\bs \phi_1^{[1]}, \ol{\mbf W_2^{[2]}}, \mbf W_2^{[2]} + \mbf W_0^{[1,1]}\right)
            \s 
            + 12 \, \mbf F_3\left(\ol{\bs \phi_1^{[1]}}, \mbf W_0^{[2]}, \mbf W_0^{[1, 1]}\right) + 6 \, \mbf F_3\left(\ol{\bs \phi_1^{[1]}}, \mbf W_2^{[2]}, \mbf W_2^{[1, 1]}\right) + 6 \, \mbf F_3\left(\bs \phi_1^{[1]}, \ol{\bs \phi_1^{[1]}}, \mbf W_1^{[3]} + \mbf W_1^{[2, 1]}\right)
            \s 
            + 3 \, \mbf F_3\left(\ol{\bs \phi_1^{[1]}}, \ol{\bs \phi_1^{[1]}}, \mbf W_{1,2}^{[2,1]}\right) + 3 \, \mbf F_3\left(\bs \phi_1^{[1]}, \bs \phi_1^{[1]}, 2 \, \ol{\mbf W_1^{[3]}} + \ol{\mbf W_3^{[2, 1]}}\right)
            \s 
            + 12 \, \mbf F_4\left(\bs \phi_1^{[1]}, \bs \phi_1^{[1]}, \ol{\bs \phi_1^{[1]}}, 4 \, \mbf W_0^{[2]} + \mbf W_2^{[1,1]}\right) + 12 \, \mbf F_4\left(\bs \phi_1^{[1]}, \bs \phi_1^{[1]}, \bs \phi_1^{[1]}, \ol{\mbf W_2^{[2]}}\right)
            \s
            \left. + 12 \, \mbf F_4\left(\bs \phi_1^{[1]}, \ol{\bs \phi_1^{[1]}}, \ol{\bs \phi_1^{[1]}}, \mbf W_2^{[2]} + \mbf W_0^{[1,1]}\right) + 30 \, \mbf F_5\left(\bs \phi_1^{[1]}, \bs \phi_1^{[1]}, \bs \phi_1^{[1]}, \ol{\bs \phi_1^{[1]}}, \ol{\bs \phi_1^{[1]}}\right)\right), \label{C114}
        \end{multline}
        }
        and
        {\small
        \begin{multline}
            C_1^{[3, 2]} = \frac{1}{\bs \phi_1^{[1]}\cdot \bs \psi_1^{[1]}} \, \bs \psi_1^{[1]} \cdot \left(- \left(C_1^{[3,0]} + C_1^{[1,2]}+\ol{C_1^{[1,2]}}\right) \mbf W_1^{[2,1]} - \left(2 \, C_1^{[1,2]} + \ol{C_1^{[1,2]}}\right) \mbf W_1^{[3]} \right.
            \s
            + 2 \, \mbf F_2\left(\bs \phi_1^{[1]}, \mbf W_0^{[2,2]} + \ol{\mbf W_0^{[2,2]}} + \mbf W_2^{[3,1]}\right) + 2 \, \mbf F_2\left(\ol{\bs \phi_1^{[1]}}, \mbf W_2^{[2, 2]} + \mbf W_0^{[3,1]}\right) + 4 \, \mbf F_2\left(\mbf W_0^{[2]}, \mbf W_1^{[2,1]}\right)
            \s 
            + 2 \, \mbf F_2\left(\ol{\mbf W_1^{[2,1]}}, \mbf W_2^{[2]} + \mbf W_0^{[1, 1]}\right) + 2 \, \mbf F_2\left(\mbf W_2^{[1,1]}, \mbf W_1^{[2,1]} + \mbf W_3^{[2, 1]}\right) + 2 \, \mbf F_2\left(\ol{\mbf W_0^{[1, 1]}}, \mbf W_{1,2}^{[2,1]}\right)
            \s
            + 4 \, \mbf F_2\left(\mbf W_0^{[2]}, \mbf W_1^{[3]}\right) + 6 \, \mbf F_3\left(\bs \phi_1^{[1]}, \mbf W_2^{[1, 1]}, \mbf W_2^{[1, 1]}\right) + 6 \, \mbf F_3\left(\bs \phi_1^{[1]}, \ol{\mbf W_0^{[1, 1]}}, \mbf W_2^{[2]} + \mbf W_0^{[1,1]}\right)
            \s 
            + 6 \, \mbf F_3\left(\ol{\bs \phi_1^{[1]}}, 2 \, \mbf W_0^{[2]} + \mbf W_2^{[1,1]}, \mbf W_2^{[2]} + \mbf W_0^{[1,1]}\right) + 12 \, \mbf F_3\left(\bs \phi_1^{[1]}, \mbf W_0^{[2]}, 2 \, \mbf W_0^{[2]} + \mbf W_2^{[1,1]}\right)
            \s 
            + 9 \, \mbf F_3\left(\bs \phi_1^{[1]}, \bs \phi_1^{[1]}, \ol{\mbf W_1^{[2,1]}}\right) + 6 \, \mbf F_3\left(\bs \phi_1^{[1]}, \ol{\bs \phi_1^{[1]}}, \mbf W_1^{[3]} + 2 \, \mbf W_1^{[2,1]} + \mbf W_3^{[2,1]}\right)
            \s 
            + 6 \, \mbf F_3\left(\ol{\bs \phi_1^{[1]}}, \ol{\bs \phi_1^{[1]}}, \mbf W_{1, 2}^{[2,1]}\right) + 36 \, \mbf F_4\left(\bs \phi_1^{[1]}, \bs \phi_1^{[1]}, \ol{\bs \phi_1^{[1]}}, 2 \, \mbf W_0^{[2]} + \mbf W_2^{[1,1]}\right)
            \s 
            + 12 \, \mbf F_4\left(\bs \phi_1^{[1]}, \bs \phi_1^{[1]}, \bs \phi_1^{[1]}, \ol{\mbf W_0^{[1,1]}}\right) + 24 \, \mbf F_4\left(\bs \phi_1^{[1]}, \ol{\bs \phi_1^{[1]}}, \ol{\bs \phi_1^{[1]}}, \mbf W_2^{[2]} + \mbf W_0^{[1, 1]}\right)
            \s
            \left. + 60 \, \mbf F_5\left(\bs \phi_1^{[1]}, \bs \phi_1^{[1]}, \bs \phi_1^{[1]}, \ol{\bs \phi_1^{[1]}}, \ol{\bs \phi_1^{[1]}}\right)\right), \label{C132}
        \end{multline}
        }
        where
        \begin{align*}
            \left(\jac \mbf f\left(\mbf P,0\right) - \hat k^2 \, \mathbb D\left(0\right) - i \hat \omega \, I\right) \mbf W_1^{[3]} = C_1^{[3,0]} \, \bs \phi_1^{[1]} - 4 \, \mbf F_2\left(\bs \phi_1^{[1]}, \mbf W_0^{[2]}\right) - 2 \, \mbf F_2\left(\overline{\bs \phi_1^{[1]}}, \mbf W_2^{[2]}\right)
            \\
            - 3 \, \mbf F_3\left(\bs \phi_1^{[1]}, \bs \phi_1^{[1]}, \ol{\bs \phi_1^{[1]}}\right),
        \end{align*}
        \begin{align*}
            \left(\jac \mbf f\left(\mbf P, 0\right) - 9 \, \hat k^2 \, \mathds D\left(0\right) - 3 \, i \hat \omega \, I\right)\mbf W_3^{[3]} = - 2 \, \mbf F_2\left(\bs \phi_1^{[1]}, \mbf W_2^{[2]}\right) - \mbf F_3\left(\bs \phi_1^{[1]}, \bs \phi_1^{[1]}, \bs \phi_1^{[1]}\right),
        \end{align*}
        \begin{multline*}
            \left(\jac \mbf f\left(\mbf P, 0\right) - \hat k^2 \, \mathds D\left(0\right) - i \hat \omega \, I\right) \mbf W_1^{[2,1]} = C_1^{[1,2]} \, \bs \phi_1^{[1]} - 2 \, \mbf F_2\left(\ol{\bs \phi_1^{[1]}}, \mbf W_0^{[1,1]}\right)
            \s 
            - 2 \, \mbf F_2\left(\bs \phi_1^{[1]},2 \, \mbf W_0^{[2]} + \mbf W_2^{[1,1]}\right) - 6 \, \mbf F_3\left(\bs \phi_1^{[1]}, \bs \phi_1^{[1]}, \ol{\bs \phi_1^{[1]}}\right),
        \end{multline*}
        \begin{multline*}
            \left(\jac \mbf f\left(\mbf P, 0\right) - \hat k^2 \, \mathds D\left(0\right) - 3 \, i \hat \omega \, I\right)\mbf W_{1, 2}^{[2, 1]} = - 2 \, \mbf F_2\left(\bs \phi_1^{[1]}, \mbf W_2^{[2]} + \mbf W_0^{[1, 1]}\right)
            \s
            - 3 \, \mbf F_3\left(\bs \phi_1^{[1]}, \bs \phi_1^{[1]}, \bs \phi_1^{[1]}\right),
        \end{multline*}
        \begin{multline*}
            \left(\jac \mbf f\left(\mbf P, 0\right) - 9 \, \hat k^2 \, \mathds D\left(0\right) - i \hat \omega \, I\right)\mbf W_3^{[2,1]} = - 2 \, \mbf F_2\left(\bs \phi_1^{[1]}, \mbf W_2^{[1,1]}\right) - 2 \, \mbf F_2\left(\ol{\bs \phi_1^{[1]}}, \mbf W_2^{[2]}\right)
            \\
            - 3 \, \mbf F_3\left(\bs \phi_1^{[1]}, \bs \phi_1^{[1]}, \ol{\bs \phi_1^{[1]}}\right),
        \end{multline*}
        \begin{multline*}
            \jac \mbf f(\mbf 0) \, \mbf W_0^{[4]} = 2 \, C_1^{[3,0]} \, \mbf W_0^{[2]} - 2 \, \mbf F_2\left(\mbf W_0^{[2]}, \mbf W_0^{[2]}\right) - \mbf F_2\left(\mbf W_2^{[2]}, \ol{\mbf W_2^{[2]}}\right) - 2 \, \mbf F_2\left(\bs \phi_1^{[1]}, \ol{\mbf W_1^{[3]}}\right)
            \s 
            - 3 \, \mbf F_3\left(\bs \phi_1^{[1]}, \bs \phi_1^{[1]}, \ol{\mbf W_2^{[2]}}\right) - 6 \, \mbf F_3\left(\bs \phi_1^{[1]}, \ol{\bs \phi_1^{[1]}}, \mbf W_0^{[2]}\right) - 3 \, \mbf F_4\left(\bs \phi_1^{[1]}, \bs \phi_1^{[1]}, \ol{\bs \phi_1^{[1]}}, \ol{\bs \phi_1^{[1]}}\right)
        \end{multline*}
        \begin{multline*}
            \left(\jac \mbf f\left(\mbf P, 0\right) - 4 \, \hat k^2 \, \mathds D\left(0\right) - 2 \, i \hat \omega I\right) \mbf W_2^{[4]} = 2 \, C_1^{[3,0]} \, \mbf W_2^{[2]} - 4 \, \mbf F_2\left(\mbf W_0^{[2]}, \mbf W_2^{[2]}\right)
            \s 
            - 2 \, \mbf F_2\left(\bs \phi_1^{[1]}, \mbf W_1^{[3]}\right) - 2 \, \mbf F_2\left(\ol{\bs \phi_1^{[1]}}, \mbf W_3^{[3]}\right) - 6 \, \mbf F_3\left(\bs \phi_1^{[1]}, \bs \phi_1^{[1]}, \mbf W_0^{[2]}\right)
            \s
            - 6 \, \mbf F_3\left(\bs \phi_1^{[1]}, \ol{\bs \phi_1^{[1]}}, \mbf W_2^{[2]}\right) - 4 \, \mbf F_4\left(\bs \phi_1^{[1]}, \bs \phi_1^{[1]}, \bs \phi_1^{[1]}, \ol{\bs \phi_1^{[1]}}\right),
        \end{multline*}
        \begin{multline*}
            \left(\jac \mbf f\left(\mbf P, 0\right) - 2 \, i \hat \omega I\right) \mbf W_0^{[3, 1]} = \left(C_1^{[3,0]} + C_1^{[1,2]}\right) \mbf W_0^{[1,1]} - 4 \, \mbf F_2\left(\mbf W_0^{[2]}, \mbf W_0^{[1,1]}\right)
            \s 
            - 2 \, \mbf F_2\left(\mbf W_2^{[2]}, \mbf W_2^{[1,1]}\right) - 2 \, \mbf F_2\left(\bs \phi_1^{[1]}, \mbf W_1^{[3]} + \mbf W_1^{[2, 1]}\right) - 2 \, \mbf F_2\left(\ol{\bs \phi_1^{[1]}}, \mbf W_{1,2}^{[2,1]}\right)
            \s 
            - 3 \, \mbf F_3\left(\bs \phi_1^{[1]}, \bs \phi_1^{[1]}, 4 \, \mbf W_0^{[2]} + \mbf W_2^{[1,1]}\right) - 6 \, \mbf F_3\left(\bs \phi_1^{[1]}, \ol{\bs \phi_1^{[1]}}, \mbf W_2^{[2]} + \mbf W_0^{[1,1]}\right)
            \s
            - 12 \, \mbf F_4\left(\bs \phi_1^{[1]}, \bs \phi_1^{[1]}, \bs \phi_1^{[1]}, \ol{\bs \phi_1^{[1]}}\right),
        \end{multline*}
        \begin{multline*}
            \left(\jac \mbf f\left(\mbf P, 0\right) - 4 \, \hat k^2 \, \mathds D\left(0\right)\right) \mbf W_2^{[3,1]} = \left(C_1^{[3,0]} + \ol{C_1^{[1,2]}}\right) \mbf W_2^{[1,1]} - 4 \, \mbf F_2\left(\mbf W_0^{[2]}, \mbf W_2^{[1,1]}\right)
            \s 
            - 2 \, \mbf F_2\left(\mbf W_2^{[2]}, \ol{\mbf W_0^{[1,1]}}\right) - 2 \, \mbf F_2\left(\bs \phi_1^{[1]}, \ol{\mbf W_1^{[2,1]}}\right) - 2 \, \mbf F_2\left(\ol{\bs \phi_1^{[1]}}, \mbf W_1^{[3]} + \mbf W_3^{[2,1]}\right)
            \s 
            - 3 \, \mbf F_3\left(\bs \phi_1^{[1]}, \bs \phi_1^{[1]}, \ol{\mbf W_0^{[1,1]}}\right) - 6 \, \mbf F_3\left(\bs \phi_1^{[1]}, \ol{\bs \phi_1^{[1]}}, 2 \, \mbf W_0^{[2]} + \mbf W_2^{[1,1]}\right) - 6 \, \mbf F_3\left(\ol{\bs \phi_1^{[1]}}, \ol{\bs \phi_1^{[1]}}, \mbf W_2^{[2]}\right)
            \s
            - 12 \, \mbf F_4\left(\bs \phi_1^{[1]}, \bs \phi_1^{[1]}, \ol{\bs \phi_1^{[1]}}, \ol{\bs \phi_1^{[1]}}\right).
        \end{multline*}
        \begin{multline*}
            \jac \mbf f(\mbf 0) \, \mbf W_0^{[2,2]} = 4 \, C_1^{[1,2]} \, \mbf W_0^{[2]} - 4 \, \mbf F_2\left(\mbf W_0^{[2]}, \mbf W_0^{[2]}\right) - \mbf F_2\left(\mbf W_2^{[1,1]}, \mbf W_2^{[1,1]}\right)
            \\
            - \mbf F_2\left(\mbf W_0^{[1,1]}, \ol{\mbf W_0^{[1,1]}}\right) - 4 \, \mbf F_2\left(\bs \phi_1^{[1]}, \ol{\mbf W_1^{[2,1]}}\right) - 6 \, \mbf F_3\left(\bs \phi_1^{[1]}, \ol{\bs \phi_1^{[1]}}, 2 \, \mbf W_0^{[2]} + \mbf W_2^{[1,1]}\right)
            \\
            - 6 \, \mbf F_3\left(\bs \phi_1^{[1]}, \bs \phi_1^{[1]}, \ol{\mbf W_0^{[1,1]}}\right) - 12 \, \mbf F_4\left(\bs \phi_1^{[1]}, \bs \phi_1^{[1]}, \ol{\bs \phi_1^{[1]}}, \ol{\bs \phi_1^{[1]}}\right),
        \end{multline*}
        and
        \begin{multline*}
            \left(\jac \mbf f\left(\mbf P, 0\right) - 4 \, \hat k^2 \, \mathds D\left(0\right) - 2 \, i \hat \omega I\right) \mbf W_2^{[2,2]} = 2 \, C_1^{[1,2]} \, \mbf W_2^{[2]} - 2 \,  \mbf F_2\left(\mbf W_0^{[1,1]}, \mbf W_2^{[1,1]}\right)
            \s 
            - 4 \, \mbf F_2\left(\mbf W_0^{[2]}, \mbf W_2^{[2]}\right) - 2 \, \mbf F_2\left(\bs \phi_1^{[1]}, \mbf W_1^{[2,1]} + \mbf W_3^{[2,1]}\right) - 2 \, \mbf F_2\left(\ol{\bs \phi_1^{[1]}}, \mbf W_{1,2}^{[2,1]}\right)
            \\
            - 6 \, \mbf F_3\left(\bs \phi_1^{[1]}, \bs \phi_1^{[1]}, \mbf W_0^{[2]} + \mbf W_2^{[1,1]}\right) - 6 \, \mbf F_3\left(\bs \phi_1^{[1]}, \ol{\bs \phi_1^{[1]}}, \mbf W_2^{[2]} + \mbf W_0^{[1,1]}\right)
            \\
            - 12 \, \mbf F_4\left(\bs \phi_1^{[1]}, \bs \phi_1^{[1]}, \bs \phi_1^{[1]}, \ol{\bs \phi_1^{[1]}}\right).
        \end{multline*}
    	
        The main purpose of this article is to prove the following result for the class of systems \eqref{eq:geneq}.

    	\newtheorem{theorem}{Theorem}
    	\begin{theorem} \label{th:main} 
            Let $\mbf P = \mbf P(\varepsilon)$ be a stable homogeneous steady state of \eqref{eq:geneq}. Suppose that at the parameter value $\eps = 0$  there exists a \textit{critical wavenumber} $\hat k>0$, a \textit{critical frequency}, $\hat \omega > 0$, and a unique eigenvalue $\lambda(k; \eps)$ of $\jac \mbf f(\mbf P(\eps), \eps) - k^2 \, \mathbb D(\eps)$ such that $\lambda\left(k; 0\right) = i\hat \omega$, $C_1^{[1, 0, 1]} \neq 0$,
    		\begin{align*}
                \dim \left(\ker \left(\jac \mbf f \left(\mbf P,0\right) - \hat k^2 \, \mathbb D\left(0\right) \pm i \hat \omega \, I\right)\right) &= 1,
                \\
                \dim \left(\ker \left(\jac \mbf f \left(\mbf P,0 \right)- k^2 \, \mathbb D\left(0\right) \pm i \omega \, I\right)\right) = 0 \text{ for all } &\left(k, \omega\right)\neq \left(\hat k,\hat \omega\right),
    			\s 
    			\ker \left(\jac \mbf f \left(\mbf P, 0\right) - \hat k^2 \, \mathbb D\left(0\right) \pm i \hat \omega \, I\right) \, \cap \, \Im \left(\jac \mbf f \left(\mbf P,0\right)- \hat k^2 \, \mathbb D\left(0\right)\pm i \hat \omega \, I\right)&=\{\mbf 0\},
    		\end{align*}
            where
            \begin{align*}
                C_1^{[1,0,1]} &= \frac{1}{\bs \psi_1^{[1]} \cdot \bs \phi_1^{[1]}} \, \bs \psi_1^{[1]} \cdot \left(\mbf F_{1,1}\left(\bs \phi_1^{[1]}\right) - k^2 \, \frac{\dd \mathbb D}{\dd \varepsilon}(0) \, \bs \phi_1^{[1]}\right)
                \\
                \bs \phi_1^{[1]} &\in \ker(\jac \mbf f \left(\mbf P,0\right)-\hat k^2 \, \mathbb D\left(0\right)-i \hat \omega  \, I)\setminus \{\mbf 0\},
    			\s 
    			\bs \psi_1^{[1]}&\in \ker(\jac \mbf f \left(\mbf P,0\right)^\intercal-\hat k^2 \, \mathbb D\left(0\right)^\intercal-i \hat \omega \, I)\setminus \{\mbf 0\}.
            \end{align*}
    		Then the critical cubic coefficients of the normal form \eqref{eq:nf1} and \eqref{eq:nf2} are given by
            {\small
    		\begin{align*}
    			C_1^{[3,0]}=\frac{1}{\bs \psi_1^{[1]} \cdot \bs \phi_1^{[1]}} \, \bs \psi_1^{[1]} \cdot \left(4 \, \mbf F_2\left(\bs \phi_1^{[1]}, \mbf W_0^{[2]}\right) + 2 \, \mbf F_2 \left(\ol{\bs \phi_1^{[1]}},\mbf W_2^{[2]}\right) + 3 \, \mbf F_3\left(\bs \phi_1^{[1]}, \bs \phi_1^{[1]}, \ol{\bs \phi_1^{[1]}}\right)\right),
    		\end{align*}
    		\begin{align*}
    			C_1^{[1,2]} = \frac{1}{\bs \psi_1^{[1]} \cdot \bs \phi_1^{[1]}} \, \bs \psi_1^{[1]} \cdot \left(2 \, \mbf F_2\left(\ol{\bs \phi_1^{[1]}}, \mbf W_0^{[1,1]}\right) + 2 \, \mbf F_2\left(\bs \phi_1^{[1]}, 2 \, \mbf W_0^{[2]} + \mbf W_2^{[1,1]}\right) + 6 \, \mbf F_3\left(\bs \phi_1^{[1]}, \bs \phi_1^{[1]}, \ol{\bs \phi_1^{[1]}}\right)\right),
    		\end{align*}
            }
    	  where
    		\begin{align*}
    			\jac \mbf f\left(\mbf P, 0\right) \, \mbf W_0^{[2]} &= - \mbf F_2\left(\bs \phi_1^{[1]},\overline{\bs \phi_1^{[1]}}\right),
    			\s 
    			\left(\jac \mbf f\left(\mbf P, 0\right) - 4 \, \hat k^2 \, \mathbb D\left(0\right) - 2 \, i \hat \omega \, I\right) \mbf W_2^{[2]} &= - \mbf F_2\left(\bs \phi_1^{[1]},\bs \phi_1^{[1]}\right),
    			\s
    			\left(\jac \mbf f\left(\mbf P, 0\right)-2 \, i \hat \omega \, I\right) \mbf W_0^{[1, 1]} &= - 2 \, \mbf F_2\left(\bs \phi_1^{[1]},\bs \phi_1^{[1]}\right),
    			\s 
    			\left(\jac \mbf f\left(\mbf P, 0\right) - 4 \, \hat k^2 \, \mathbb D\left(0\right)\right) \mbf W_2^{[1, 1]} &= - 2 \, \mbf F_2\left(\bs \phi_1^{[1]},\ol{\bs \phi_1^{[1]}}\right), 
    		\end{align*}
            Moreover, depending on the values of the coefficients, the bifurcation of small amplitude standing or travelling waves around $\mbf P$ are then as given in Fig.~\ref{fig:knobloch2par}.  
    		
            Furthermore, if $\Re\left(C_1^{[3, 0]}\right) = 0$ $\left(\text{resp. } \Re\left(C_1^{[3,0]} + C_1^{[1,2]}\right) = 0\right)$, then $\mbf P$ undergoes a  codimension-two wave-Bautin bifurcation, which is generic if and only if $\Re\left(C_1^{[5, 0]}\right) \neq 0$ $\left(\text{resp. } \Re\left(C_1^{[5, 0]} + C_1^{[1, 4]} + C_1^{[3, 2]}\right)\neq 0\right)$, where $C_1^{[5, 0]}, C_1^{[1, 4]}$ and $C_1^{[3, 2]}$ are given by \eqref{C150}--\eqref{C132}.
    	\end{theorem}

    \newtheorem{remark}{Remark}
    \begin{remark}
        Although the theory developed here is general, it has the limitation that we only keep track of bifurcating waves at the bifurcation point whilst ignoring that when one has a positive dispersion relation, one does indeed have interaction between different wave modes. The full analysis of the partial differential equation can be done by defining a small variable $\delta > 0$ and defining a long spatial variable $X = \delta \, x$, a slow time scale $T = \delta^2 \, t$, and expanding the wave frequency $\omega = \omega^* + \delta \, \omega_1 + \delta^2 \, \omega_2 + \ldots$ \cite{Knobloch86,TOMS}. Furthermore, assuming that $A = A(X, T)$ and $B = B(X, T)$, a similar asymptotic analysis can be carried out. This will add a few more terms to the amplitude equations (such as $\partial_{XX} A$ and $\partial_{XX} B$) but it will not change the values of $C_1^{[3, 0]}$ and $C_1^{[1, 2]}$ given above. Calculation of the corresponding coefficients is beyond the scope of this paper, as their relative order depends on the scale (in terms of $\delta$) at which parameters affect the linear coefficients of the system. Depending on this scaling we will find different unfoldings of codimension-two points, affecting the temporal stability of waves of different wavenumbers. A full analysis is left for future work.
    \end{remark}
    
\section{Computation of normal form coefficients} \label{sec:coefs}
    The proof of Theorem \ref{th:main} will be based on the well-known amplitude equations that let us study the system in a sufficiently small neighbourhood of the bifurcation curve \cite{tirapegui}.

    Without loss of generality, we suppose that $\mbf P = \mbf 0$ for all $\varepsilon \geq 0$. If this was not the case, we can always perform the following change of variables:
	\begin{align*}
		\left(u_1, \ldots, u_n\right)^\intercal\mapsto \left(u_1, \ldots, u_n\right)^\intercal - \mbf P(\varepsilon),
	\end{align*}
    which translates $\mbf P(\varepsilon)$ to the origin. Also, in order to simplify notation, we drop the $\varepsilon$ dependence, assuming that everything is evaluated at $\varepsilon = 0$ unless otherwise stated. We also write $k$ and $\omega$ instead of $\hat k$ and $\hat \omega$, provided the meaning is clear.
	
    Next, with our hypotheses, we can expand each function $f_1, \ldots, f_n$ using a Taylor series with respect to $\left(u_1, \ldots, u_n\right)^\intercal$ as follows:
	\begin{align*}
		f_j(\mbf u) = \sum_{\ell = 1}^r \frac{1}{\ell!} \left(\sum_{m = 1}^n u_m \frac{\partial}{\partial u_m}\right)^\ell f_j(\mbf 0) + \mathcal O\left(\left \lVert \mbf u\right \rVert^{r + 1}\right),
	\end{align*}
	where $\mbf u = \left(u_1, \ldots, u_n\right)^\intercal$, and $\mbf 0$ is the zero vector with $n$ components.

    At the bifurcation point where $\varepsilon = 0$, we expand
    \begin{align}
		\mbf u = \sum_{\substack{j, \ell = 0 \\ j + \ell > 0}}^5 \mbf u^{[j, \ell]}(A, B) \label{u-expansion}
	\end{align}
    and seek the equations for $A$ and $B$ to make the residual small. We consider a Taylor expansion for them in the following form:
    \begin{align}
		\partial_t A = \sum_{\substack{j, \ell = 0 \\ j + \ell > 0}}^5 h_1^{[j, \ell]}(A, B), & \qquad \partial_t B = \sum_{\substack{j, \ell = 0 \\ j + \ell > 0}}^5 h_2^{[j, \ell]}(A, B), \label{geneqAB}
	\end{align}
    where $h_1^{[j, \ell]}(A, B)$ and $h_2^{[j, \ell]}(A, B)$ are homogeneous polynomials of degree $j$ in $A$ and $\ell$ in $B$.
    \\
    The unfolding of the bifurcation when $\varepsilon \neq 0$ will be carried out in Sub-section \ref{sub:unfolding}.
    
    Now, using expansion \eqref{u-expansion}, we start our analysis going order by order. We highlight that we are often going to omit the dependence of each function in order to avoid writing cumbersome expressions.
	
    \subsection{First-order expansion}
        At first order, \eqref{eq:geneq} becomes
		\begin{multline}
			\partial_A \mbf u^{[1, 0]} \left(h_1^{[1,0]} + h_1^{[0,1]}\right) + \partial_B \mbf u^{[0, 1]} \left(h_2^{[1,0]} + h_2^{[0,1]}\right) + \mbox{conj.} =
            \\
            \left(\jac \mbf f(\mbf 0) +\mathds D \, \partial_{xx}\right)\left(\mbf u^{[1, 0]} + \mbf u^{[0, 1]}\right), \label{firstorder}
		\end{multline}
		where $\mbox{conj.}$ stands for the conjugate version of the term to the left, which is not necessarily its complex conjugate. For example, $\partial_A \mbf u^{[1, 0]} \, h_1^{[1, 0]} + \mbox{conj.} = \partial_A \mbf u^{[1, 0]} \, h_1^{[1, 0]} + \partial_{\bar A} \mbf u^{[1, 0]} \, \overline{h_1^{[1, 0]}}$.
		
		Next, if we set $h_1^{[1, 0]} = i\omega A$, $h_2^{[0, 1]} = -i\omega B$, and $h_1^{[0, 1]} = h_2^{[1, 0]} = 0$, then \eqref{firstorder} becomes
		\begin{align*}
			i\omega A \, \partial_A \mbf u^{[1,0]} - i\omega B \, \partial_B \mbf u^{[0,1]} + \mbox{conj.} = \left(\jac \mbf f(\mbf 0) + \mathds D \, \partial_{xx}\right)\left(\mbf u^{[1,0]} + \mbf u^{[0,1]}\right).
		\end{align*}
		Therefore, if we consider our ansatz and set
		\begin{align*}
			\mbf u^{[1, 0]} = A \, e^{ikx} \, \bs \phi_1^{[1]} + c.c,  \qquad 
			\mbf u^{[0, 1]} = B \, e^{ikx} \, \ol{\bs \phi_1^{[1]}} + c.c.,
		\end{align*}
		then the equation for $A$ is simplified to
		\begin{align*}
			i\omega A \, e^{ikx} \, \bs \phi_1^{[1]} + c.c. = \left(\jac \mbf f(\mbf 0) -k^2 \, \mathds D\right)\left(A \, e^{ikx} \, \bs \phi_1^{[1]}\right) + c.c.,
		\end{align*}
		which implies
		\begin{align*}
			\left(\jac \mbf f(\mbf 0) -k^2 \, \mathds D - i \omega \, I\right)\bs \phi_1^{[1]}&=\mbf 0.
		\end{align*}
        As the kernel of this linear system is one-dimensional, then $\bs \phi_1^{[1]}$ is well-defined and can be chosen to be non-zero. In particular, it can be chosen to be a unit vector, although it is not strictly necessary. We remark that solving a similar equation for $B$ is not necessary since it is only the complex-conjugate version of the equation for $A$, and the definition of $\bs \phi_1^{[1]}$ is consistent with that. Now we can proceed with the expansion at the next order.

    \subsection{Second order} At second order, \eqref{eq:geneq} is given by
		\begin{multline}
			\partial_A \mbf u^{[1,0]}\left(h_1^{[2,0]} + h_1^{[1,1]} + h_1^{[0,2]}\right) + \partial_B \mbf u^{[0,1]} \left(h_2^{[2,0]} + h_2^{[1,1]} + h_2^{[0,2]}\right) 
			\s 
			+ \partial_A \mbf u^{[2,0]} \, h_1^{[1,0]} + \partial_B \mbf u^{[0,2]} \, h_2^{[0,1]} + \partial_A \mbf u^{[1,1]} \, h_1^{[1,0]} + \partial_B \mbf u^{[1,1]} \, h_2^{[0,1]} + \mbox{conj.}
			\s 
			= \left(\jac \mbf f(\mbf 0) +\mathds D \, \partial_{xx} \right) \left(\mbf u^{[2,0]} + \mbf u^{[1,1]} + \mbf u^{[0,2]}\right)
			\s 
			+ \mbf F_2\left(\mbf u^{[1,0]} + \mbf u^{[0,1]}, \mbf u^{[1,0]} + \mbf u^{[0,1]}\right), \label{secondordereq}
		\end{multline}
		which is equivalent to:
		\begin{multline*}
			\partial_A \mbf u^{[1,0]}\left(h_1^{[2,0]} + h_1^{[1,1]} + h_1^{[0,2]}\right) + \partial_B \mbf u^{[0,1]} \left(h_2^{[2,0]} + h_2^{[1,1]} + h_2^{[0,2]}\right)
			\s 
			+ i\omega A \, \partial_A \mbf u^{[2,0]} - i \omega B \, \partial_B \mbf u^{[0,2]} + i\omega A \, \partial_A \mbf u^{[1,1]} - i\omega B \, \partial_B \mbf u^{[1,1]} + \mbox{conj.}
			\s 
			= \left(\jac \mbf f(\mbf 0) +\mathds D \, \partial_{xx} \right) \left(\mbf u^{[2,0]} + \mbf u^{[1,1]} + \mbf u^{[0,2]}\right)
			\s 
			+ \mbf F_2\left(\mbf u^{[1,0]}, \mbf u^{[1,0]}\right) + 2 \, \mbf F_2\left(\mbf u^{[1,0]}, \mbf u^{[0,1]}\right) + \mbf F_2\left(\mbf u^{[0,1]}, \mbf u^{[0,1]}\right).
		\end{multline*}
		Next, as $\partial_A \mbf u^{[1,0]}$, $\partial_{\bar A} \mbf u^{[1,0]}$, $\partial_B \mbf u^{[0,1]}$, $\partial_{\bar B} \mbf u^{[0,1]}$ are in $\ker\left(\jac \mbf f(\mbf 0)-k^2 \, \mathds D \pm i\omega \, I\right)$, and $\ker\left(\jac \mbf f(\mbf 0) - k^2 \, \mathds D \pm i\omega \, I\right)$ has a trivial intersection with $\Im\left(\jac \mbf f(\mbf 0) - k^2 \, \mathds D \pm i\omega \, I\right)$, then there is no  solution to the above equation if $h_p^{[j, \ell]}\neq 0$ for and $j + \ell = 2$, and $p = 1, 2$. This implies that any such functions must be 0. Therefore, when splitting \eqref{secondordereq} according to the amplitudes, it becomes
		\begin{multline}
			i\omega A \, \partial_A \mbf u^{[2,0]} + \mbox{conj.} = \left(\jac \mbf f(\mbf 0) + \mathds D \, \partial_{xx} \right) \, \mbf u^{[2,0]} + |A|^2 \, \mbf F_2\left(\bs \phi_1^{[1]}, \ol{\bs \phi_1^{[1]}}\right)
            \\
            + A^2 \, e^{2ikx} \, \mbf F_2\left(\bs \phi_1^{[1]}, \bs \phi_1^{[1]}\right) + c.c., \label{firstequationsecondorder}
		\end{multline}
		and
		\begin{multline}
			i\omega A \, \partial_A \mbf u^{[1,1]} - i\omega B \, \partial_B \mbf u^{[1,1]} + \mbox{conj.} = \left(\jac \mbf f(\mbf 0) +\mathds D \, \partial_{xx} \right) \mbf u^{[1,1]} + 2 \, A \bar B \, \mbf F_2\left(\bs \phi_1^{[1]}, \bs \phi_1^{[1]}\right)
			\s 
			+ 2 \, A B \, e^{2ikx} \, \mbf F_2\left(\bs \phi_1^{[1]}, \ol{\bs \phi_1^{[1]}}\right) + c.c., \label{secondequationsecondorder}
		\end{multline}
		while the remaining equation is the complex-conjugate version of \eqref{firstequationsecondorder}, and we seek solutions to these three equations that are orthogonal to $\ker\left(\jac \mbf f(\mbf 0) + \mathbb D \, \partial_{xx} \pm i \omega I\right)$. The form of these solutions is given by:
		\begin{align*}
			\mbf u^{[2, 0]} &= |A|^2 \, \mbf W_0^{[2]} + A^2 \, e^{2ikx} \, \mbf W_2^{[2]} + c.c.,
			\s 
			\mbf u^{[1, 1]} &= A \bar B \, \mbf W_0^{[1,1]} + AB \, e^{2ikx} \, \mbf W_2^{[1,1]} + c.c.,
			\s 
			\mbf u^{[0, 2]} &= |B|^2 \, \mbf W_0^{[2]} + B^2 \, e^{2ikx} \, \ol{\mbf W_2^{[2]}} +c.c.
		\end{align*}
		With this, because \eqref{firstequationsecondorder} and \eqref{secondequationsecondorder} are linear equations, we can make these substitutions for $\mbf u^{[2,0]}$, and $\mbf u^{[1,1]}$ to determine the unknown vectors of coefficients $\mbf W_\alpha^{[\beta]}$. Thus, we obtain
		\begin{align*}
			\jac \mbf f(\mbf 0) \, \mbf W_0^{[2]}&= - \mbf F_2\left(\bs \phi_1^{[1]}, \overline{\bs \phi_1^{[1]}}\right),
			\s 
			\left(\jac \mbf f(\mbf 0) - 4 \, k^2 \, \mathbb D - 2 \, i\omega \, I\right) \mbf W_2^{[2]} &= - \mbf F_2\left(\bs \phi_1^{[1]}, \bs \phi_1^{[1]}\right),
			\s
			\left(\jac \mbf f(\mbf 0) - 2 \, i\omega \, I\right) \mbf W_0^{[1,1]} &= - 2 \, \mbf F_2\left(\bs \phi_1^{[1]},\bs \phi_1^{[1]}\right),
			\s 
			\left(\jac \mbf f(\mbf 0) - 4 \, k^2 \, \mathbb D\right) \mbf W_2^{[1,1]} &= - 2 \, \mbf F_2\left(\bs \phi_1^{[1]}, \ol{\bs \phi_1^{[1]}}\right).
		\end{align*}
		These expressions imply that $\mbf W_0^{[2]}, \mbf W_2^{[1, 1]}\in \mathbb R^n$.
		
		Furthermore, this lets us continue to the next order.

    \subsection{Third order}
        At third order, \eqref{eq:geneq} is given by:
		\begin{multline}
			\partial_A \mbf u^{[1,0]} \left(h_1^{[3,0]} + h_1^{[2,1]} + h_1^{[1,2]} + h_1^{[0,3]}\right) + \partial_B \mbf u^{[0,1]}\left(h_2^{[3,0]} + h_2^{[2,1]} + h_2^{[1,2]} + h_2^{[0,3]}\right)
			\s 
			+ \partial_A \mbf u^{[3,0]} \, h_1^{[1,0]} + \partial_B \mbf u^{[0,3]} \, h_2^{[0,1]} + \partial_A \mbf u^{[2,1]} \, h_1^{[1,0]} + \partial_B \mbf u^{[1,2]} \, h_2^{[0,1]}
			\s 
			+ \partial_A \mbf u^{[1,2]} \, h_1^{[1,0]} + \partial_B \mbf u^{[2,1]} \, h_2^{[0,1]} + \mbox{conj.}
			\s 
			= \left(\jac \mbf f(\mbf 0) + \mathds D \, \partial_{xx}\right)\left(\mbf u^{[3,0]} + \mbf u^{[2,1]} + \mbf u^{[1,2]} + \mbf u^{[0,3]}\right)
			\s 
			+ 2 \, \mbf F_2\left(\mbf u^{[1,0]} + \mbf u^{[0,1]}, \mbf u^{[2,0]} + \mbf u^{[1,1]} + \mbf u^{[0,2]}\right)
			\s 
			+ \mbf F_3\left(\mbf u^{[1,0]} + \mbf u^{[0,1]}, \mbf u^{[1,0]} + \mbf u^{[0,1]}, \mbf u^{[1,0]} + \mbf u^{[0,1]}\right). \label{thirdordereq}
		\end{multline}
		This equation can be expanded as:
		\begin{multline*}
			\partial_A \, \mbf u^{[1,0]} \left(h_1^{[3,0]} + h_1^{[2,1]} + h_1^{[1,2]} + h_1^{[0,3]}\right) + \partial_B \, \mbf u^{[0,1]} \left(h_2^{[3,0]} + h_2^{[2,1]} + h_2^{[1,2]} + h_2^{[0,3]}\right)
			\s 
			+ i\omega A \, \partial_A \mbf u^{[3,0]} - i\omega B \, \partial_B \mbf u^{[0,3]} + i\omega A \, \partial_A \mbf u^{[2,1]} - i\omega B \, \partial_B \mbf u^{[1,2]}
			\s 
			+ i\omega A \, \partial_A \mbf u^{[1,2]} - i\omega B \, \partial_B \mbf u^{[2,1]} + \mbox{conj.}
			\s 
			= \left(\jac \mbf f(\mbf 0) + \mathds D \, \partial_{xx}\right)\left(\mbf u^{[3,0]} + \mbf u^{[2,1]} + \mbf u^{[1,2]} + \mbf u^{[0,3]}\right)
			\s 
			+ 2 \left(\mbf F_2\left(\mbf u^{[1,0]}, \mbf u^{[2,0]}\right) + \mbf F_2\left(\mbf u^{[1,0]}, \mbf u^{[1,1]}\right) + \mbf F_2\left(\mbf u^{[1,0]}, \mbf u^{[0,2]}\right)\right.
			\s 
			\left. + \mbf F_2\left(\mbf u^{[0,1]}, \mbf u^{[2,0]}\right) + \mbf F_2\left(\mbf u^{[0,1]}, \mbf u^{[1,1]}\right) + \mbf F_2\left(\mbf u^{[0,1]}, \mbf u^{[0,2]}\right)\right)
			\s 
			+ \mbf F_3\left(\mbf u^{[1,0]}, \mbf u^{[1,0]}, \mbf u^{[1,0]}\right) + \mbf F_3\left(\mbf u^{[0,1]}, \mbf u^{[0,1]}, \mbf u^{[0,1]}\right)
			\s 
			+ 3 \left(\mbf F_3\left(\mbf u^{[1,0]}, \mbf u^{[1,0]}, \mbf u^{[0,1]}\right) + \mbf F_3\left(\mbf u^{[1,0]}, \mbf u^{[0,1]}, \mbf u^{[0,1]}\right)\right),
		\end{multline*}
		which can be split, according to the amplitudes, as
		\begin{multline}
			\partial_A \, \mbf u^{[1,0]} \, h_1^{[3,0]} + \partial_B \, \mbf u^{[0,1]} \, h_2^{[3,0]} + i\omega A \, \partial_A \mbf u^{[3,0]} + \mbox{conj.} = \left(\jac \mbf f(\mbf 0) + \mathds D \, \partial_{xx}\right) \mbf u^{[3,0]}
			\s 
			+ |A|^2 \, A \, e^{ikx} \left(4 \, \mbf F_2\left(\bs \phi_1^{[1]}, \mbf W_0^{[2]}\right) + 2 \, \mbf F_2\left(\ol{\bs \phi_1^{[1]}}, \mbf W_2^{[2]}\right) + 3 \, \mbf F_3\left(\bs \phi_1^{[1]}, \bs \phi_1^{[1]}, \ol{\bs \phi_1^{[1]}}\right)\right)
			\s 
			+ A^3 \, e^{3ikx} \left( 2 \, \mbf F_2\left(\bs \phi_1^{[1]}, \mbf W_2^{[2]}\right) + \mbf F_3\left(\bs \phi_1^{[1]}, \bs \phi_1^{[1]}, \bs \phi_1^{[1]}\right)\right) + c.c., \label{firstequationthirdorder}
		\end{multline}
		and
		\begin{multline}
			\partial_A \, \mbf u^{[1,0]} \, h_1^{[2,1]} + \partial_B \, \mbf u^{[0,1]} \, h_2^{[2,1]} + i\omega A \, \partial_A \mbf u^{[2,1]} - i\omega B \, \partial_B \mbf u^{[2,1]} + \mbox{conj.}
			\s 
			= \left(\jac \mbf f(\mbf 0) + \mathds D \, \partial_{xx}\right) \mbf u^{[2,1]}
			\s 
			+ |A|^2 B \, e^{ikx} \left(2 \, \mbf F_2\left(\bs \phi_1^{[1]}, \ol{\mbf W_0^{[1,1]}}\right) + 2 \, \mbf F_2\left(\ol{\bs \phi_1^{[1]}}, 2 \, \mbf W_0^{[2]} + \mbf W_2^{[1,1]}\right) + 6 \, \mbf F_3\left(\bs \phi_1^{[1]}, \ol{\bs \phi_1^{[1]}}, \ol{\bs \phi_1^{[1]}}\right)\right)
			\s 
			+ A^2 \bar B \, e^{ikx} \left(2 \, \mbf F_2\left(\bs \phi_1^{[1]}, \mbf W_2^{[2]} + \mbf W_0^{[1,1]}\right) + 3 \, \mbf F_3\left(\bs \phi_1^{[1]}, \bs \phi_1^{[1]}, \bs \phi_1^{[1]}\right)\right)
			\s 
			+ A^2 B \, e^{3ikx} \left(2 \, \mbf F_2\left(\bs \phi_1^{[1]}, \mbf W_2^{[1,1]}\right) + 2 \, \mbf F_2\left(\ol{\bs \phi_1^{[1]}}, \mbf W_2^{[2]}\right) + 3 \, \mbf F_3\left(\bs \phi_1^{[1]}, \bs \phi_1^{[1]}, \ol{\bs \phi_1^{[1]}}\right)\right) + c.c., \label{secondequationthirdorder}
		\end{multline}
		while the remaining equations are just complex-conjugates of \eqref{firstequationthirdorder} and \eqref{secondequationthirdorder}. Next, as these are linear systems of equations, we know that the vector functions that solve them have the following forms:
		\begin{align*}
			\mbf u^{[3,0]}&= |A|^2 \, A \, e^{ikx} \, \mbf W_1^{[3]} + A^3 \, e^{3ikx} \, \mbf W_3^{[3]} +c.c.,
			\s 
			\mbf u^{[2,1]} &= |A|^2 \, B \, e^{ikx} \, \ol{\mbf W_1^{[2,1]}} + A^2 \bar B \, e^{ikx} \, \mbf W_{1,2}^{[2,1]} + A^2 B \, e^{3ikx} \, \mbf W_3^{[2,1]} + c.c.,
			\s 
			\mbf u^{[1,2]} &= |B|^2 \, A \, e^{ikx} \, \mbf W_1^{[2,1]} + \bar A B^2 \, e^{ikx} \, \ol{\mbf W_{1,2}^{[2,1]}} + AB^2 \, e^{3ikx} \, \ol{\mbf W_3^{[2,1]}} + c.c.,
			\s 
			\mbf u^{[0,3]} &= |B|^2 \, B \, e^{ikx} \, \ol{\mbf W_1^{[3]}} + B^3 \, e^{3ikx} \, \ol{\mbf W_3^{[3]}} + c.c..
		\end{align*}
		Next, when replacing each term of $\mbf u^{[3,0]}$ and $\mbf u^{[2,1]}$ into \eqref{firstequationthirdorder} and \eqref{secondequationthirdorder}, respectively, we can see that
        \begin{multline}
			\left(\jac \mbf f(\mbf 0) - k^2 \, \mathbb D - i \omega \, I\right) \, \mbf W_1^{[3]} = \bs \phi_1^{[1]} \, \frac{h_1^{[3,0]}}{\abs{A}^2 A} + \ol{\bs \phi_1^{[1]}} \, \frac{h_2^{[3,0]}}{\abs{A}^2 A}
			\s 
			- \left(4 \, \mbf F_2\left(\bs \phi_1^{[1]}, \mbf W_0^{[2]}\right) + 2 \, \mbf F_2\left(\overline{\bs \phi_1^{[1]}}, \mbf W_2^{[2]}\right) + 3 \, \mbf F_3\left(\bs \phi_1^{[1]}, \bs \phi_1^{[1]}, \ol{\bs \phi_1^{[1]}}\right)\right), \label{secularQ13}
		\end{multline}
		\begin{align}
			\left(\jac \mbf f(\mbf 0) - 9 \, k^2 \,\mathds D - 3 \, i\omega \, I\right)\mbf W_3^{[3]} = - 2 \, \mbf F_2\left(\bs \phi_1^{[1]}, \mbf W_2^{[2]}\right) - \mbf F_3\left(\bs \phi_1^{[1]}, \bs \phi_1^{[1]}, \bs \phi_1^{[1]}\right), \label{nonsecularQ33}
		\end{align}
        \begin{multline}
			\left(\jac \mbf f(\mbf 0) - k^2 \, \mathds D + i\omega \, I\right) \, \ol{\mbf W_1^{[2,1]}} = \bs \phi_1^{[1]} \, \frac{h_1^{[2,1]}}{|A|^2 \, B} + \ol{\bs \phi_1^{[1]}} \, \frac{h_2^{[2,1]}}{|A|^2 \, B}
			\s 
			- \left(2 \, \mbf F_2\left(\bs \phi_1^{[1]}, \ol{\mbf W_0^{[1,1]}}\right) + 2 \, \mbf F_2\left(\ol{\bs \phi_1^{[1]}}, 2 \, \mbf W_0^{[2]} + \mbf W_2^{[1,1]}\right) + 6 \, \mbf F_3\left(\bs \phi_1^{[1]}, \ol{\bs \phi_1^{[1]}}, \ol{\bs \phi_1^{[1]}}\right)\right), \label{secularQ121}
		\end{multline}
		\begin{align}
			\left(\jac \mbf f(\mbf 0) - k^2 \, \mathds D - 3 \, i\omega \, I\right)\mbf W_{1,2}^{[2,1]} = - 2 \, \mbf F_2\left(\bs \phi_1^{[1]}, \mbf W_2^{[2]} + \mbf W_0^{[1,1]}\right) - 3 \, \mbf F_3\left(\bs \phi_1^{[1]}, \bs \phi_1^{[1]}, \bs \phi_1^{[1]}\right), \label{nonsecularQ1221}
		\end{align}
		and
		\begin{align}
			\left(\jac \mbf f(\mbf 0) - 9 \, k^2 \, \mathds D - i \omega \, I\right)\mbf W_3^{[2,1]} = - 2 \, \mbf F_2\left(\bs \phi_1^{[1]}, \mbf W_2^{[1,1]}\right) - 2 \, \mbf F_2\left(\ol{\bs \phi_1^{[1]}}, \mbf W_2^{[2]}\right) \notag
            \\
            - 3 \, \mbf F_3\left(\bs \phi_1^{[1]}, \bs \phi_1^{[1]}, \ol{\bs \phi_1^{[1]}}\right). \label{nonsecularQ321}
		\end{align}
		Note that \eqref{nonsecularQ33}, \eqref{nonsecularQ1221}, and \eqref{nonsecularQ321} have a unique solution since the matrices on their left-hand side are invertible.
  
        On the other hand, equations \eqref{secularQ13} and \eqref{secularQ121} may have secular terms which would break the assumptions of our asymptotic scaling. To deal with such terms, as is common in multi-scale asymptotic analysis, we use the Fredholm alternative to demand that
		\begin{align*}
			\Im(\jac \mbf f(\mbf 0)+ \mathds D \, \partial_{xx} - i\omega \, I)^\perp &=\ker\left(\left(\jac \mbf f(\mbf 0)+\mathds D \, \partial_{xx} - i\omega \, I\right)^*\right)
			\s 
			&=\ker\left(\jac \mbf f(\mbf 0)^\intercal+ \mathds D^\intercal \, \partial_{xx} + i\omega \, I\right).
		\end{align*}
        Here, and henceforth we choose the natural inner product within the space of complex-valued $L^2$ periodic functions with period $2\pi/k$:
        \begin{align}
			\left \langle f(x) \, \mbf a, g(x) \, \mbf b \right \rangle &= \frac{k}{2\pi}\int_0^{\frac{2\pi}{k}} f(x) \, \ol{g(x)} \, \dd x \, \mbf a \cdot \bar{\mbf b}, \label{innerp}
        \end{align}
		for $f,g\in L^2\left(\left[0, \frac{2\pi}{k}\right]\right)$ and $\mbf a,\mbf b\in \mathds C^n$, where $\cdot$ is the usual dot product of vectors with $n$ components.

        We now need to choose a vector $\bs{\hat \psi}_{1,\pm}^{[1]}$ as a non-zero solution to the following equation:
		\begin{align*}
			\left(\jac \mbf f(\mbf 0)^\intercal + \mathds D^\intercal \, \partial_{xx} - i\omega \, I\right) \bs{\hat \psi}_{1,\pm}^{[1]}&=0.
		\end{align*}
        A convenient choice is $\bs{\hat \psi}_{1,\pm}^{[1]} = \bs \psi_1^{[1]} \, e^{\pm ikx}$, where
		\begin{align*}
			\left(\jac \mbf f(\mbf 0)^\intercal - k^2 \, \mathds D^\intercal - i\omega \, I\right) \bs \psi_1^{[1]} &= 0.
		\end{align*}
        Then we observe that for a real matrix $A$ with a complex eigenvalue $i\omega$, and the corresponding eigenvectors $\mbf p, \mbf q$ defined as:
		\begin{align*}
			A \, \mbf p &= i\omega \, \mbf p,
			\s 
			A^\intercal \, \mbf q &= i\omega \, \mbf q,
		\end{align*}
		we have that
		\begin{align*}
			\left \langle \mbf p, \mbf q \right \rangle&=\left \langle \mbf p,\frac{A^\intercal \, \mbf q}{i \omega} \right \rangle=\left \langle A \, \mbf p,\frac{\mbf q}{i \omega} \right \rangle=-\frac{1}{i\omega}\left \langle i\omega \, \mbf p,\mbf q\right \rangle=-\left \langle \mbf p, \mbf q\right \rangle,
		\end{align*}
        which implies that $\left \langle \mbf p,\mbf q \right \rangle=0$. In our particular case, this implies that
		\begin{align*}
			\left \langle \bs \phi_1^{[1]} \, e^{\pm ikx},\bs \psi_1^{[1]} \, e^{\pm ikx}\right \rangle =\bs \phi_1^{[1]} \cdot \ol{\bs \psi_1^{[1]}}=0.
		\end{align*}
        With this, applying the inner product with $\ol{\bs \psi_1^{[1]}}$ to the right-hand side of \eqref{secularQ13}, and equating that expression to zero, we can see that $h_1^{[3,0]}=C_1^{[3,0]} \, |A|^2A$, where
		\begin{align*}
			C_1^{[3,0]} = \frac{1}{\bs \psi_1^{[1]} \cdot \bs \phi_1^{[1]}} \, \bs \psi_1^{[1]} \cdot \left(4 \, \mbf F_2\left(\bs \phi_1^{[1]}, \mbf W_0^{[2]}\right) + 2 \, \mbf F_2 \left(\ol{\bs \phi_1^{[1]}}, \mbf W_2^{[2]}\right) + 3 \, \mbf F_3\left(\bs \phi_1^{[1]}, \bs \phi_1^{[1]}, \ol{\bs \phi_1^{[1]}}\right)\right).
		\end{align*}
        Similarly, using $\bs \psi_1^{[1]}$ to the right-hand side of \eqref{secularQ121}, we find $h_2^{[2, 1]} = C_2^{[2,1]} \, |A|^2 \, B$, where
		\begin{multline*}
			C_2^{[2, 1]} = \frac{1}{\ol{\bs \psi_1^{[1]}} \cdot \ol{\bs \phi_1^{[1]}}} \, \ol{\bs \psi_1^{[1]}} \cdot \left(2 \, \mbf F_2\left(\bs \phi_1^{[1]}, \ol{\mbf W_0^{[1,1]}}\right) + 2 \, \mbf F_2\left(\ol{\bs \phi_1^{[1]}}, 2 \, \mbf W_0^{[2]} + \mbf W_2^{[1,1]}\right)\right.
            \\
            \left. + 6 \, \mbf F_3\left(\bs \phi_1^{[1]}, \ol{\bs \phi_1^{[1]}}, \ol{\bs \phi_1^{[1]}}\right)\right).
		\end{multline*}
		In addition, we have that $h_2^{[0,3]} = \ol{C_1^{[3,0]}} \, |B|^2 \, B$, and $h_1^{[1, 2]} = C_1^{[1,2]} \, A \, |B|^2$, where $C_1^{[1,2]} = \ol{C_2^{[2, 1]}}$.
  
        Finally, note that we have not chosen the functions $h_1^{[0,3]}$, $h_2^{[3,0]}$, $h_1^{[2,1]}$ and $h_2^{[1,2]}$. Following \cite{Kuznetsov}, these terms are not important at the third order since they can be removed under a smooth change of variables. We, therefore, make the choice that each of these coefficients is zero, which is consistent with our choice of parametrisation of the resonant terms in order to write the truncated normal form as \eqref{eq:nf1}, \eqref{eq:nf2}.
        
        With this choice, equations \eqref{secularQ13} and \eqref{secularQ121} become
		\begin{align*}
			\left(\jac \mbf f(\mbf 0)-k^2 \, \mathbb D - i \omega \, I\right) \mbf W_1^{[3]} = C_1^{[3,0]} \, \bs \phi_1^{[1]} - 4 \, \mbf F_2\left(\bs \phi_1^{[1]}, \mbf W_0^{[2]}\right) - 2 \, \mbf F_2\left(\overline{\bs \phi_1^{[1]}}, \mbf W_2^{[2]}\right)
            \\
            - 3 \, \mbf F_3\left(\bs \phi_1^{[1]}, \bs \phi_1^{[1]}, \overline{\bs \phi_1^{[1]}}\right),
		\end{align*}
		and
		\begin{multline*}
			\left(\jac \mbf f(\mbf 0) - k^2 \, \mathds D + i\omega \, I\right) \, \ol{\mbf W_1^{[2,1]}} = \ol{C_1^{[1, 2]}} \, \ol{\bs \phi_1^{[1]}}
			\s 
			- \left(2 \, \mbf F_2\left(\bs \phi_1^{[1]}, \ol{\mbf W_0^{[1,1]}}\right) + 2 \, \mbf F_2\left(\ol{\bs \phi_1^{[1]}}, 2 \, \mbf W_0^{[2]} + \mbf W_2^{[1,1]}\right) + 6 \, \mbf F_3\left(\bs \phi_1^{[1]}, \ol{\bs \phi_1^{[1]}}, \ol{\bs \phi_1^{[1]}}\right)\right),
		\end{multline*}
		respectively. These equations have infinitely many solutions because the matrices of coefficients on the left-hand side have a non-trivial kernel. As is conventional, we choose a solution by demanding that $\mbf W_1^{[3]}$ and $\mbf W_1^{[2,1]}$ to be orthogonal to $\bs \phi_1^{[1]}$ with respect to the inner product defined by \eqref{innerp}.
		
		We are now ready to continue up to orders 4 and 5. The method is similar, but the calculations are somewhat laborious. Therefore, the details are relegated to \ref{ap:order5}.

    \subsection{Unfolding} \label{sub:unfolding}
        To complete the normal form, we need to analyse what happens in the case that $\varepsilon\neq 0$. To do that, we need to consider terms in the amplitude equation that are of first order in $\varepsilon$. Thus we adjust the general expansions to read 
		\begin{align*}
			\partial_t \, A = \sum_{\substack{0\leq j, \ell\leq 5 \\ j + \ell > 0 \\ 0 \leq m \leq 1}} h_1^{[j, \ell, m]}(A, B, \varepsilon),
			\qquad 
			\partial_t \, B = \sum_{\substack{0\leq j, \ell\leq 5 \\ j + \ell > 0 \\ 0 \leq m \leq 1}} h_2^{[j, \ell, m]}(A, B, \varepsilon),
		\end{align*}
		where the index $m$ considers the degree on $\varepsilon$. Here, we have already determined all the terms in which $m = 0$. We only need to find $h_1^{[j, \ell, 1]}$ and $h_2^{[j, \ell, 1]}$ for $j + \ell = 1$. To achieve this, we need to consider a new term in the expansion of $\mbf u$ as well. Specifically, we consider that
		\begin{align*}
			\mbf u = \sum_{\substack{0\leq j, \ell\leq 5 \\ j + \ell>0 \\ 0\leq m\leq 1}} \mbf u^{[j, \ell, m]}(A, B, \varepsilon).
		\end{align*}
		Now, note that the expansion of \eqref{eq:geneq} at first order in $A, B$ and $\varepsilon$ is given by
		\begin{multline*}
			\partial_A \, \mbf u^{[1,0,1]} \, h_1^{[1,0,0]} + \partial_B \, \mbf u^{[0,1,1]} \, h_2^{[0,1,0]}
			\s 
			+ \partial_A \, \mbf u^{[1,0,0]} \left(h_1^{[1,0,1]} + h_1^{[0,1,1]}\right) + \partial_B \, \mbf u^{[0,1,0]} \left(h_2^{[1,0,1]} + h_2^{[0,1,1]}\right) + \mbox{conj.}
			\\
			= \left(\jac \mbf f(\mbf 0) - k^2 \, \mathbb D\right) \left(\mbf u^{[1,0,1]} + \mbf u^{[0,1,1]}\right) + \varepsilon \, \mbf F_{1,1} \left(\mbf u^{[1,0,0]} + \mbf u^{[0,1,0]}\right)
			\\
			+ \varepsilon \, \frac{\dd \mathbb D}{\dd \varepsilon}(0) \, \partial_{xx} \left(\mbf u^{[1,0,0]} + \mbf u^{[0,1,0]}\right).
		\end{multline*}
		Next, if we set $h_1^{[0, 1, 1]} = h_2^{[1, 0, 1]} = 0$, then this expression gets equivalent to
		\begin{multline*}
			i\omega A \, \partial_A \, \mbf u^{[1,0,1]} - i\omega B \, \partial_B \, \mbf u^{[0,1,1]} + \partial_A \, \mbf u^{[1,0,0]} \, h_1^{[1,0,1]} + \partial_B \, \mbf u^{[0,1,0]} \, h_2^{[0,1,1]} + \mbox{conj.}
			\\
			= \left(\jac \mbf f(\mbf 0) - k^2 \, \mathbb D\right) \left(\mbf u^{[1,0,1]} + \mbf u^{[0,1,1]}\right) + \varepsilon \, \mbf F_{1,1} \left(\mbf u^{[1,0,0]} + \mbf u^{[0,1,0]}\right)
			\\
			+ \varepsilon \, \frac{\dd \mathbb D}{\dd \varepsilon}(0) \, \partial_{xx} \left(\mbf u^{[1,0,0]} + \mbf u^{[0,1,0]}\right).
		\end{multline*}
		Now, if we split this equation according to the amplitudes, we have that
		\begin{multline*}
			i\omega A \, \partial_A \, \mbf u^{[1,0,1]} + \partial_A \, \mbf u^{[1, 0, 0]} \, h_1^{[1, 0, 1]} + \mbox{conj.} = \left(\jac \mbf f(\mbf 0) - k^2 \, I\right) \, \mbf u^{[1,0,1]} + \varepsilon \, \mbf F_{1,1} \left(\mbf u^{[1, 0, 0]}\right)
			\\
			- k^2 \, \varepsilon \, \frac{\dd \mathbb D}{\dd \varepsilon}(0) \, \left(\mbf u^{[1,0,0]}\right).
		\end{multline*}
		This implies that the solution is given by:
		\begin{align*}
			\mbf u^{[1,0,1]} = \varepsilon A \, e^{ikx} \, \mbf W_1^{[1,0,1]} + c.c., \qquad \mbf u^{[0,1,1]} = \varepsilon B \, e^{ikx} \, \ol{\mbf W_1^{[1,0,1]}} + c.c.,
		\end{align*}
		where
		\begin{multline*}
			\left(\jac \mbf f(\mbf 0) - k^2 \, I - i\omega \, I\right) \left(\varepsilon A \, e^{ikx} \, \mbf W_1^{[1,0,1]}\right) = e^{ikx} \, \bs \phi_1^{[1]} \, h_1^{[1,0,1]} - \varepsilon A \, e^{ikx} \, \mbf F_{1,1} \left(\bs \phi_1^{[1]}\right)
            \\
            + k^2 \, \varepsilon A \, e^{ikx} \, \frac{\dd \mathbb D}{\dd \varepsilon}(0) \, \bs \phi_1^{[1]}.
		\end{multline*}
		Again, in this equation we might have some secular terms, so we need to apply the inner product with $\ol{\bs \psi_1^{[1]}} \, e^{ikx}$ again to determine the value of $h_1^{[1,0,1]}$. In particular, we have that $h_1^{[1,0,1]} = C_1^{[1,0,1]} \, \varepsilon A$, where
		\begin{align*}
			C_1^{[1,0,1]} = \frac{1}{\bs \psi_1^{[1]} \cdot \bs \phi_1^{[1]}} \, \bs \psi_1^{[1]} \cdot \left(\mbf F_{1,1}\left(\bs \phi_1^{[1]}\right) - k^2 \, \frac{\dd \mathbb D}{\dd \varepsilon}(0) \, \bs \phi_1^{[1]}\right),
		\end{align*}
		and $h_2^{[0, 1, 1]} = \ol{C_1^{[1, 0, 1]}} \, \varepsilon B$.

        This completes the proof of our main result.

\section{Examples}\label{sec:examples}
    We now proceed to illustrate the Theorem through the computation of bifurcation conditions and normal form coefficients in example systems. Details of all the computations are available in an accompanying online repository \cite{criticality-wave}, where complete algebraic expressions are available. As in most cases, these expressions are rather cumbersome, we present here just a summary of the results in terms of two-parameter bifurcation diagrams indicating the super- or sub-criticality of the wave bifurcations, along with other relevant curves. We also present results of numerical simulations, obtained using {\tt Matlab}, {\tt PDE2path} \cite{pde2path} and {\tt Auto} \cite{auto}.
            
    \subsection{Three-component model of excitable media}
        We consider a three-variable generalisation of the FitzHugh-Nagumo system of equations, which allows for diffusion of all three components, as studied
        independently by Yang {\em et al.}~\cite{Yang} and Yochelis {\em et al.}~\cite{YochelisKnobloch}. Rather than representing specific neurodynamics, the system can be considered as a general model for an excitable system with three active components; $u(x, t)$, $v(x, t)$ and $w(x, t)$. The model can be written in dimensionless form as 
 	\begin{align}
    		\partial_t u &= f_3\left(u, k_0\right) - k_v \, v - k_w \, w + D_1 \, \partial_{xx} u, \label{eq:fhn1}
    		\\
    		\partial_t v &= \eps_v \left(u - a_v \, v - a_w \, w - a_0\right) + D_2 \, \partial_{xx} v, \label{eq:fhn2}
    		\\ 
    		\partial_t w &= \eps_w \left(u - w\right) + D_3 \, \partial_{xx} w, \label{eq:fhn3}
    	\end{align}
        where $f_3\left(u, k_0\right)$ is a cubic function in $u$, and $\eps_v$, $\eps_w$, $a_\alpha$, $k_{\alpha}$ and $D_i$ for $\alpha \in \{0, v, w\}$, and $i = 1, 2, 3$ are non-negative parameters. 

        \paragraph{Subcritical bifurcation} 
            We begin by considering the variant of the model \eqref{eq:fhn1}--\eqref{eq:fhn3} previously studied in \cite{YochelisKnobloch}, in which $f_3\left(u, k_0\right) = u - u^3$, $k_w = 0$, and $k_v = 1$. The changes we consider here with respect to \cite{YochelisKnobloch} is that we exclusively study the problem in one spatial dimension and allow $D_2 > 0$.
        
            The homogeneous steady state we analyze is given by
        	\begin{align*}
        		\mbf P &= \left(u^*, v^*, w^*\right),
        	\end{align*}
        	where
        	\begin{align*}
        		u^*& = w^* = -\frac{\sqrt[3]{2} \, \left(a_v+a_w-1\right)}{M}-\frac{M}{3 \sqrt[3]{2} \, a_v},
        		\s 
        		v^* &= \begin{multlined}[t][10cm]
        			\frac{2^{2/3} \, \left(a_w - 1\right) \, M^2 - 6 \, a_0 \, a_v \, M + 6 \, \sqrt[3]{2} \, a_v \, \left(a_w - 1\right) \left(a_v + a_w - 1\right)}{6 \, a_v^2 M},
        		\end{multlined}
        	\end{align*}
            and
            $$
                M = \left[\sqrt{ 729 \, a_0^2 \, a_v^4-108 \, a_v^3 (a_v+a_w-1)^3} \: -27 \, a_0 \, a_v^2 \right]^{1/3}.
            $$
            We consider the following parameters fixed
        	\begin{align}
        		\left(a_w, \eps_v, \eps_w, a_0, D_1, D_2\right) = \left(0.5, 0.2, 1, -0.1, 0.005, 1\right) \label{parvalarik}
        	\end{align}
            and allow $D_2$ and $a_v$ to be bifurcation parameters.

            In \cite{YochelisKnobloch}, the authors state that there is a subcritical wave bifurcation when $D_2 = 0$ and $a_v \approx 0.4497$. Figure \ref{fig:bifdiag} shows the bifurcation curves we have computed for the system. We find three different kinds of bifurcation of $\mbf P$ --- a global (i.e.~wave-number $k = 0$) Hopf, Turing and wave bifurcations. In this figure and subsequent similar ones, we also show as insets the shape of the dispersion relation local to $\mbf P$ in each parameter region. We also adopt a uniform colouring scheme for different bifurcation curves, which also captures changes in criticality. Moreover, in Figure \ref{fig:bifdiag}, we can identify three separate codimension-two bifurcation points: a Turing-Hopf, a wave-Turing, and a global-Hopf-wave.

            However, the three bifurcations we identified in this example are all found to be subcritical for all the parameter values in the diagram. Hence, the low-amplitude Turing, wave, or globally periodic patterns arising at those bifurcations are unstable. 
        	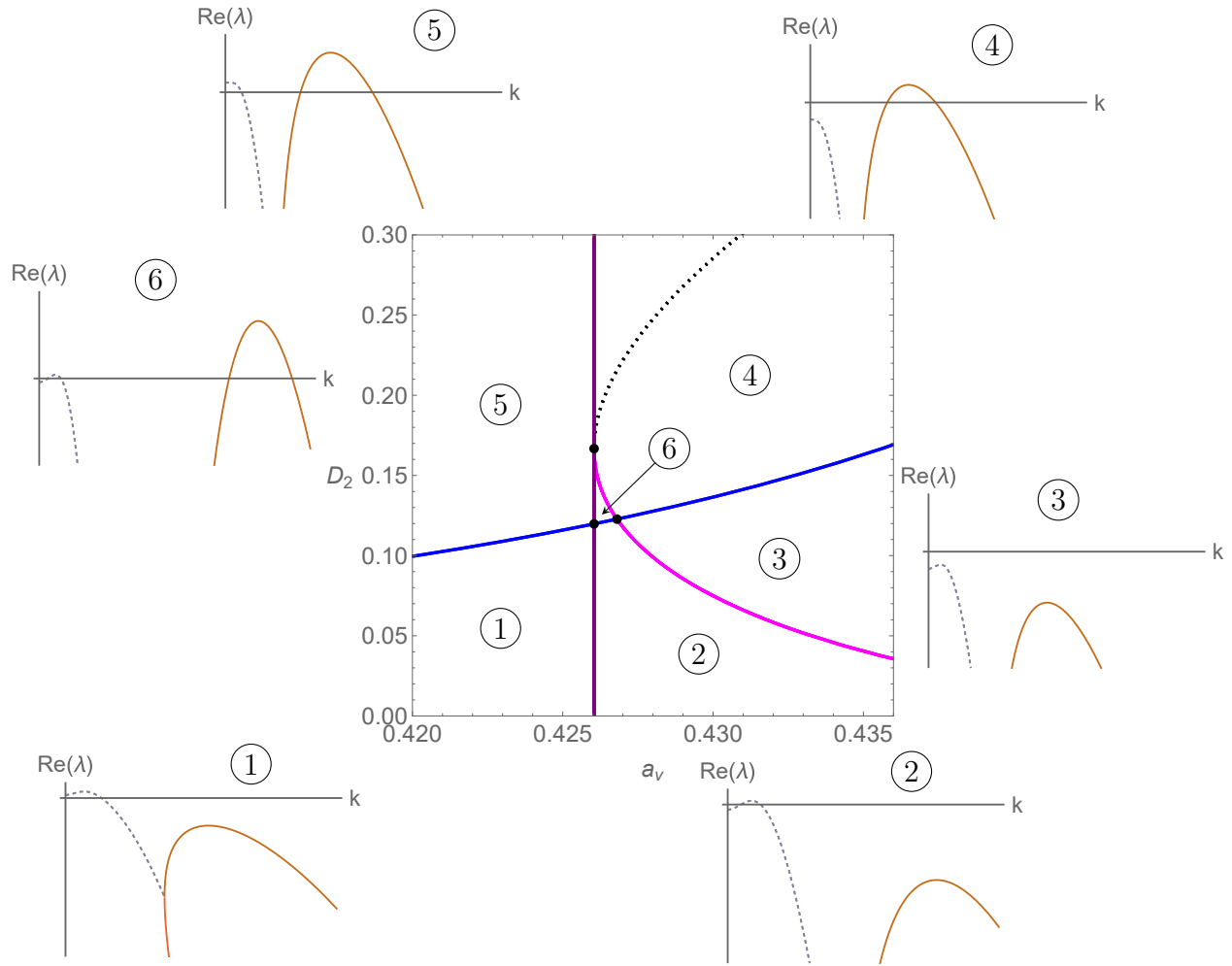
\begin{figure}[tbp]
                \centering
                \begin{tikzpicture}
                    \node (image) at (0,0) {
                        \includegraphics[width=\textwidth]{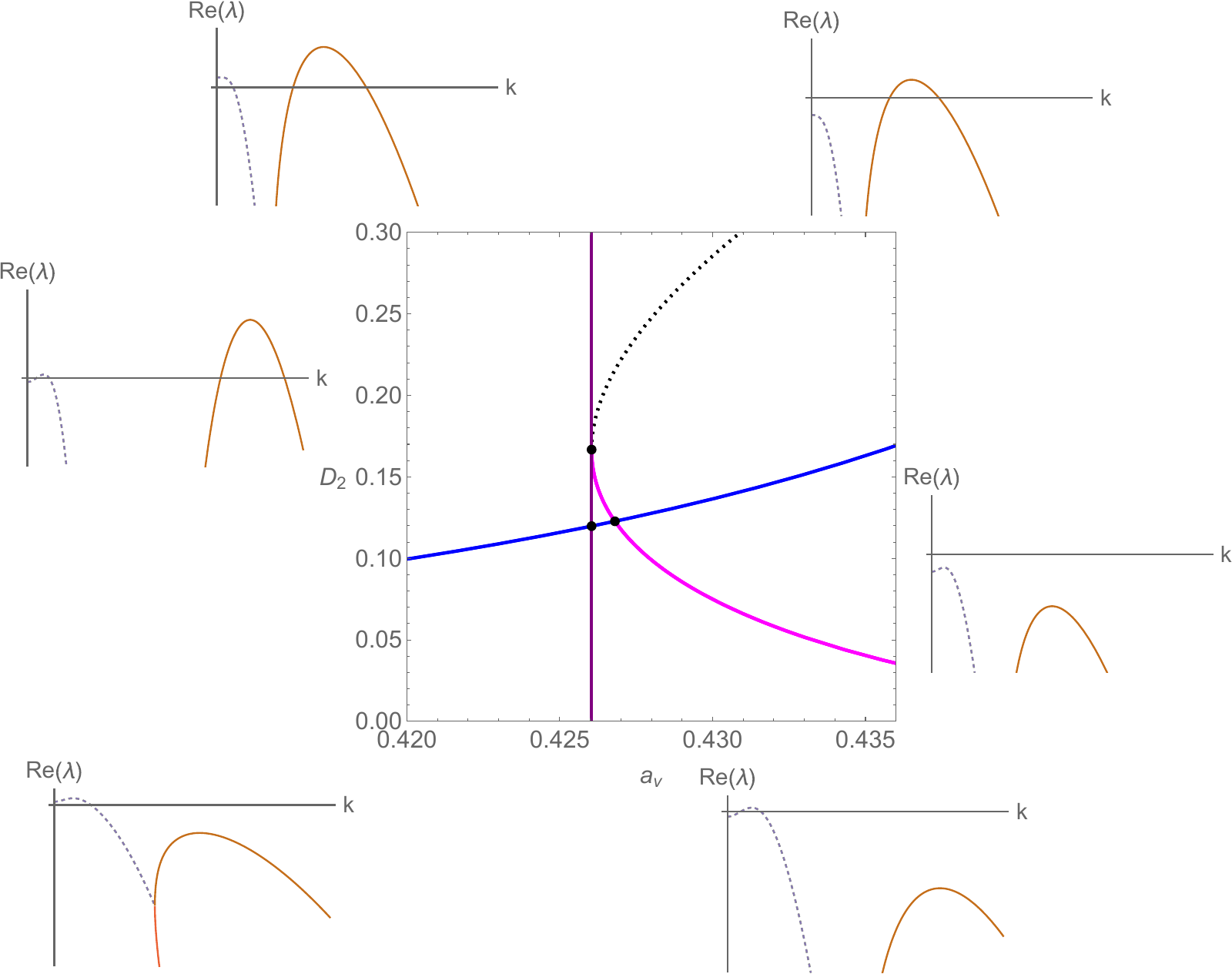}
                    };
                    \node (text) at (- 1.6, - 1.95) {\circled{1}};
                    \node (text) at (- 5.0, - 3.8) {\circled{1}};
                    \node (text) at (1.1, - 2.3) {\circled{2}};
                    \node (text) at (4.0, - 3.9) {\circled{2}};
                    \node (text) at (2.2, - 1.0) {\circled{3}};
                    \node (text) at (6.0, - 0.2) {\circled{3}};
                    \node (text) at (1.8, 1.5) {\circled{4}};
                    \node (text) at (5.1, 6.0) {\circled{4}};
                    \node (text) at (- 1.6, 1.1) {\circled{5}};
                    \node (text) at (- 2.5, 6.2) {\circled{5}};
                    \node (text) at (0.7, 0.5) {\circled{6}};
                    \node (text) at (- 6.3, 2.8) {\circled{6}};
                    \draw [-stealth](0.51, 0.31) -- (- 0.22, - 0.4);
                \end{tikzpicture}
        		\caption{Bifurcation diagram of model \eqref{eq:fhn1}--\eqref{eq:fhn3} when $f_3\left(u, k_0\right) = u - u^3$, $k_w = 0$, and $k_v = 1$ at $\mbf P$, for the parameter values given in \eqref{parvalarik}. The purple vertical line is a subcritical Hopf bifurcation of the steady-state, the magenta curve represents a subcritical wave bifurcation, while the blue curve represents a subcritical Turing bifurcation. Furthermore, the black dotted curve is formed by points in which the conditions for a wave bifurcation are fulfilled for a negative value of $k^2$. Here, $(a_v, D_2)\in [0.42, 0.44] \times [0.05, 0.25]$. The marked dots are codimension-two bifurcation points and the numbers denote the regions delimited by the bifurcation curves, except the dotted black curve. The images surrounding the bifurcation diagram show the dispersion relations one finds in each of its regions. Continuous (resp.~ dashed) lines represent real (resp.~complex with nonzero imaginary part) eigenvalues}
        		\label{fig:bifdiag}
        	\end{figure}
            \begin{figure}
                \centering
                \begin{subfigure}[b]{0.4\textwidth}
                    \centering
                    \includegraphics[width = \textwidth]{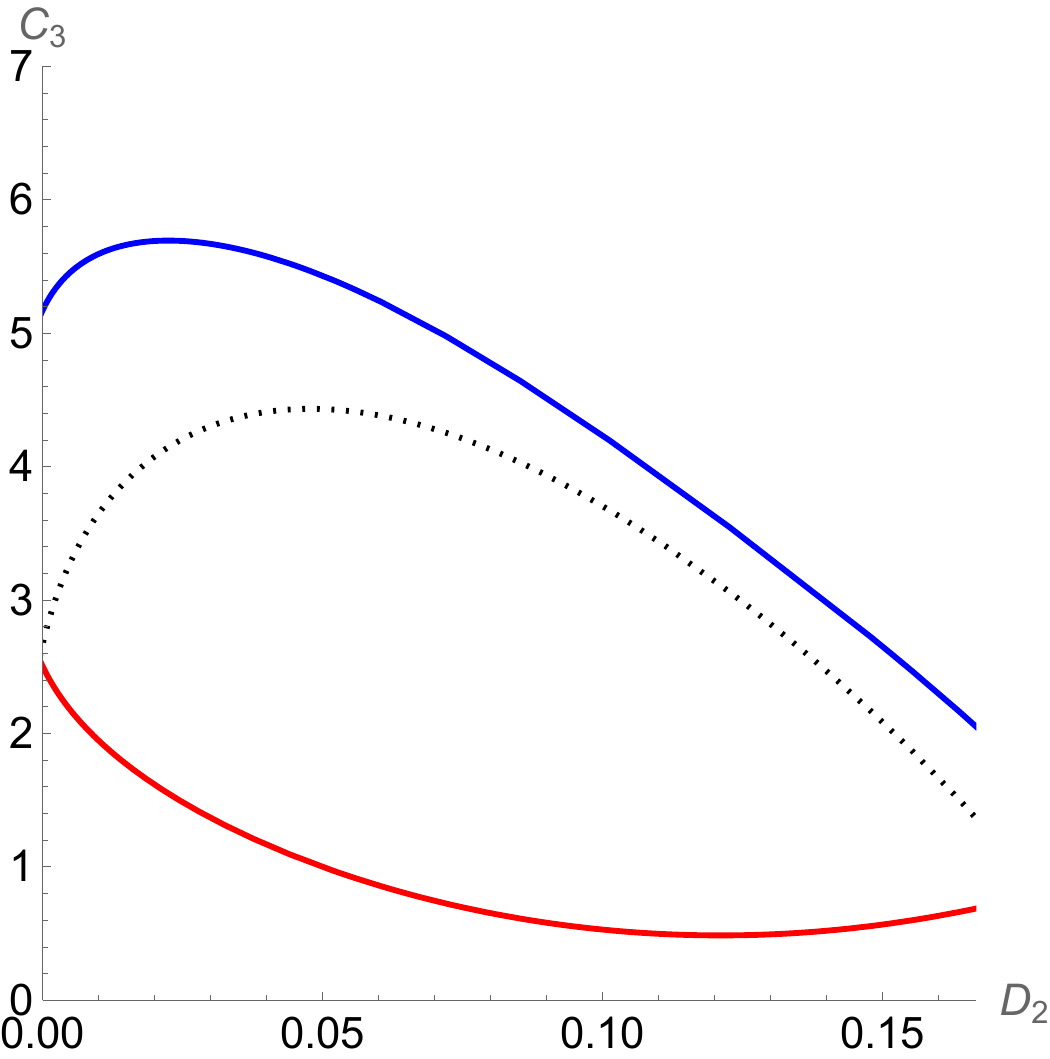}
                    \caption{}
                \end{subfigure}
                \\
                \begin{subfigure}[b]{0.4\textwidth}
                    \centering
                    \includegraphics[width = \textwidth]{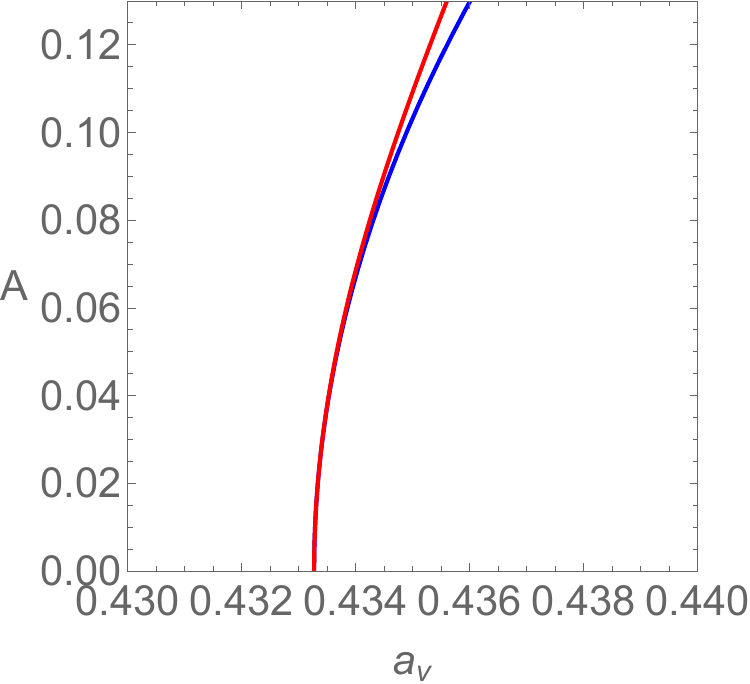}
                    \caption{}
                \end{subfigure}
                \hfill
                \begin{subfigure}[b]{0.4\textwidth}
                    \centering
                    \includegraphics[width = \textwidth]{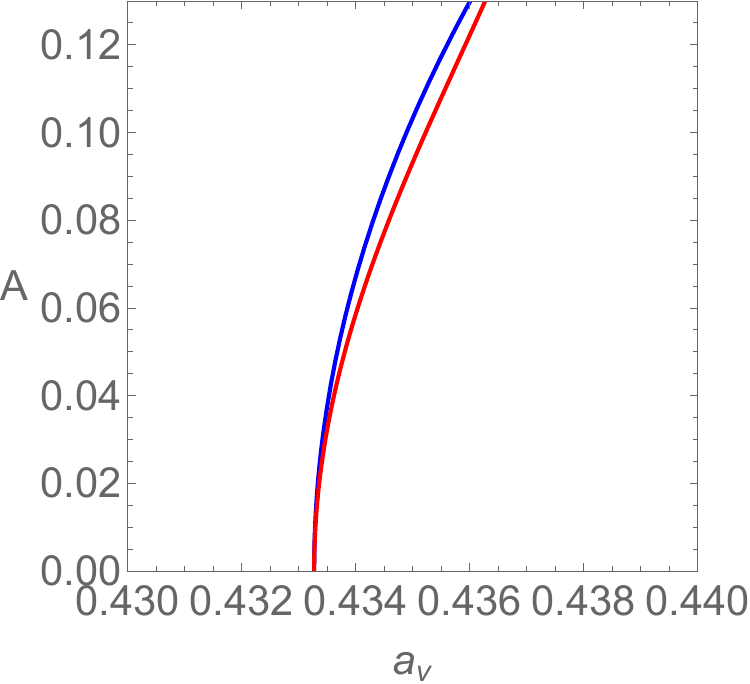}
                    \caption{}
                \end{subfigure}
                \caption{(a) Third-order coefficients along the wave bifurcation curve shown in Figure \ref{fig:bifdiag}. The red (resp.~blue) curve represents the numerical value of $\Re\left(C_1^{[3, 0]}\right)$ $\left(\text{resp. } \Re\left(C_1^{[3, 0]} + C_1^{[1, 2]}\right)\right)$, while the black dotted line represents the numerical value of $\Re\left(C_1^{[1, 2]}\right)$. On the other hand, (b) and (c) show a comparison of amplitudes of traveling and standing waves, respectively,  with the amplitude obtained from the third-order normal form for system \eqref{eq:fhn1}--\eqref{eq:fhn3}, when $f_3\left(u, k_0\right) = u - u^3$, $k_w = 0$, and $k_v = 1$, for the parameter values given in \eqref{parvalarik}, and $D_2 = 0.05$. The blue (resp. red) curves represent the amplitudes obtained through the normal form (resp.~numerics).}
                \label{fig:C3FHN}
            \end{figure}
	
            Figure \ref{fig:bifdiag} also shows the qualitative form of the dispersion relation in each region of the bifurcation diagram. The continuous lines represent real eigenvalues, whilst the dotted lines are the real part of the complex ones. There, we can see where the homogeneous steady-state is stable and the kind of instability it goes through after a bifurcation. In particular, $\mbf P$ is stable in Region 3, and it can go unstable through either a wave bifurcation as $\left(a_v, D_2\right)$ moves into Region 2, or a Turing bifurcation when moving to Region 4. Another option is, of course, to go directly through the corresponding codimension-two bifurcation point to Region 6, where $\mbf P$ has two kinds of instabilities for different wavenumbers. For a low wavenumber, $\mbf P$ is wave-unstable, whilst it is unstable to a Turing instability for a higher value of $k$. Nevertheless, the integration of this system just goes away from those patterns after some time because they are unstable when they appear.

        \paragraph{Unstable supercritical travelling waves.}
            Next, we consider the variant of model \eqref{eq:fhn1}--\eqref{eq:fhn3} studied in \cite{Yang}, in which $f_3\left(u, k_0\right) = k_0 + 2 \, u - u^3$, $a_0 = a_w = 0$, $a_v = a_w = 1$, and $\eps_v = 1/\tau$. The authors of \cite{Yang} focus exclusively on standing waves, for which they find the wave bifurcation to be subcritical, leading to the existence of solitary waves. Here, we focus instead on travelling waves that appear at the same bifurcation point under appropriate boundary conditions.

            We work with the homogeneous steady-state given by
        	\begin{align*}
        		\mbf P = \left(u^*, v^*, w^*\right),
        	\end{align*}
            where
        	\begin{align*}
        		u^* &= v^* = w^* = \frac{M}{3 \sqrt[3]{2}} - \frac{\sqrt[3]{2} \, \left(k_v + k_w - 2\right)}{M},
        		\\
        		M &= \left[\sqrt{108 \left(k_v + k_w - 2\right)^3 + 729 \, k_0^2} + 27 \, k_0\right]^{1/3}, 
        	\end{align*}
            choose values of fixed parameters to be \cite{Yang}
        	\begin{align}
        		\left(k_v, \tau, D_1, D_2, D_3\right) = (10, 50, 1, 1, 60), \label{parametervals}
        	\end{align}
            and allow $(k_0, k_w)\in [-12, -6]\times [1, 15]$ to be bifurcation parameters. 

            Again, in this case, we find three different kinds of bifurcation at $\mbf P$ --- a global (i.e.~wave-number $k = 0$) Hopf, Turing, and wave bifurcations --- see Fig.~\ref{fig:THbifurcationdiagram}.

        	\begin{figure}[tbp]
        		\centering
        		\begin{tikzpicture}
        		    \node (image) at (0,0) {
        		    \includegraphics[scale=0.5]{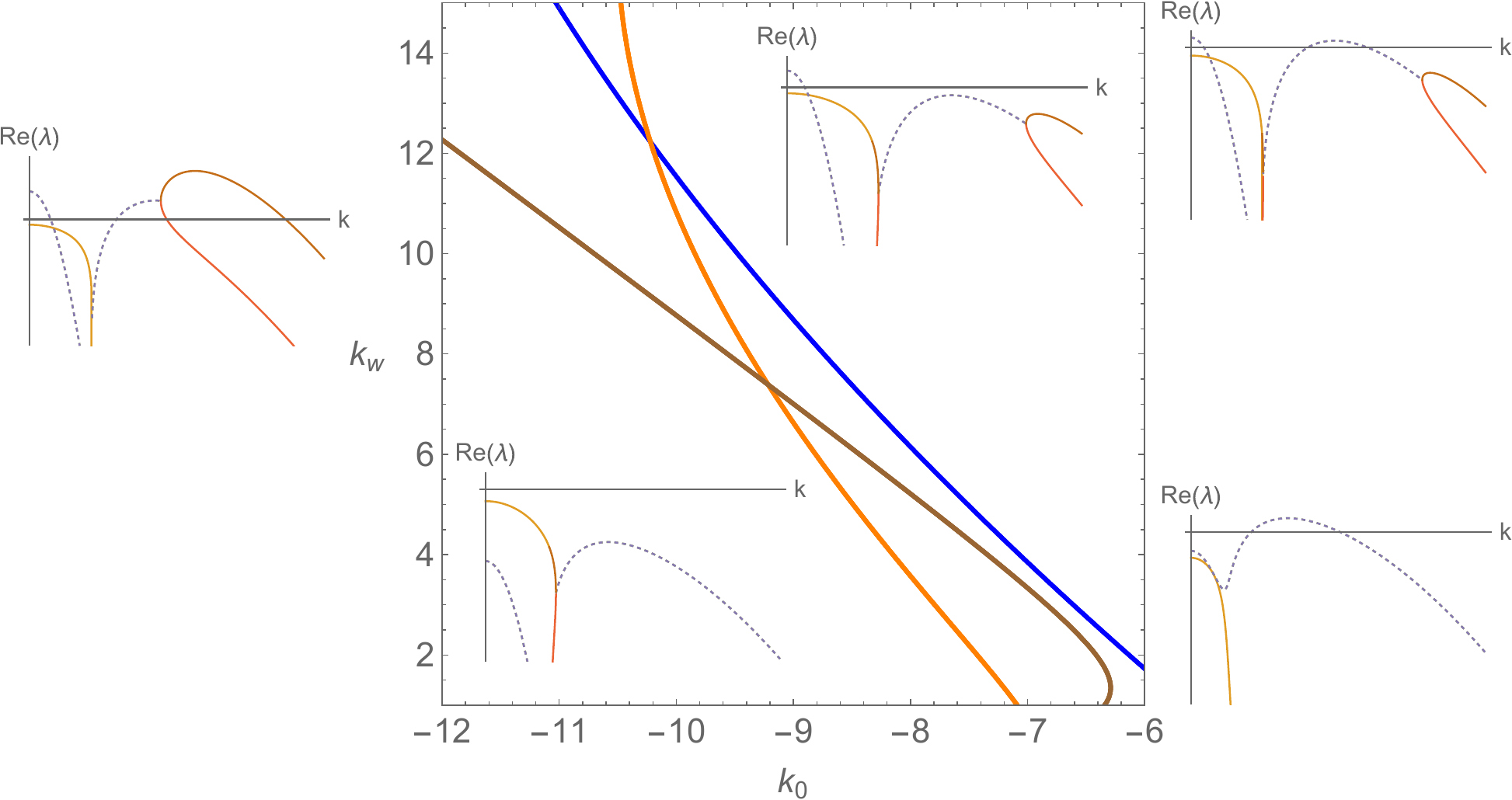}
        		    };
                    \node (text) at (2.25, - 2.0) {\circled{1}};
        		    \node (text) at (6.9, - 1.0) {\circled{1}};
        		    \node (text) at (0.43, 0.35) {\footnotesize \circled{2}};
                    \node (text) at (7.5, 4.2) {\circled{2}};
        		    \node (text) at (- 2.2, 2.9) {\circled{3}};
                    \node (text) at (- 5.2, 2.7) {\circled{3}};
        		\end{tikzpicture}
        		\caption{Bifurcation diagram of system \eqref{eq:fhn1}--\eqref{eq:fhn3} with $f_3\left(u, k_0\right) = k_0 + 2 \, u - u^3$, $a_0 = a_w = 0$, $a_v = a_w = 1$, and $\eps_v = \frac{1}{\tau}$ around $\mbf P$, with the parameter values indicated in \eqref{parametervals}. The brown line is a supercritical Hopf bifurcation of the steady-state, the blue curve represents a subcritical Turing bifurcation, while the orange curve represents a supercritical wave bifurcation according to the sign of $\Re\left(C_1^{[3, 0]}\right)$.}
        		\label{fig:THbifurcationdiagram}
        	\end{figure}
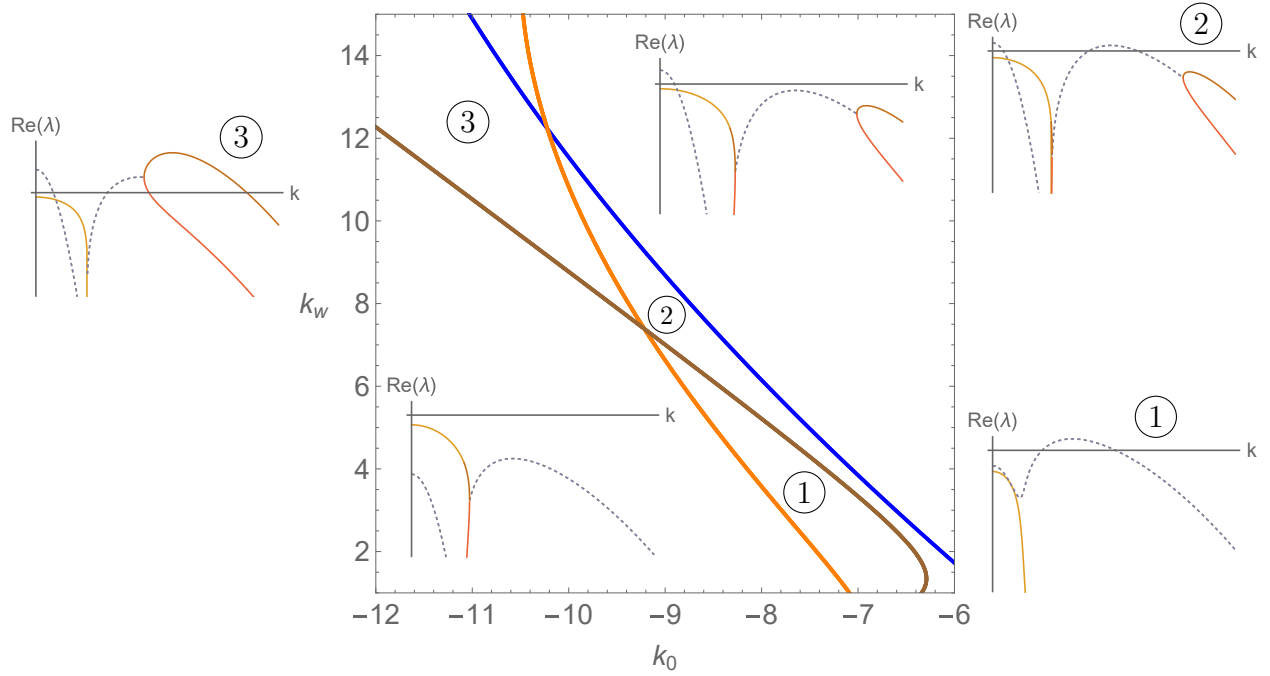

            Looking at Fig.~\ref{fig:THbifurcationdiagram} we observe that $\mbf P$ is stable only when $(k_0, k_w)$ belongs to the bottom-left region. It can become unstable either through a wave bifurcation when moving the parameters to Region 1, or a Hopf bifurcation when moving them to Region 3. Another option is to go through the codimension-two bifurcation point directly to Region 2. However, we find that there is no small amplitude travelling wave orbit in Region 1 beyond the wave bifurcation; for reasons that are explained in Fig. \ref{fig:C3Yang1}. 

            \begin{figure}
                \centering
                \begin{subfigure}[b]{0.5\textwidth}
                    \centering
                    \includegraphics[width = \textwidth]{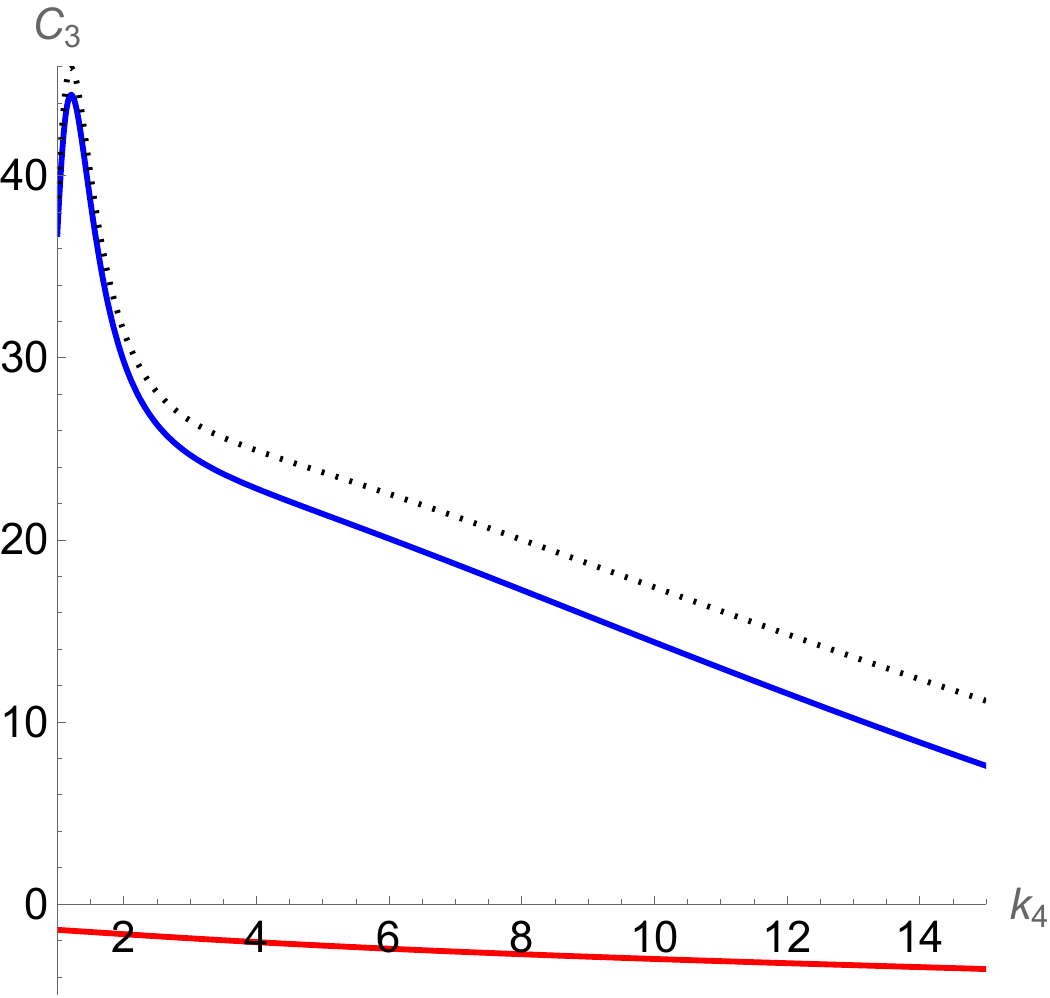}
                    \caption{}
                \end{subfigure}
                \\
                \begin{subfigure}[b]{0.49\textwidth}
                    \centering
                    \includegraphics[width = \textwidth]{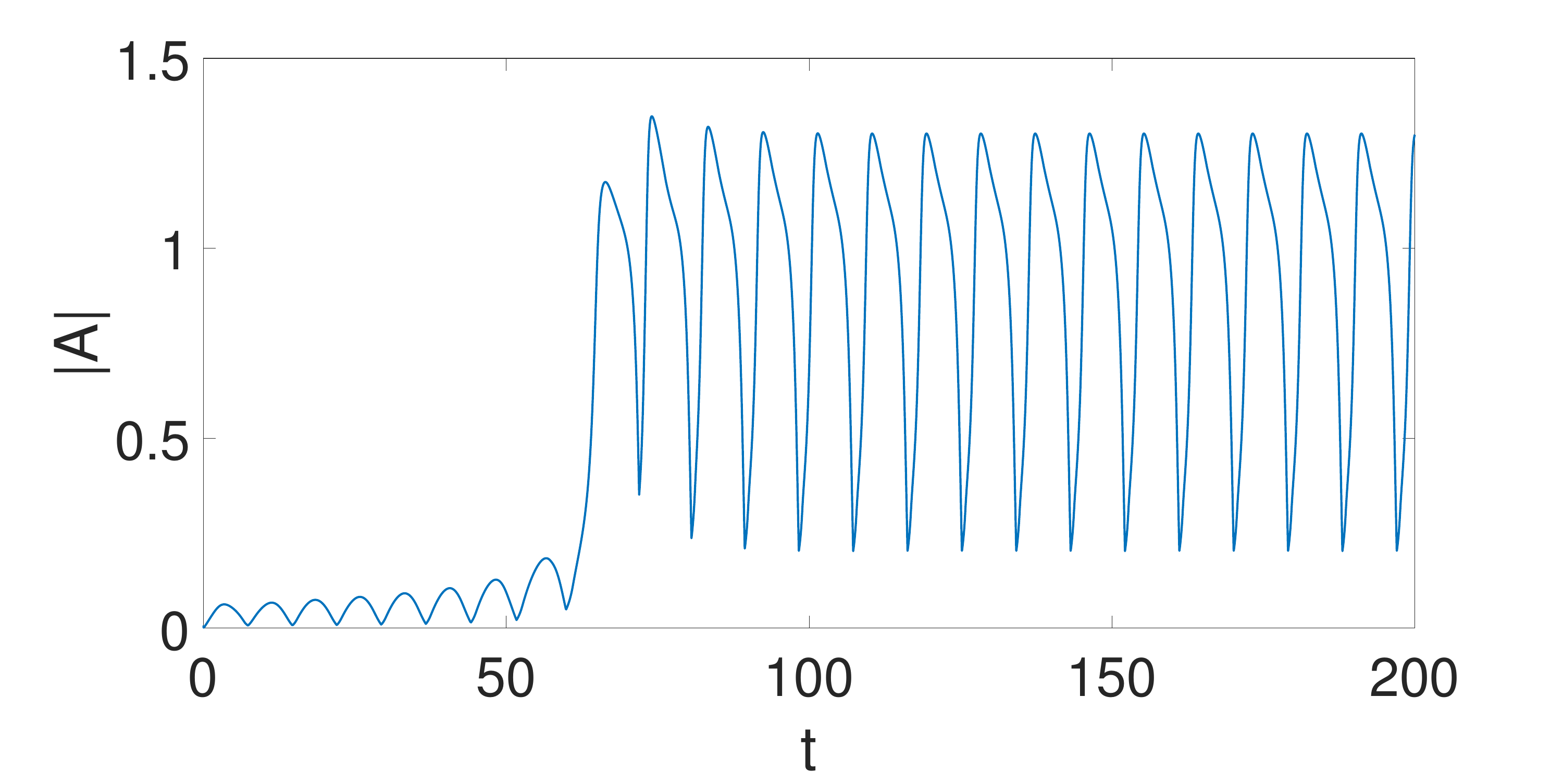}
                    \caption{}
                \end{subfigure}
                \hfill
                \begin{subfigure}[b]{0.49\textwidth}
                    \centering
                    \includegraphics[width = \textwidth]{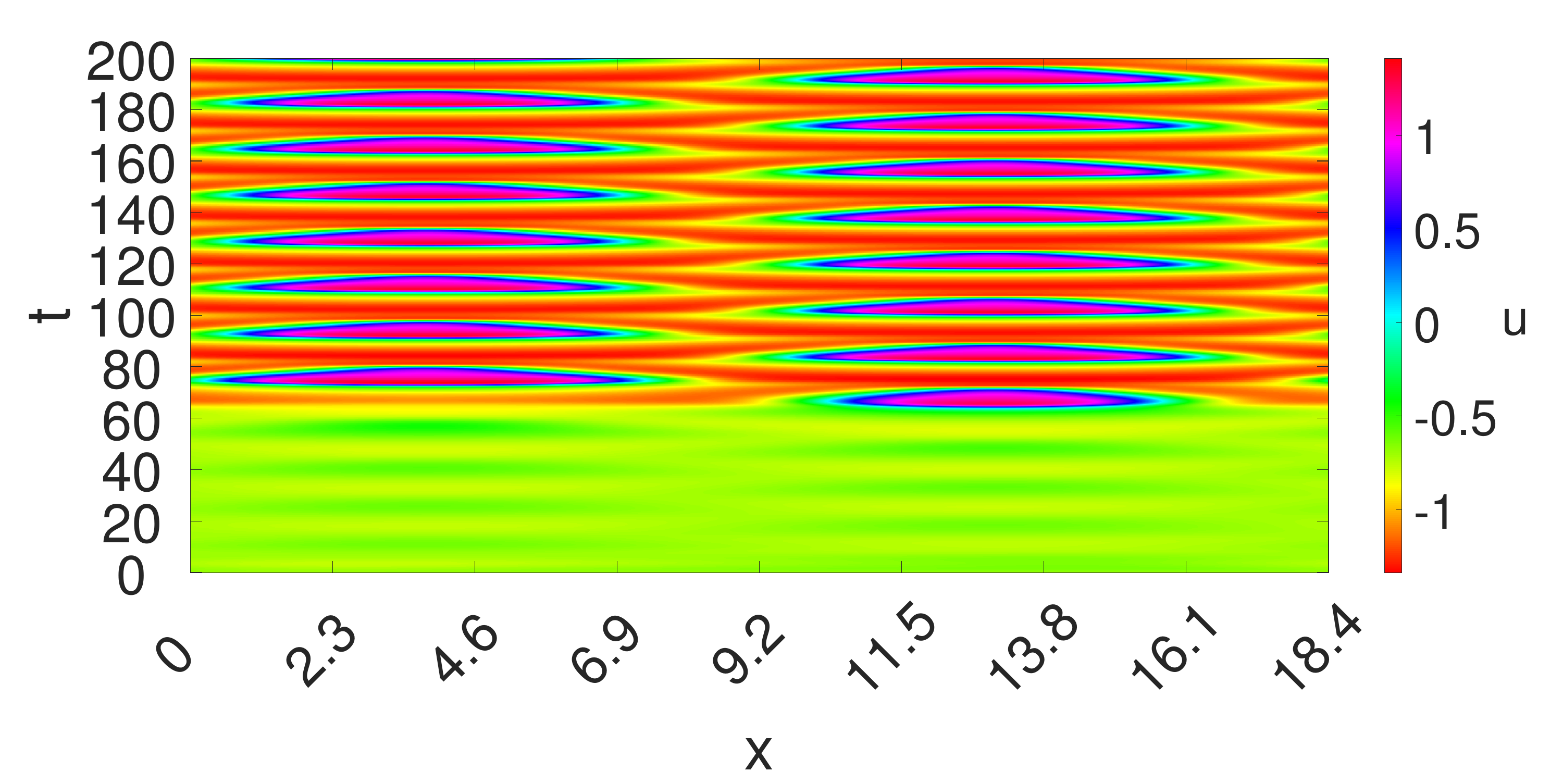}
                    \caption{}
                \end{subfigure}
                \caption{(a) Third-order coefficients along the wave bifurcation curve shown in Fig.~\ref{fig:THbifurcationdiagram}. Colors have the same meaning as in Figure \ref{fig:C3FHN}(a). (b) Evolution of the amplitude $|A|$ in normal-form co-ordinates, and $u$-component of the solution starting from the low-amplitude traveling wave near the wave bifurcation, for $k_0, k_w)) = (- 7.45, 2.2)$ for
                $L=18.402345$ with periodic boundary conditions. (c) Color plot showing the shape of the solution at each time shown in (b).
                }
                \label{fig:C3Yang1}
            \end{figure}
	
            Figure \ref{fig:C3Yang1}(a) shows more details of the criticality of the wave bifurcation by showing the evolution of the appropriate third-order normal-form coefficients along the bifurcation curve depicted in Fig.~\ref{fig:THbifurcationdiagram}. Referring to the bifurcation diagram in Fig.~\ref{fig:knobloch2par} we note that the signs of $\Re\left(C_1^{[3, 0]}\right)$ and $\Re\left(C_1^{[3, 0]} + C_1^{[1, 2]}\right)$ imply that while travelling waves bifurcate supercritically, they are nevertheless unstable. This is confirmed by the numerical simulation in Fig.\ref{fig:C3Yang1}(b) which shows the evolution of the initial condition of the supercritical travelling wave predicted by the normal form, using periodic boundary conditions. Here, we see this unstable state clearly develops an instability that eventually leads to a jump to a large amplitude wave. The time-trace in Fig.\ref{fig:C3Yang1}(c) shows that this large-amplitude state is actually a standing wave, not a travelling wave. Only the $u$-component is shown, the other two components behave similarly.
 
        \paragraph{Codimension-two Bautin bifurcation.}
            We consider again the Yang {\em et al} model but with fixed parameter values:
            \begin{align}
                \left(k_w, \tau, D_1, D_2, D_3\right) = (3, 10, 1, 1, 60). \label{thirdparval}
            \end{align}
            For these values, we now find the two-parameter bifurcation diagram shown in Figure \ref{fig:SW-critchange}. The figure indicates that there is a change of criticality of the Turing bifurcation and the wave bifurcation for these values. The change in criticality of the Turing bifurcation was computed using the methods in \cite{TOMS}. Our focus here is instead the change in criticality of the wave bifurcation. 
     
            Figure \ref{fig:C3Yang2} shows the variation of the cubic coefficients of the normal form along the wave bifurcation curve. Here, it is clear that the codimension-two point corresponds to a change of sign of  $\Re(C_1^{[3, 0]} + C_1^{[1, 2]})$, which shows that the standing wave changes from being sub- to super-critical. The specific codimension-two point occurs at (-6.685534, 8.67174).
            \begin{figure}
                \centering
                \begin{tikzpicture}
        		    \node (image) at (0,0) {
        		    \includegraphics[width = \textwidth]{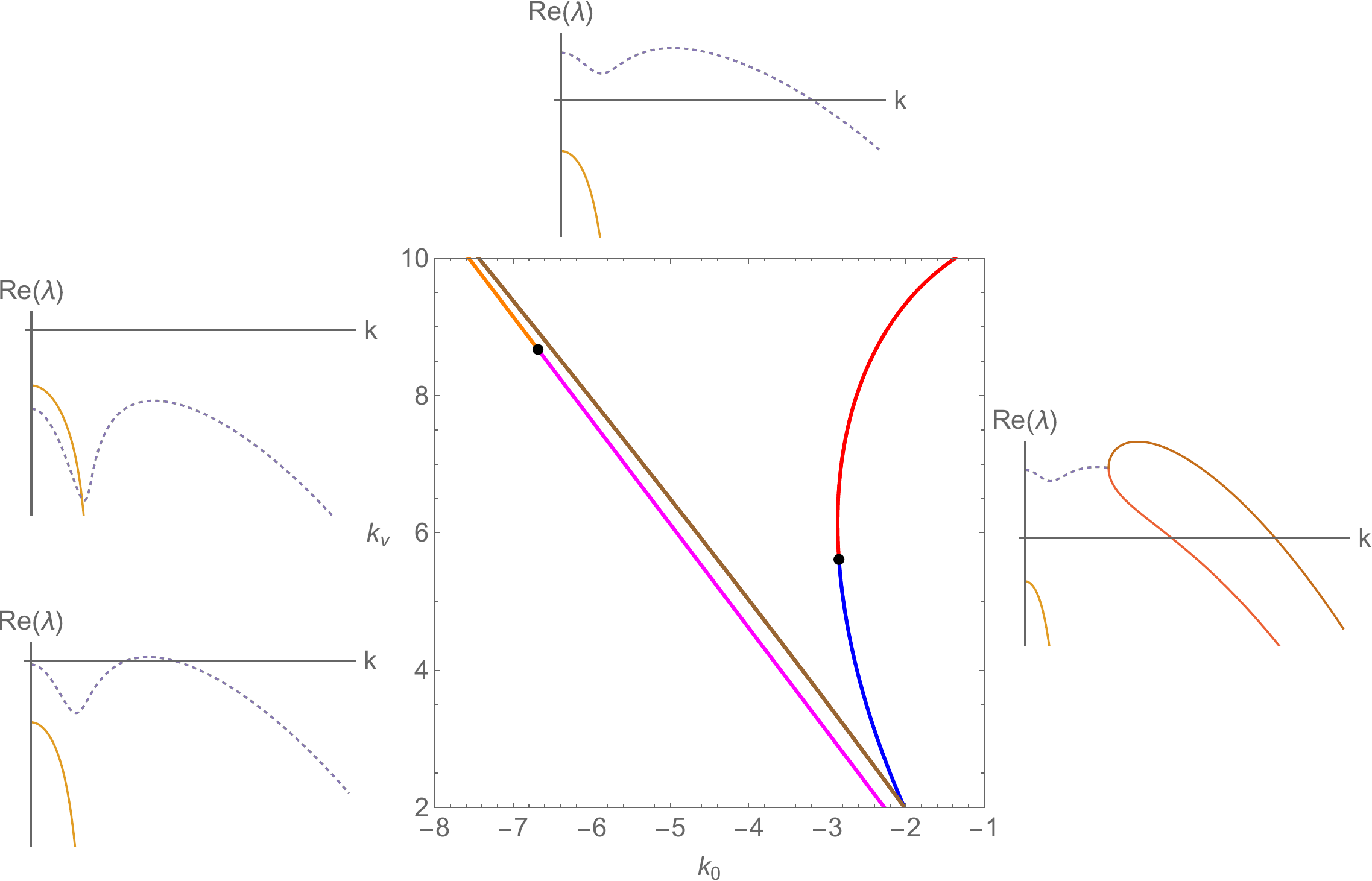}
        		    };
                    \node (text) at (- 5.25, 1.8) {\circled{1}};
        		    \node (text) at (- 1.3, - 2.0) {\circled{1}};
        		    \node (text) at (- 5.25, - 2.2) {\circled{2}};
                    \node (text) at (0.0, 0.2) {\footnotesize \circled{2}};
        		    \node (text) at (1.0, 5.0) {\circled{3}};
                    \node (text) at (0.8, 1.0) {\circled{3}};
                    \node (text) at (6.8, 0.0) {\circled{4}};
                    \node (text) at (2.7, - 1.0) {\circled{4}};
                    \draw [-stealth](- 0.2, 0.05) -- (- 0.7, - 0.3);
        		\end{tikzpicture}
                \caption{Bifurcation curves of model \eqref{eq:fhn1}--\eqref{eq:fhn3} when $f_3\left(u, k_0\right) = k_0 + 2 \, u - u^3$, $a_0 = a_w = 0$, $a_v = a_w = 1$, and $\eps_v = \frac{1}{\tau}$ around $\mbf P$, for the parameter values given in \eqref{thirdparval}. The colors have the same meaning as in Figure \ref{fig:THbifurcationdiagram}, except that the magenta (resp.~red) curve represents a subcritical wave (resp.~supercritical Turing) bifurcation. Also, the black dots represent codimension-two points where the criticality of the bifurcation changes.}
                \label{fig:SW-critchange}
            \end{figure}
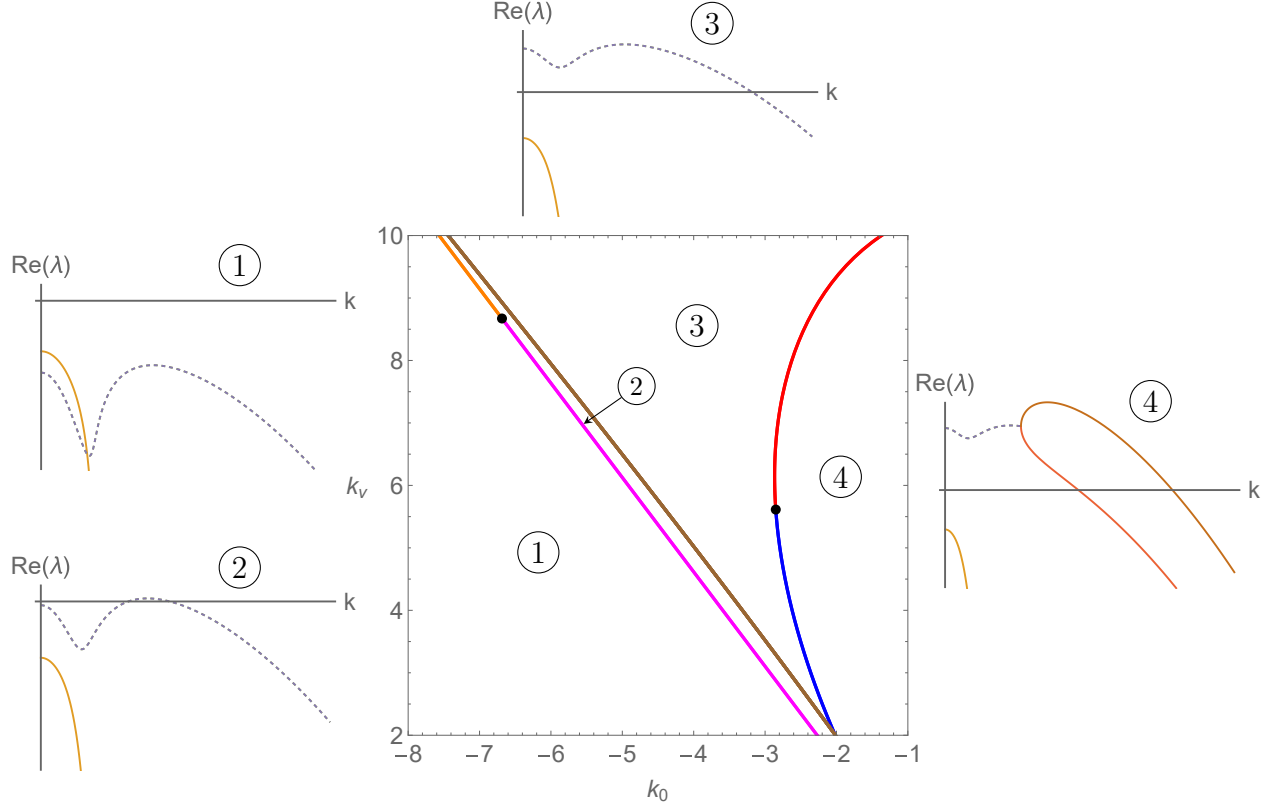
            \begin{figure}
                \centering
                \begin{subfigure}[b]{0.4\textwidth}
                    \centering
                    \includegraphics[width = \textwidth]{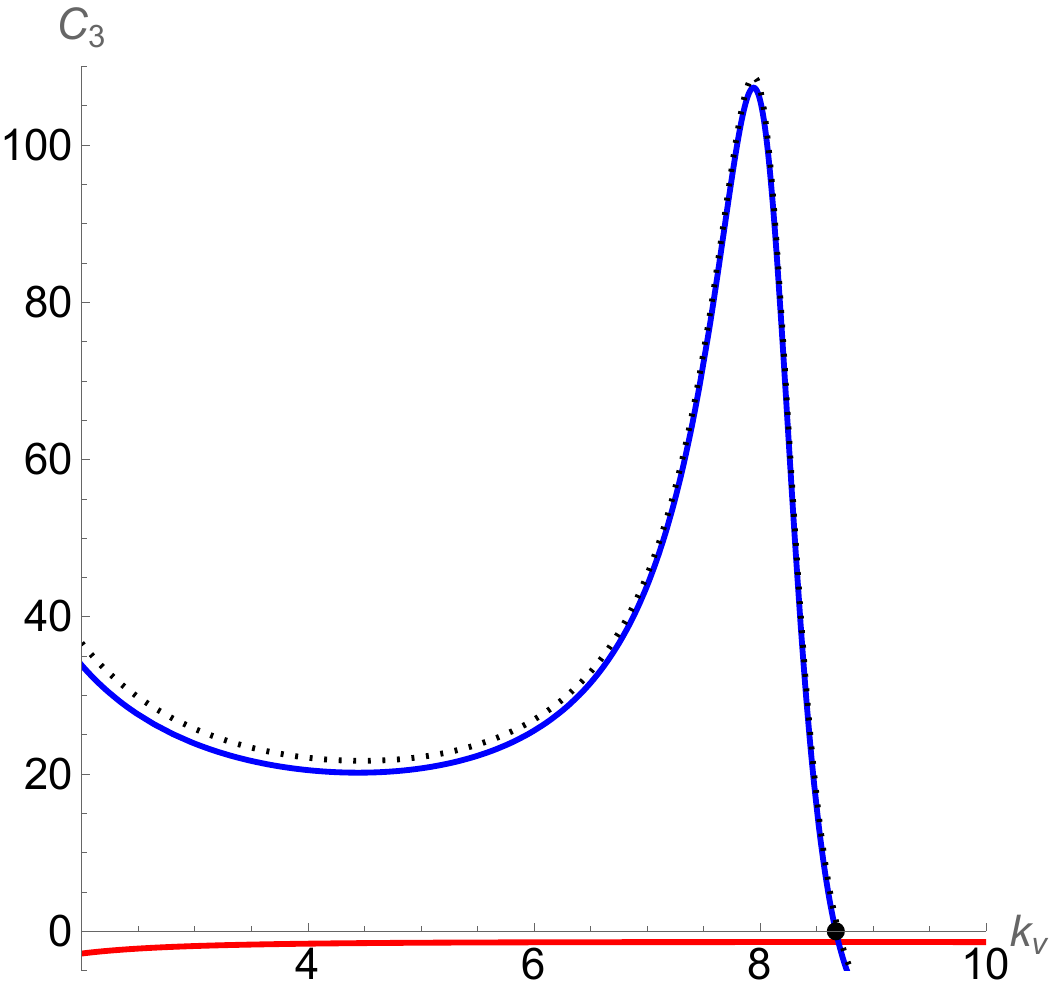}
                    \caption{}
                \end{subfigure}
                \hfill
                \begin{subfigure}[b]{0.4\textwidth}
                    \centering
                    \includegraphics[width = \textwidth]{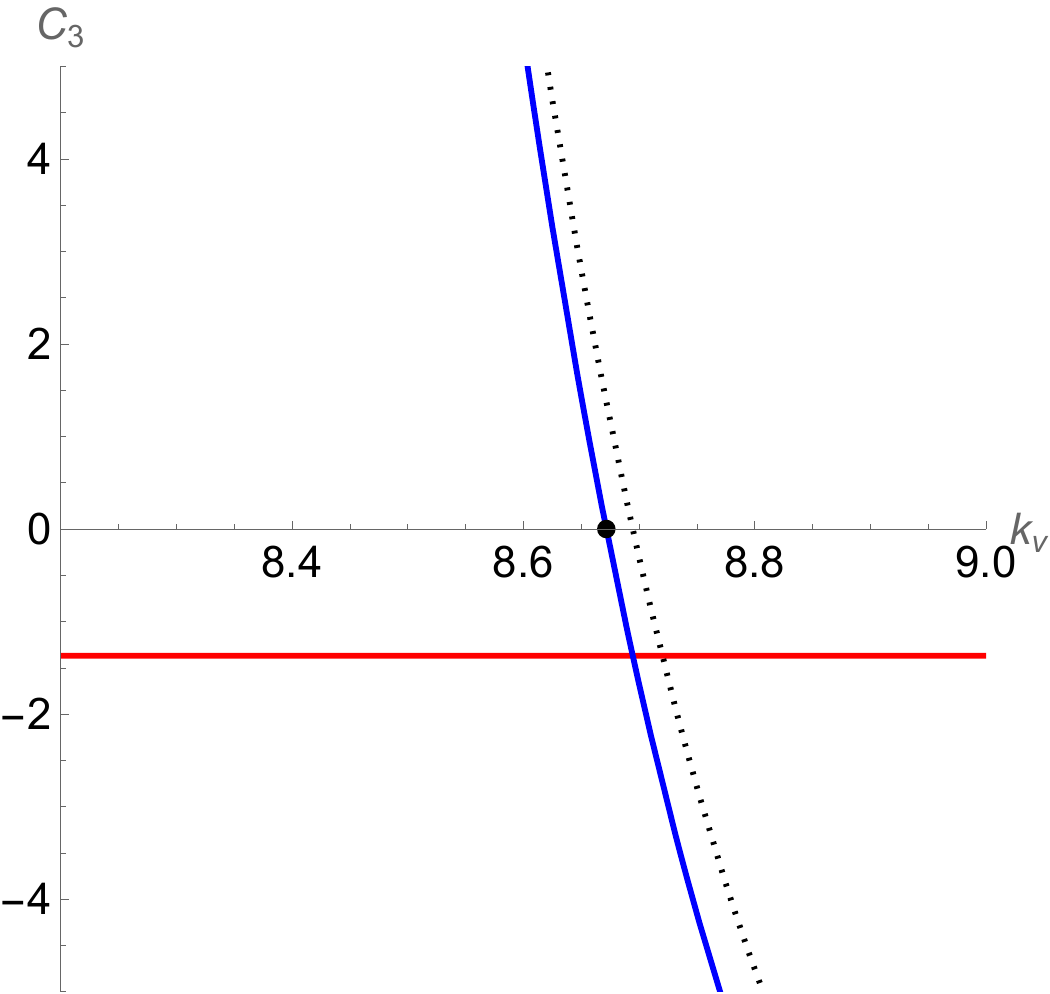}
                    \caption{}
                \end{subfigure}
                \caption{(a) Third-order coefficients along the wave bifurcation curve shown in Figure \ref{fig:SW-critchange}. (b) Zoomed-in version of the third-order coefficients around the codimension-two point where $\Re\left(C_1^{[3, 0]} + C_1^{[1, 2]}\right)$ changes sign. Colors have the same meaning as in Figure \ref{fig:C3FHN}(a).}
                \label{fig:C3Yang2}
            \end{figure}

            Figure \ref{fig:loglog} shows a zoom of the two-parameter bifurcation diagram close to the Bautin point. The inset to the figure shows the agreement between the numerically computed standing wave close to the bifurcation point and that predicted by the roots of \eqref{eq:rootR1} for the normal form. Note the good agreement between the two curves. Also plotted in the figure is the normal-form prediction \eqref{eq:standingfold} for the location of the fold in two parameters, up to order five, together with the numerical evaluation of the fold points. Note that the normal-form coefficients are evaluated only at the codimension-two point, so the theory should only predict the leading-order term in a Taylor expansion of the fold curve. Such an agreement can be observed in panels \ref{fig:loglog}(b), (c) which plot the difference between the fold and the wave bifurcation curves on linear and log-log scales respectively.

            \begin{figure}
                \centering
                \begin{subfigure}[c]{0.45\textwidth}
                    \centering
                    \includegraphics[width = \textwidth]{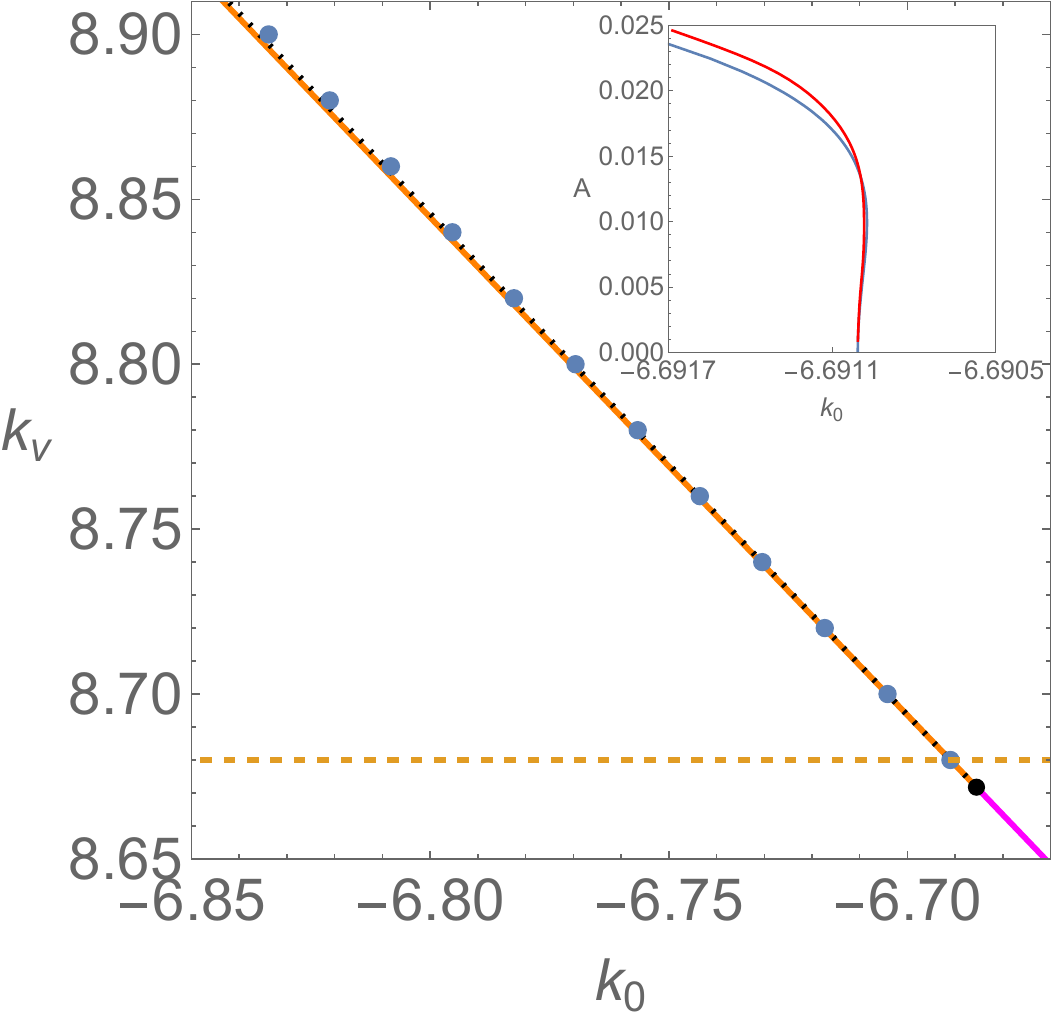}
                    \caption{}
                \end{subfigure}
                \begin{subfigure}[c]{0.45\textwidth}
                    \centering
                    \includegraphics[width = \textwidth]{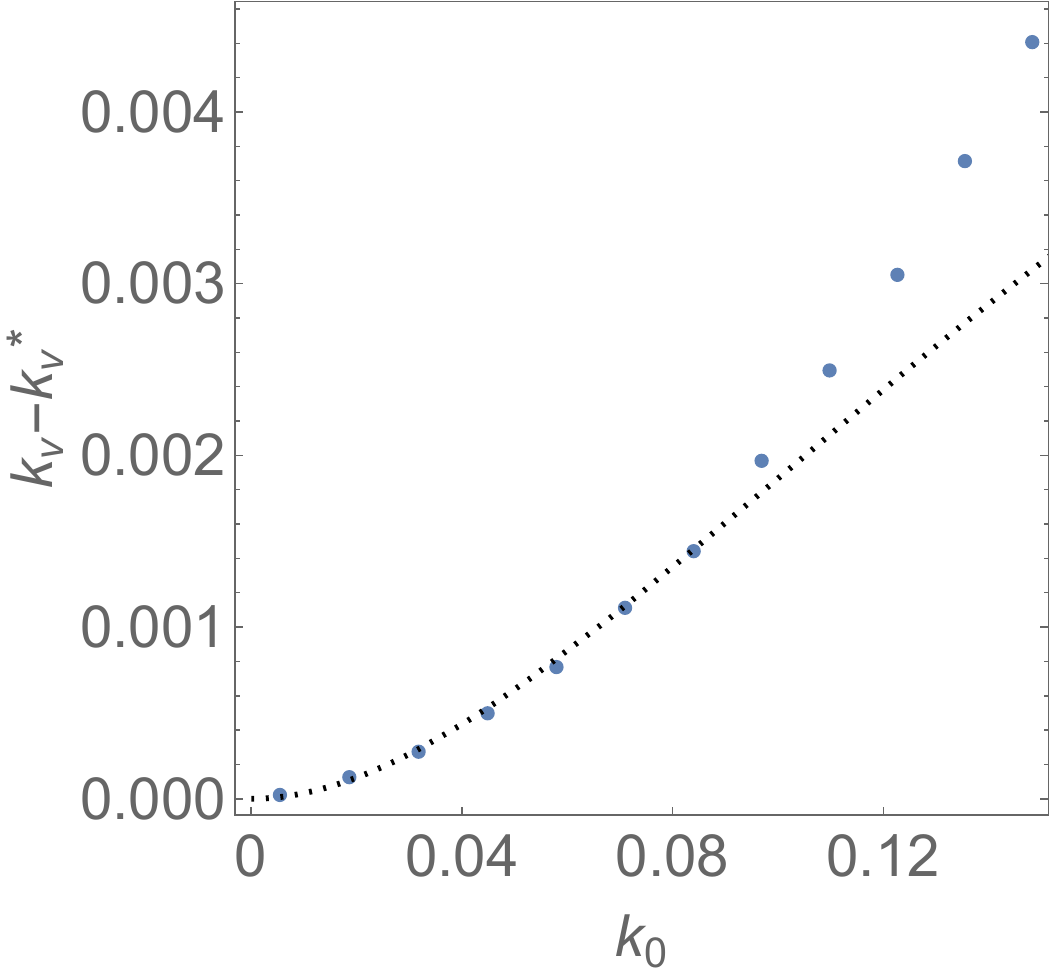}
                    \caption{}
                \end{subfigure}           
    
                \begin{subfigure}[c]{0.45\textwidth}
                    \centering
                    \includegraphics[width = \textwidth]{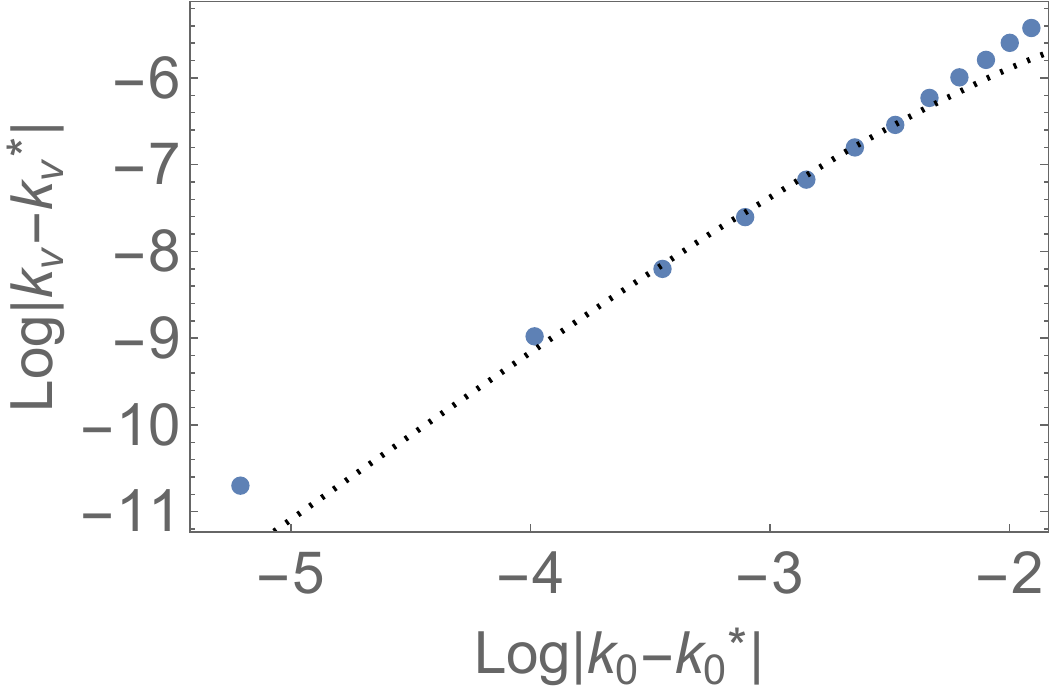}
                    \caption{}
                \end{subfigure}
                \caption{(a) Zoom of bifurcation diagram in Figure \ref{fig:SW-critchange} close to  the codimension-two Bautin bifurcation point $\left(k_0^*, k_v^*\right) = (- 6.68553, 8.67174)$. The dotted line corresponds to the fold curve predicted by \eqref{eq:standingfold}, while the blue circles are fold points obtained through numerical continuation. The inset shows the one-parameter bifurcation diagram of standing waves in $k_0$ close to the bifurcation point, at the value of $k_v$ depicted by the yellow dotted line (red curve numerics, blue normal-form prediction). (b) Same information showing the distance $k_0 - k_0^*$, where $k_0^* = - 6.68553$, and the distance $k_v - k_v^*$, where $k_v^*$ is the $k_v$-value of the wave bifurcation curve. (c) The same information on a log-log plot.}
                \label{fig:loglog}
            \end{figure}

            To see the shape of the bifurcating patterns, we turn to direct numerical integration. First, we integrate the system fixing the values $\left(k_0, k_v\right) = (- 7.36, 9.7)$, with $L = 17.046381$, using a random perturbation around $\mbf P$ as an initial condition, and considering periodic boundary conditions; the solution is found to converge to the stable low-amplitude traveling-wave solution shown for $u$ in Figures \ref{fig:u-TW}. On the other hand, a similar integration but with Neumann boundary conditions, results in the the stable standing-wave pattern shown in Figure \ref{fig:u-SW}. In each figure, we show only the $u$-component; patterns for the other variables look similar.
            \begin{figure}
                \centering
                \begin{subfigure}[b]{0.49\textwidth}
            	    \centering
            		\includegraphics[width = \textwidth]{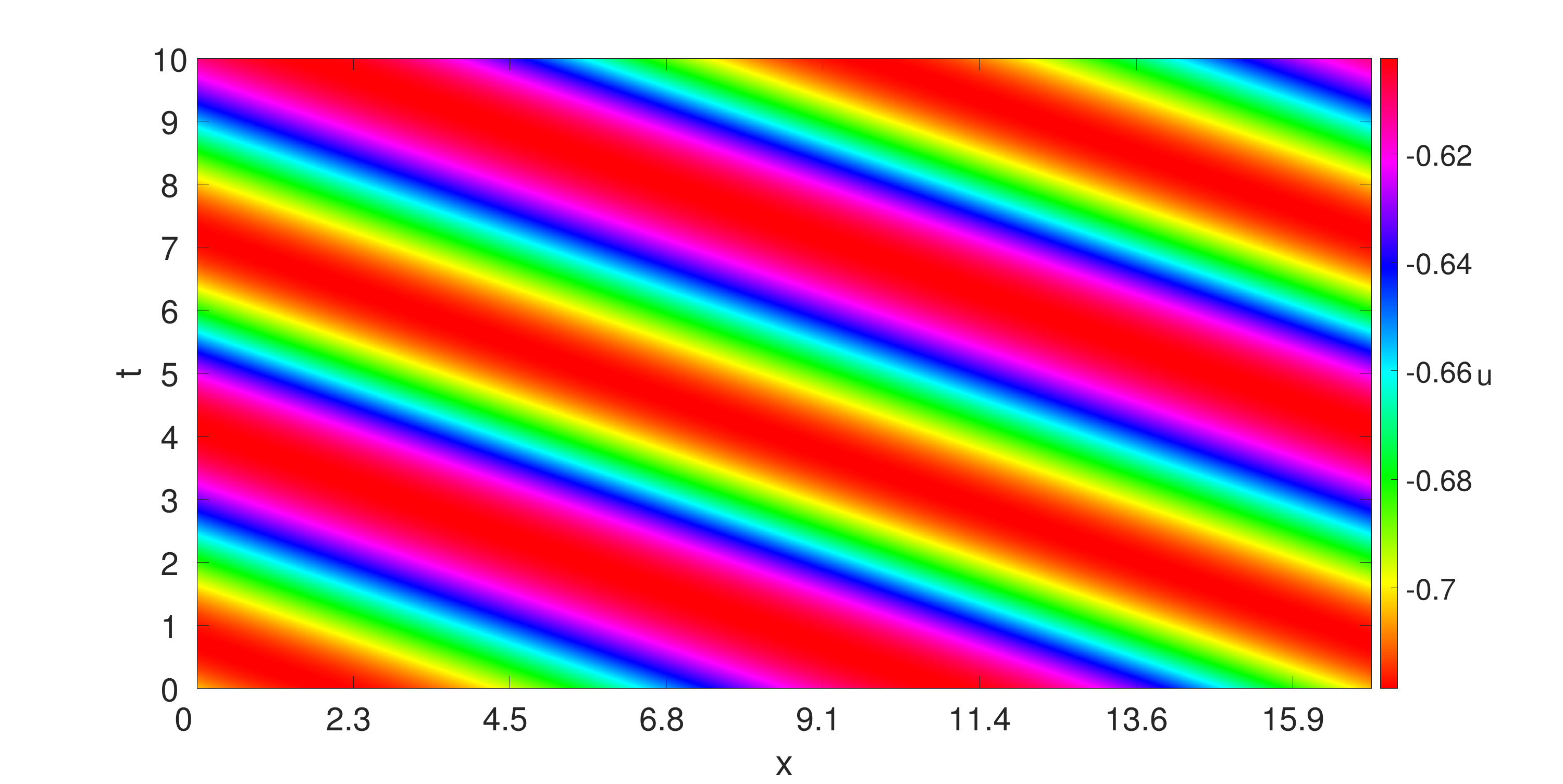}
                    \caption{}
                    \label{fig:u-TW}
                \end{subfigure}
                \hfill
                \begin{subfigure}[b]{0.49\textwidth}
            	    \centering
                    \includegraphics[width = \textwidth]{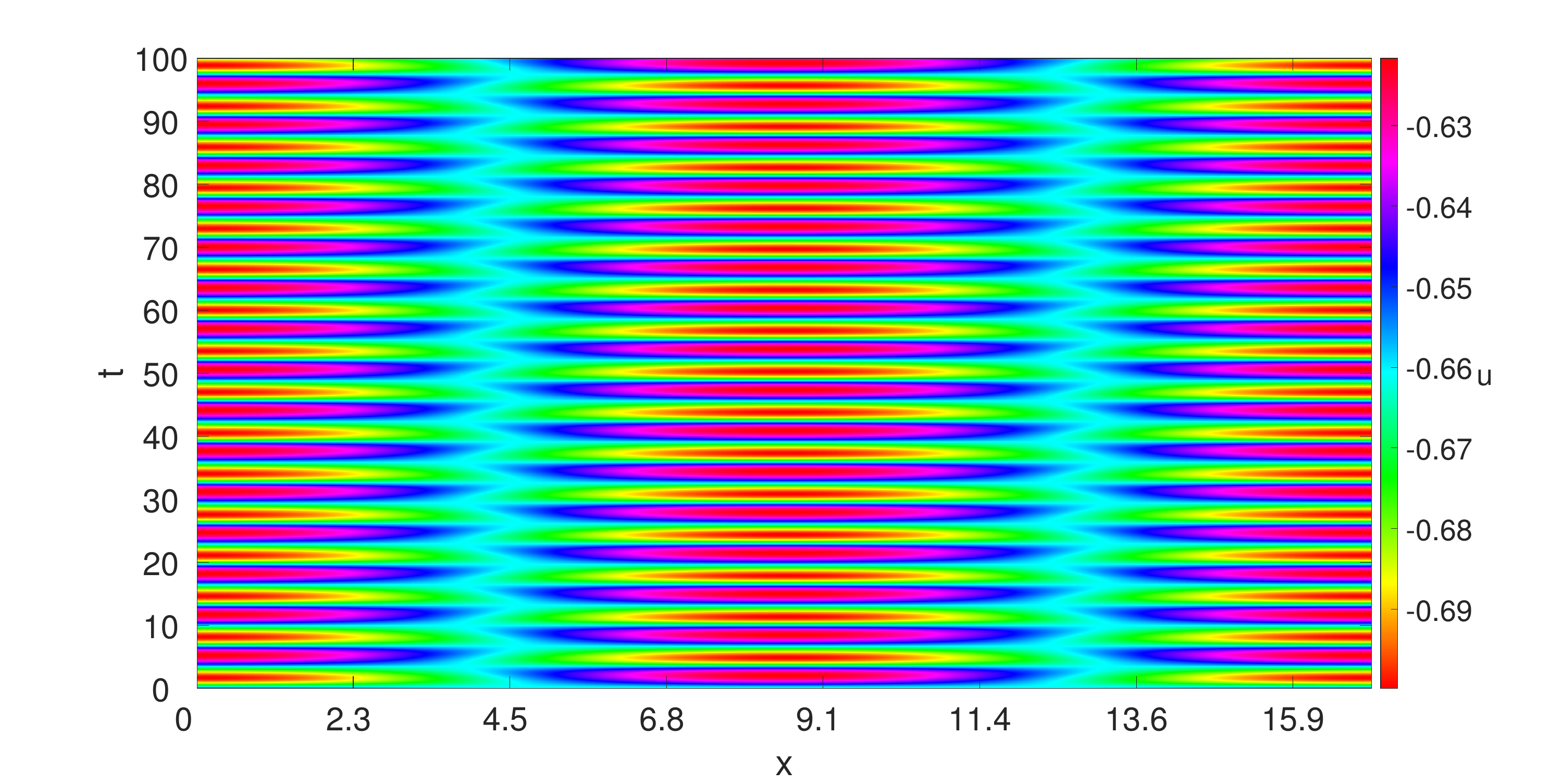}
                    \caption{}
                    \label{fig:u-SW}
     	      \end{subfigure}
                \caption{Stable low-amplitude waves around $\mbf P$ for the system \eqref{eq:fhn1}--\eqref{eq:fhn3}, with $f_3\left(u, k_0\right) = k_0 + 2 \, u - u^3$, $a_0 = a_w = 0$, $a_v = a_w = 1$, and $\eps_v = \frac{1}{\tau}$, for the parameter values given by \eqref{thirdparval}, and $\left(k_0, k_v\right) = (-7.36, 9.7)$ obtained through direct numerical simulation. (a) Travelling wave with periodic boundary conditions. (b) Standing wave with Neumann boundary conditions. }
            \end{figure}

        \subsection{A five-component example}    
            Finally, to illustrate how our analysis works in systems with an arbitrary number of components $n\geq 3$, we present results for a five-component model studied by Yang et al \cite{Yang2}. Their interest was to model a chemical reactor that is happening in two layers with different pairs of reactants in each, (whose concentrations are represented respectively by $(u, v)$ and $(y, z)$ below, in our notation). There is a within-layer diffusion of each reactant as well as a separate species (with concentration $w$) that diffuses in the interface between the two layers.
        
            The concentrations of the five reactants in one spatial dimension are assumed \cite{Yang2} (rewritten in a notation that is consistent with the present paper) to obey the reaction-diffusion system \cite{Yang2}
            \begin{align}
                \frac{\partial u}{\partial t} &= \frac{1}{\varepsilon_1} \left(u - u^2 - f_1 \, v \, \frac{u - q_1}{u + q_1}\right) - \frac{u - w}{\delta_1} + D_1 \, \nabla^2 u, \notag
                \\
                \frac{\partial v}{\partial t} &= u - v + D_2 \, \nabla^2 v, \notag
                \\
                \frac{\partial w}{\partial t} &= \frac{u - w}{\delta_1} + \frac{y - w}{\delta_2} + D_3 \, \nabla^2 w, \label{2CL}
                \\
                \frac{\partial y}{\partial t} &= \frac{1}{\varepsilon_2}\left(y - y^2 - f_2 \, z \, \frac{y - q_2}{y + q_2}\right) - \frac{y - w}{\delta_2} + D_4 \, \nabla^2 y, \notag
                \\
                \frac{\partial z}{\partial t} &= y - z + D_5 \, \nabla^2 z. \notag
            \end{align}
            The parameters $\eps_{1, 2}$ and $\delta_{1, 2}$ represent timescales of reactions; $D_i$, $i = 1, \ldots, 5$ are diffusion constants; $f_{1, 2}$ represent the ratio of cross-catalytic and autocatalytic saturation terms; and $q_{1, 2}$ code details of the cross-catalytic saturation. Each of these parameters is assumed to be positive. This system has, at most, nine homogeneous steady states. They are the solutions to a high-degree polynomial equation and cannot be obtained in a closed form. Nevertheless, following \cite{Yang2} we set default values
            \begin{equation}
                q_1 = q_2 = 0.01, \qquad \delta_i = 2 \, \varepsilon_i, \quad i = 1, 2 \label{5D_default1}
            \end{equation}
            for simplicity.
    
            Furthermore, to be concrete we also take one of the parameter sets considered in \cite{Yang2}
            \begin{align}
                \left(\varepsilon_1, f_2, \varepsilon_2, D_1, D_2, D_3, D_5\right) = (0.215, 0.65, 0.5, 0.1, 0.1, 0.1, 100.0), \label{2CLparval}
            \end{align}
            and treat $f_1$ and $D_4$ as free parameters. In this setting, we focus on the steady state with coordinates
            \begin{align*}
                (u, v, w, y, z) = (0.420482, 0.420482, 0.414265, 0.399805, 0.399805) \quad \text{when } \left(f_1, D_4\right) = (0.6, 4).
            \end{align*}
            Figure \ref{fig:2Clbifdiag} shows the result of using the code in \cite{criticality-wave} to compute the Turing and wave bifurcation curves of the system at one particular homogeneous steady state, together with the criticality of the bifurcations.
            \begin{figure}
                \centering
                \begin{tikzpicture}
            		\node (image) at (0,0) {
            		\includegraphics[width = \textwidth]{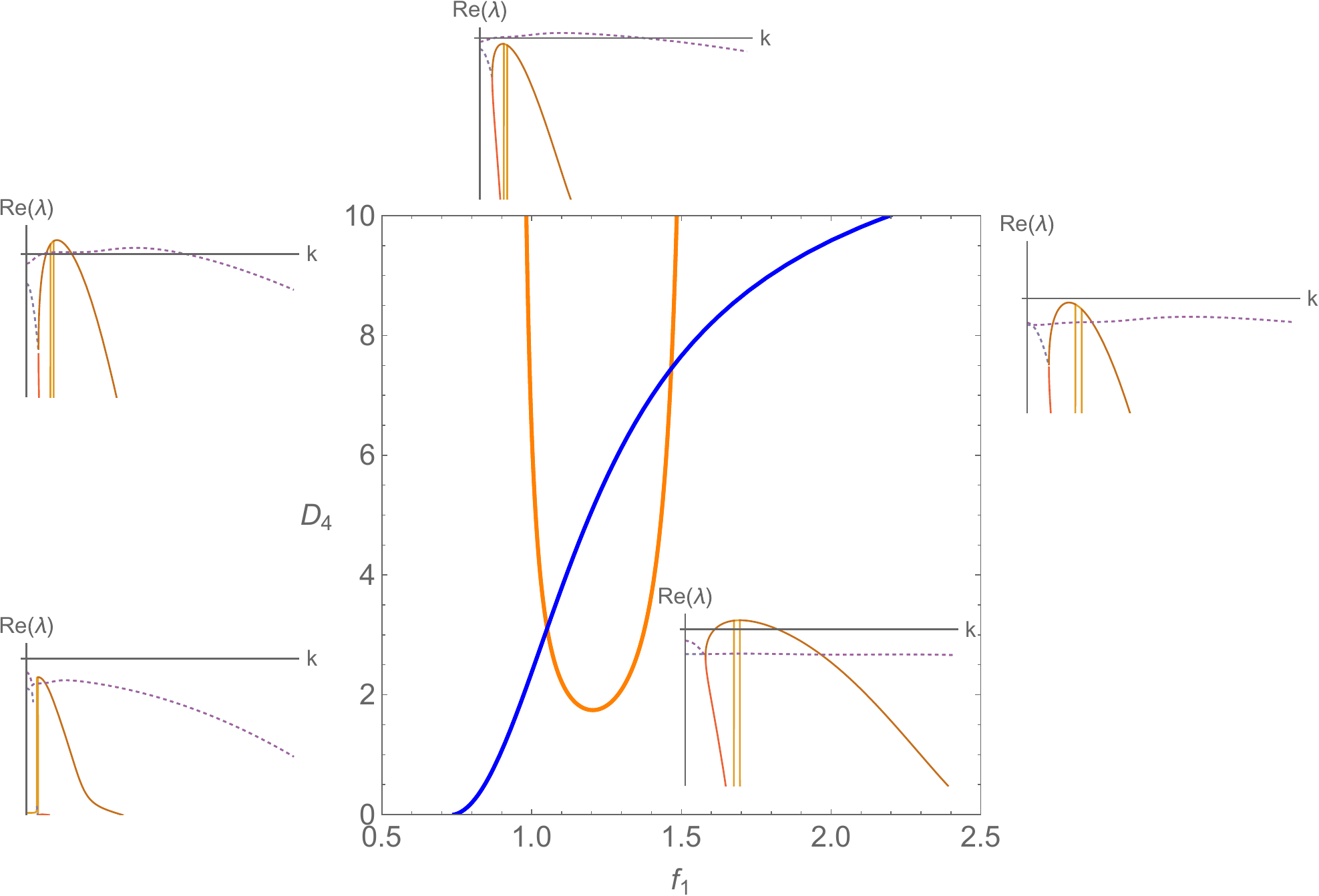}
            		};
            		\node (text) at (- 2.5, - 1.0) {\circled{1}};
                    \node (text) at (- 5.5, - 2.2) {\circled{1}};
            		\node (text) at (1.0, 2.4) {\circled{2}};
                    \node (text) at (6.5, 2.5) {\circled{2}};
            		\node (text) at (- 0.75, 1.7) {\circled{3}};
                    \node (text) at (0.0, 5.55) {\circled{3}};
            		\node (text) at (- 0.7, - 1.7) {\circled{4}};
            		\node (text) at (- 6.0, 3.0) {\circled{4}};
            	\end{tikzpicture}
                \caption{Bifurcation curves of system \eqref{2CL}, for the parameter values given by \eqref{2CLparval}. As usual, colors have the same meaning as in Figure \ref{fig:bifdiag} and the images accompanying the diagram correspond to the dispersion relations in each of the regions of the diagram.}
                \label{fig:2Clbifdiag}
            \end{figure}
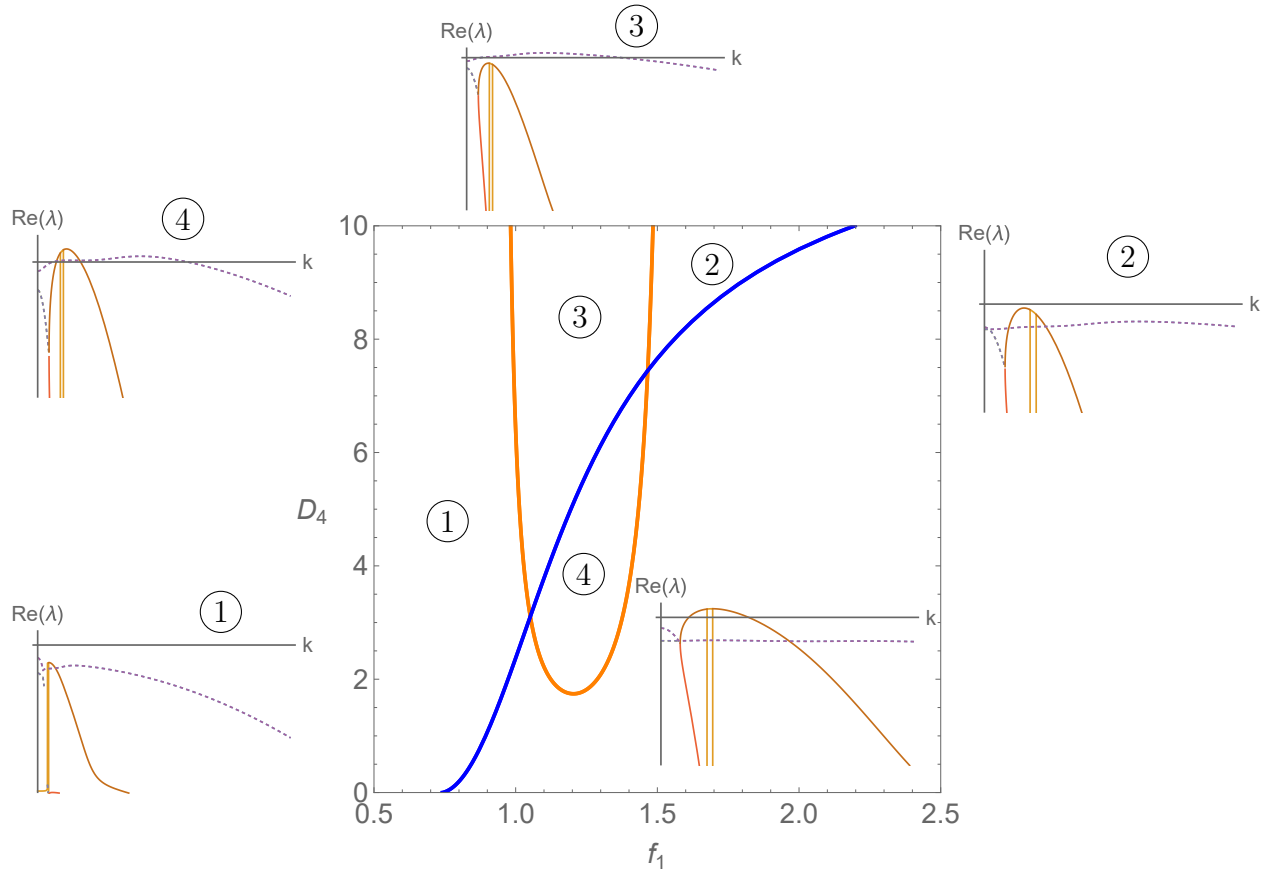
            \begin{figure}
                \centering
                \includegraphics[width = 0.4\textwidth]{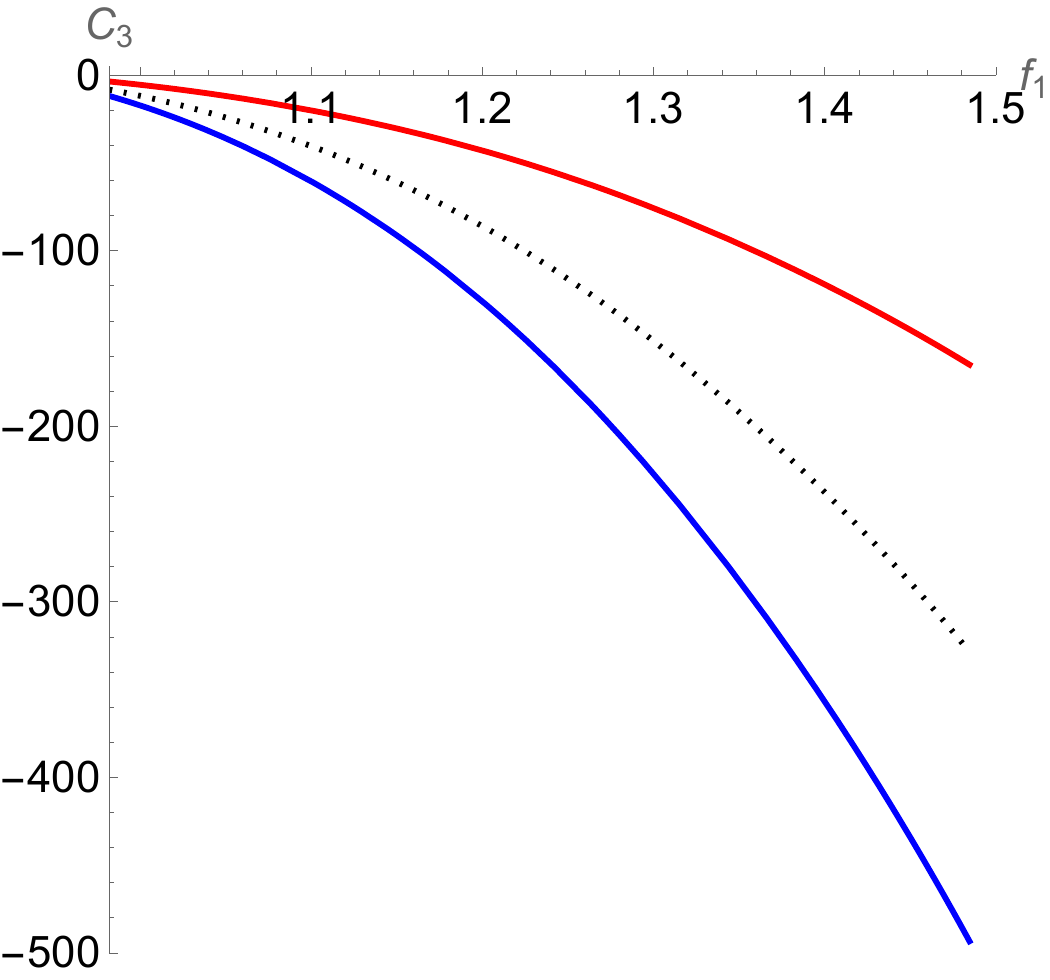}
                \caption{Third-order coefficients along the wave bifurcation curve shown in Figure \ref{fig:2Clbifdiag}. Colors have the same meaning as in Figure \ref{fig:C3Yang1}.}
                \label{fig:C32CL}
            \end{figure}
    
            So, as the authors in \cite{Yang2} suggested, we can indeed find a curve in the parameter plane of both Turing and wave bifurcations. The Turing bifurcation turns out to be subcritical --- see Fig.~\ref{fig:C32CL} --- for all parameter values considered in the diagram, the wave bifurcation is supercritical for both travelling and standing waves. Furthermore, for all the parameters along this curve, we note $C_1^{[1, 2]} < C_1^{[3, 0]}$, which means that according to Fig.~\ref{fig:knobloch2par}, the traveling wave is the stable solution with periodic boundary conditions. 
    
            To show how the theory we have developed applies to this system, we present the results of the integration of the system from a random perturbation to the homogeneous steady state when $\left(f_1, D_4\right) = (1, 9)$, which is in Region 3. If we consider periodic boundary conditions, we get the pattern shown in Figure \ref{fig:2CL-SW-uvw}(a). Only the $u$-component is depicted; the pattern for the other components is qualitatively similar.
            
            \begin{figure}
                \centering
              \begin{subfigure}[b]{0.49\textwidth}
                    \centering
                    \includegraphics[width = \textwidth]{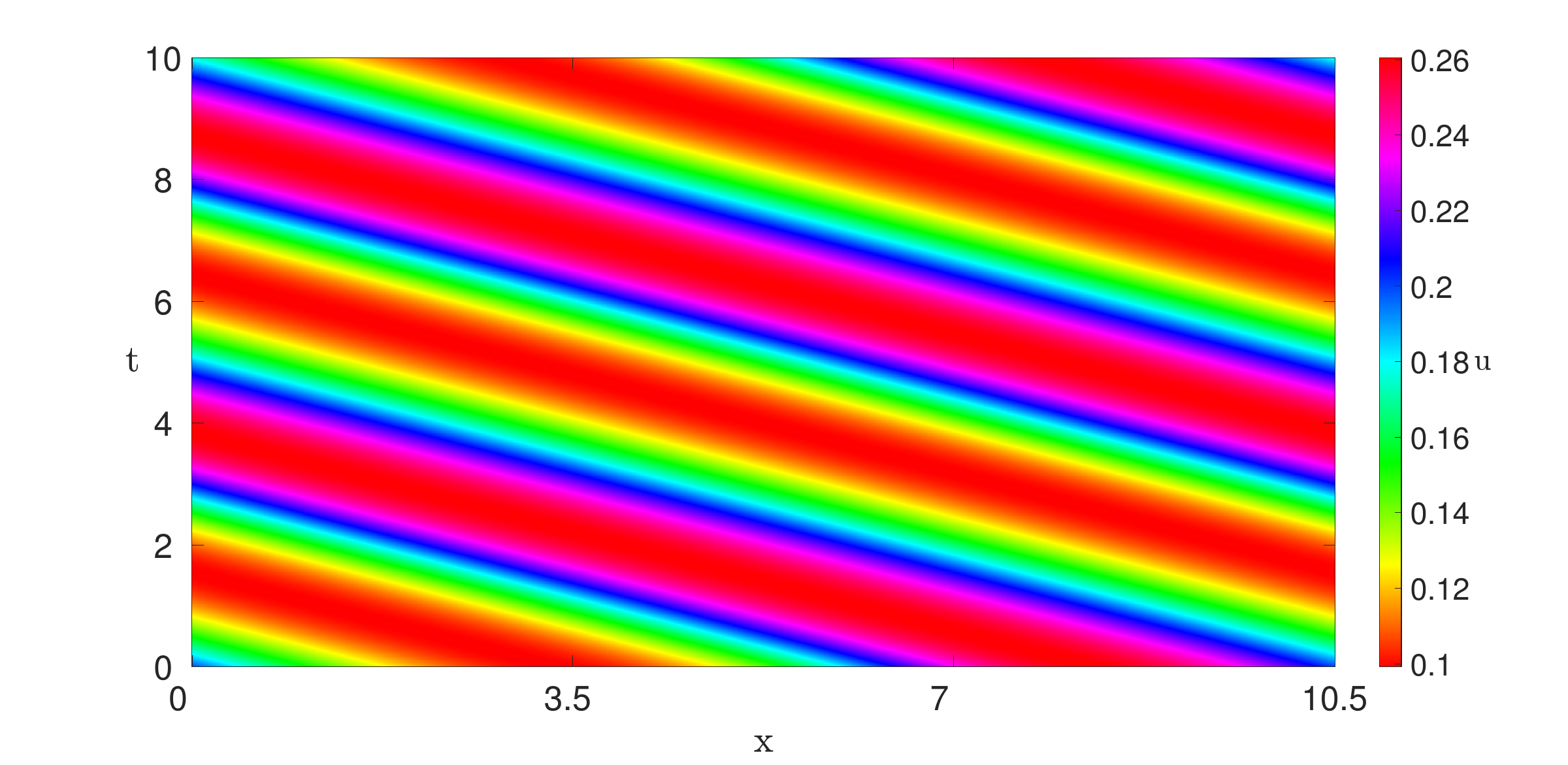}
                    \caption{}
                \end{subfigure}
                \hfill
                \begin{subfigure}[b]{0.49\textwidth}
                    \centering
                    \includegraphics[width = \textwidth]{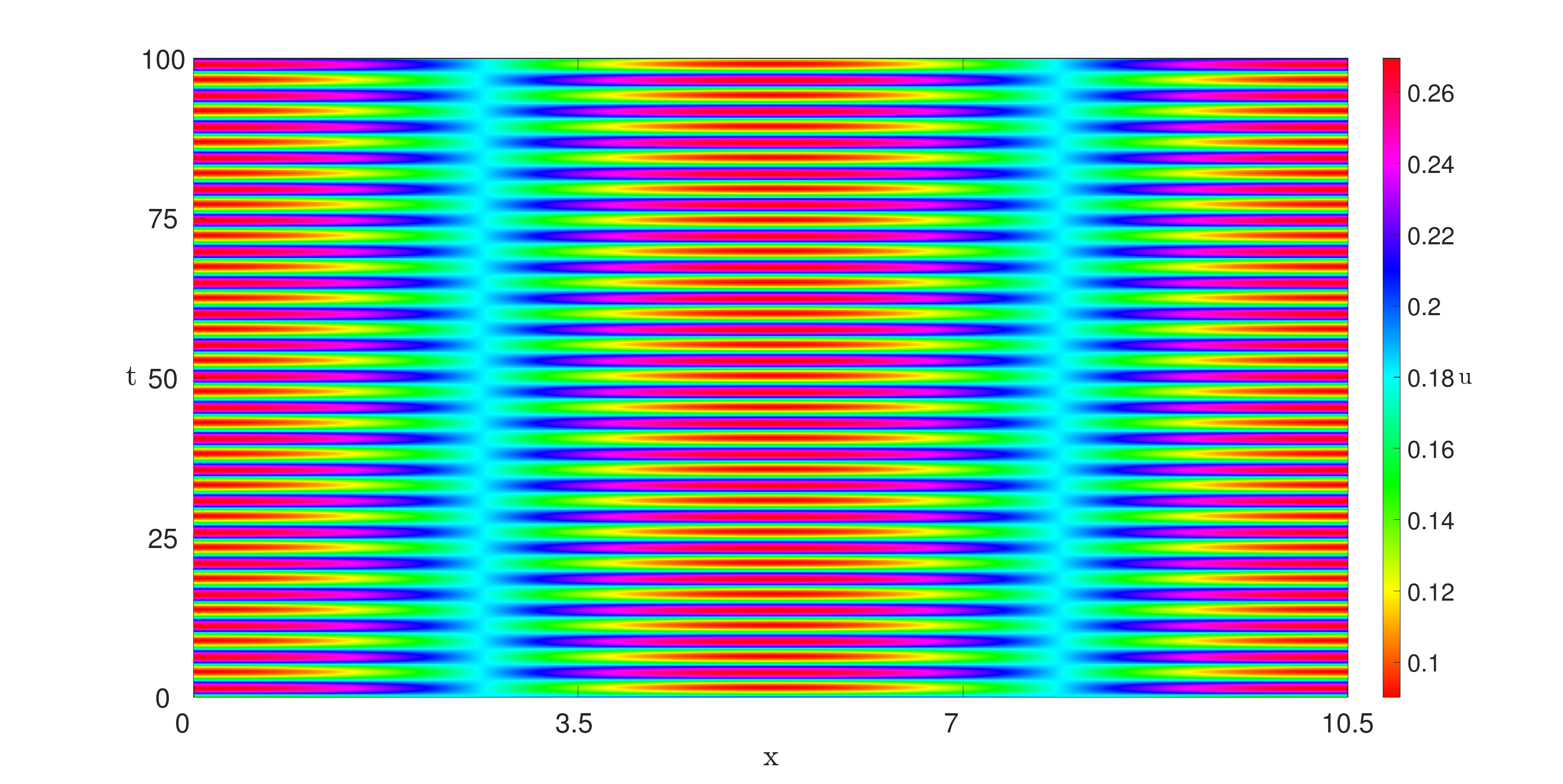}
                    \caption{}
                \end{subfigure}
                      \caption{Stable small-amplitude waves close to the homogeneous steady state of model \eqref{2CL} for parameter values \eqref{2CLparval} together with $\left(f_1, D_4\right) = (1, 9)$. (a) Travelling waves with periodic boundary conditions; (b) standing waves with Neumann boundary conditions.}
                \label{fig:2CL-SW-uvw}
            \end{figure}
            On the other hand, if we take the same parameter values but impose homogeneous Neumann boundary conditions, we get the pattern shown for $u$ in Figure \ref{fig:2CL-SW-uvw}(b).
 
\section{Discussion} \label{sec:concl}
    This paper has conducted a systematic weakly nonlinear analysis of wave bifurcations in systems of reaction-diffusion equations. In particular, we have shown how to construct the appropriate normal form for a general system of $n \geq 3$ reaction-diffusion equations. Such bifurcations have previously been studied in the context of fluid mechanics, but we are not aware of such a systematic treatment in the reaction-diffusion-system context. Understanding the onset of either travelling or standing waves can be important in many applications in biological excitable systems such as neural, cardiac, or calcium-wave models. 

    The present work is restricted to just one dimension in space, and much more needs to be done to consider analogous results in two or three spatial dimensions. We mention the literature on the onset of spatio-temporal spiral waves in that context, see e.g.~\cite{Delnitz,Scheel,Holden}. 

    One of the motivations of this work is being able to accurately identify the codimension-two points where wave bifurcations change from super- to sub-critical. In the case of Turing bifurcations, such points are important as points at which there is the birth of a Maxwell pinning region inside which we expect to find localised patterns due to so-called homoclininc snaking, see e.g.\cite{Brena}. Recently, evidence was found in 2-component reaction-diffusion systems that localised temporally periodic structures can also undergo snaking \cite{Fahadbreathing}. In that case, these snaking branches of localised {\em breathers} arise due to a secondary Hopf bifurcation along a homoclinic snaking curve. However, in 3-component systems we expect such structures to arise directly from the pinning region close to the codimension-two Bautin bifurcation studied in this paper. This may also explain the localised wave trains seen numerically in the neighbourbood of the sub-critical wave bifurcation in the first of our two examples, which was the main focus of the study in \cite{YochelisKnobloch}. 

    The exploration of this conjecture about localised structures is left to future work. The possibility of such localised structures is just one part of the full unfolding of the dynamics of the normal form, and its beyond-all-orders generalisations are left for future work. Other future work will consider unfoldings of the other codimension-two bifurcations we identified in Fig.~\ref{fig:bifdiag}, namely Hopf-Turing, Hopf-wave, and Turing-wave bifurcation points.

\subsection*{Acknowledgments}
    The work of EV-S was supported by ANID, Beca Chile Doctorado en el extranjero, number 72210071. The authors would also like to acknowledge useful conversations with Edgar Knobloch, Arik Yochelis, Hannes Uecker, Andrew Krause and Alastair Rucklidge. 

\section*{References}

\bibliographystyle{plain} 
\bibliography{refs}

\appendix
    \section{Expressions for normal-form coefficients up to order 5} \label{ap:order5}
    
	\paragraph{Fourth order} At fourth order, \eqref{eq:geneq} is given by
		\begin{multline*}
			\partial_A \mbf u^{[1,0]}\left(h_1^{[4,0]} + h_1^{[3,1]} + h_1^{[2,2]} + h_1^{[1,3]} + h_1^{[0,4]}\right)
            \s
            + \partial_B \mbf u^{[0,1]}\left(h_2^{[4,0]} + h_2^{[3,1]} + h_2^{[2,2]} + h_2^{[1,3]} + h_2^{[0,4]}\right)
			\s 
			+ \partial_A \mbf u^{[2,0]}\left(h_1^{[3,0]} + h_1^{[1,2]}\right) + \partial_B \mbf u^{[0,2]}\left(h_2^{[2,1]} + h_2^{[0,3]}\right)
			\s 
			+ \partial_A \mbf u^{[1,1]}\left(h_1^{[3,0]} + h_1^{[1,2]}\right) + \partial_B \mbf u^{[1,1]}\left(h_2^{[2,1]} + h_2^{[0,3]}\right)
			\s 
			+ \partial_A \mbf u^{[4,0]} \, h_1^{[1,0]} + \partial_B \mbf u^{[0,4]} \, h_2^{[0,1]} + \partial_A \mbf u^{[3,1]} \, h_1^{[1,0]} + \partial_B \mbf u^{[1,3]} \, h_2^{[0,1]}
			\s 
			+ \partial_A \mbf u^{[2,2]} \, h_1^{[1,0]} + \partial_B \mbf u^{[2,2]} \, h_2^{[0,1]} + \partial_A \mbf u^{[1,3]} \, h_1^{[1,0]} + \partial_B \mbf u^{[3,1]} \, h_2^{[0,1]} + \mbox{conj.}
			\s 
			= \left(\jac \mbf f(\mbf 0) + \mathds D \, \partial_{xx}\right)\left(\mbf u^{[4,0]} + \mbf u^{[3,1]} + \mbf u^{[2,2]} + \mbf u^{[1,3]} + \mbf u^{[0,4]}\right)
			\s 
			+ \mbf F_2\left(\mbf u^{[2,0]} + \mbf u^{[1,1]} + \mbf u^{[0,2]}, \mbf u^{[2,0]} + \mbf u^{[1,1]} + \mbf u^{[0,2]}\right)
			\s 
			+ 2 \, \mbf F_2\left(\mbf u^{[1,0]} + \mbf u^{[0,1]}, \mbf u^{[3,0]} + \mbf u^{[2,1]} + \mbf u^{[1,2]} + \mbf u^{[0,3]}\right)
			\s 
			+ 3 \, \mbf F_3\left(\mbf u^{[1,0]} + \mbf u^{[0,1]}, \mbf u^{[1,0]} + \mbf u^{[0,1]}, \mbf u^{[2,0]} + \mbf u^{[1,1]} + \mbf u^{[0,2]}\right)
			\s 
			+ \mbf F_4\left(\mbf u^{[1,0]} + \mbf u^{[0,1]}, \mbf u^{[1,0]} + \mbf u^{[0,1]}, \mbf u^{[1,0]} + \mbf u^{[0,1]}, \mbf u^{[1,0]} + \mbf u^{[0,1]}\right).
		\end{multline*}
		Here, again, by the same argument used at second order, we need to ensure that $h_1^{[j, \ell]} = h_2^{[j, \ell]} = 0$ for $j + \ell = 4$.
		
		With this, our equation can be expanded as:
		\begin{multline}
			\partial_A \mbf u^{[2,0]}\left(h_1^{[3,0]} + h_1^{[1,2]}\right) + \partial_B \mbf u^{[0,2]}\left(h_2^{[2,1]} + h_2^{[0,3]}\right)
			\s 
			+ \partial_A \mbf u^{[1,1]}\left(h_1^{[3,0]} + h_1^{[1,2]}\right) + \partial_B \mbf u^{[1,1]}\left(h_2^{[2,1]} + h_2^{[0,3]}\right)
			\s 
			+ i \omega A \, \partial_A \mbf u^{[4,0]} - i\omega B \, \partial_B \mbf u^{[0,4]} + i\omega A \, \partial_A \mbf u^{[3,1]} - i\omega B \, \partial_B \mbf u^{[1,3]}
			\s 
			+ i\omega A \, \partial_A \mbf u^{[2,2]} - i\omega B \, \partial_B \mbf u^{[2,2]} + i\omega A \, \partial_A \mbf u^{[1,3]} - i\omega B \, \partial_B \mbf u^{[3,1]} + \mbox{conj.}
			\s 
			= \left(\jac \mbf f(\mbf 0) + \mathds D \, \partial_{xx}\right)\left(\mbf u^{[4,0]} + \mbf u^{[3,1]} + \mbf u^{[2,2]} + \mbf u^{[1,3]} + \mbf u^{[0,4]}\right)
			\s 
			+ \mbf F_2\left(\mbf u^{[2,0]}, \mbf u^{[2,0]}\right) + \mbf F_2\left(\mbf u^{[1,1]}, \mbf u^{[1,1]}\right) + \mbf F_2\left(\mbf u^{[0,2]}, \mbf u^{[0,2]}\right)
			\s 
			+ 2 \left(\mbf F_2\left(\mbf u^{[2,0]}, \mbf u^{[1,1]}\right) + \mbf F_2\left(\mbf u^{[2,0]}, \mbf u^{[0,2]}\right) + \mbf F_2\left(\mbf u^{[1,1]}, \mbf u^{[0,2]}\right)\right)
			\s 
			+ 2 \left(\mbf F_2\left(\mbf u^{[1,0]}, \mbf u^{[3,0]}\right) + \mbf F_2\left(\mbf u^{[1,0]}, \mbf u^{[2,1]}\right) + \mbf F_2\left(\mbf u^{[1,0]}, \mbf u^{[1,2]}\right) + \mbf F_2\left(\mbf u^{[1,0]}, \mbf u^{[0,3]}\right)\right.
			\s 
			\left. + \mbf F_2\left(\mbf u^{[0,1]}, \mbf u^{[3,0]}\right) + \mbf F_2\left(\mbf u^{[0,1]}, \mbf u^{[2,1]}\right) + \mbf F_2\left(\mbf u^{[0,1]}, \mbf u^{[1,2]}\right) + \mbf F_2\left(\mbf u^{[0,1]}, \mbf u^{[0,3]}\right)\right)
			\s 
			+ 3 \left(\mbf F_3\left(\mbf u^{[1,0]}, \mbf u^{[1,0]}, \mbf u^{[2,0]}\right) + \mbf F_3\left(\mbf u^{[1,0]}, \mbf u^{[1,0]}, \mbf u^{[1,1]}\right) + \mbf F_3\left(\mbf u^{[1,0]}, \mbf u^{[1,0]}, \mbf u^{[0,2]}\right)\right)
			\s 
			+ 6 \left(\mbf F_3\left(\mbf u^{[1,0]}, \mbf u^{[0,1]}, \mbf u^{[2,0]}\right) + \mbf F_3\left(\mbf u^{[1,0]}, \mbf u^{[0,1]}, \mbf u^{[1,1]}\right) + \mbf F_3\left(\mbf u^{[1,0]}, \mbf u^{[0,1]}, \mbf u^{[0,2]}\right)\right)
			\s 
			+ 3 \left(\mbf F_3\left(\mbf u^{[0,1]}, \mbf u^{[0,1]}, \mbf u^{[2,0]}\right) + \mbf F_3\left(\mbf u^{[0,1]}, \mbf u^{[0,1]}, \mbf u^{[1,1]}\right) + \mbf F_3\left(\mbf u^{[0,1]}, \mbf u^{[0,1]}, \mbf u^{[0,2]}\right)\right)
			\s 
			+ \mbf F_4\left(\mbf u^{[1,0]}, \mbf u^{[1,0]}, \mbf u^{[1,0]}, \mbf u^{[1,0]}\right) + \mbf F_4\left(\mbf u^{[0,1]}, \mbf u^{[0,1]}, \mbf u^{[0,1]}, \mbf u^{[0,1]}\right)
			\s 
			+ 4 \left(\mbf F_4\left(\mbf u^{[1,0]}, \mbf u^{[1,0]}, \mbf u^{[1,0]}, \mbf u^{[0,1]}\right) + \mbf F_4\left(\mbf u^{[1,0]}, \mbf u^{[0,1]}, \mbf u^{[0,1]}, \mbf u^{[0,1]}\right)\right)
			\s 
			+ 6 \, \mbf F_4\left(\mbf u^{[1,0]}, \mbf u^{[1,0]}, \mbf u^{[0,1]}, \mbf u^{[0,1]}\right). \label{fourthorder}
		\end{multline}
		Therefore, when splitting \eqref{fourthorder} according to the amplitudes, we can see that
		\begin{multline}
			\partial_A \mbf u^{[2,0]} \, h_1^{[3,0]} + i \omega A \, \partial_A \mbf u^{[4,0]} + \mbox{conj.} = \left(\jac \mbf f(\mbf 0) + \mathds D \, \partial_{xx}\right) \mbf u^{[4,0]}
			\s 
			+ |A|^4\left(2 \, \mbf F_2\left(\mbf W_0^{[2]}, \mbf W_0^{[2]}\right) + \mbf F_2\left(\mbf W_2^{[2]}, \ol{\mbf W_2^{[2]}}\right) + 2 \, \mbf F_2\left(\bs \phi_1^{[1]}, \ol{\mbf W_1^{[3]}}\right)\right.
			\s 
			\left. + 3 \, \mbf F_3\left(\bs \phi_1^{[1]}, \bs \phi_1^{[1]}, \ol{\mbf W_2^{[2]}}\right) + 6 \, \mbf F_3\left(\bs \phi_1^{[1]}, \ol{\bs \phi_1^{[1]}}, \mbf W_0^{[2]}\right) + 3 \, \mbf F_4\left(\bs \phi_1^{[1]}, \bs \phi_1^{[1]}, \ol{\bs \phi_1^{[1]}}, \ol{\bs \phi_1^{[1]}}\right)\right)
			\s 
			+ |A|^2 \, A^2 \, e^{2ikx} \left(4 \, \mbf F_2\left(\mbf W_0^{[2]}, \mbf W_2^{[2]}\right) + 2 \, \mbf F_2\left(\bs \phi_1^{[1]}, \mbf W_1^{[3]}\right) + 2 \, \mbf F_2\left(\ol{\bs \phi_1^{[1]}}, \mbf W_3^{[3]}\right) \right.
			\s 
			\left. + 6 \, \mbf F_3\left(\bs \phi_1^{[1]}, \bs \phi_1^{[1]}, \mbf W_0^{[2]}\right) + 6 \, \mbf F_3\left(\bs \phi_1^{[1]}, \ol{\bs \phi_1^{[1]}}, \mbf W_2^{[2]}\right) + 4 \, \mbf F_4\left(\bs \phi_1^{[1]}, \bs \phi_1^{[1]}, \bs \phi_1^{[1]}, \ol{\bs \phi_1^{[1]}}\right)\right) + \ldots + c.c., \label{firstequationfourthorder}
		\end{multline}
		\begin{multline}
			\partial_A \mbf u^{[1,1]} \, h_1^{[3,0]} + \partial_B \mbf u^{[1,1]} \, h_2^{[2,1]} + i\omega A \, \partial_A \mbf u^{[3,1]} - i\omega B \, \partial_B \mbf u^{[3,1]} + \mbox{conj.}
			\s 
			= \left(\jac \mbf f(\mbf 0) + \mathds D \, \partial_{xx}\right) \mbf u^{[3,1]}
			\s 
			+ |A|^2 \, A \bar B \left(4 \, \mbf F_2\left(\mbf W_0^{[2]}, \mbf W_0^{[1,1]}\right) + 2 \, \mbf F_2\left(\mbf W_2^{[2]}, \mbf W_2^{[1,1]}\right) + 2 \, \mbf F_2\left(\bs \phi_1^{[1]}, \mbf W_1^{[3]} + \mbf W_1^{[2,1]}\right) \right.
			\s 
			\left. + 2 \, \mbf F_2\left(\ol{\bs \phi_1^{[1]}}, \mbf W_{1,2}^{[2,1]}\right) + 3 \, \mbf F_3\left(\bs \phi_1^{[1]}, \bs \phi_1^{[1]}, 4 \, \mbf W_0^{[2]} + \mbf W_2^{[1,1]}\right) + 6 \, \mbf F_3\left(\bs \phi_1^{[1]}, \ol{\bs \phi_1^{[1]}}, \mbf W_2^{[2]} + \mbf W_0^{[1,1]}\right) \right.
			\s 
			\left. + 12 \, \mbf F_4\left(\bs \phi_1^{[1]}, \bs \phi_1^{[1]}, \bs \phi_1^{[1]}, \ol{\bs \phi_1^{[1]}}\right)\right)
			\s 
			+ |A|^2 \, AB \, e^{2ikx} \left(4 \, \mbf F_2\left(\mbf W_0^{[2]}, \mbf W_2^{[1,1]}\right) + 2 \, \mbf F_2\left(\mbf W_2^{[2]}, \ol{\mbf W_0^{[1,1]}}\right) + 2 \, \mbf F_2\left(\bs \phi_1^{[1]}, \ol{\mbf W_1^{[2,1]}}\right) \right.
			\s 
			\left. + 2 \, \mbf F_2\left(\ol{\bs \phi_1^{[1]}}, \mbf W_1^{[3]} + \mbf W_3^{[2,1]}\right) + 3 \, \mbf F_3\left(\bs \phi_1^{[1]}, \bs \phi_1^{[1]}, \ol{\mbf W_0^{[1,1]}}\right) + 6 \, \mbf F_3\left(\bs \phi_1^{[1]}, \ol{\bs \phi_1^{[1]}}, 2 \, \mbf W_0^{[2]} + \mbf W_2^{[1,1]}\right) \right.
			\s 
			\left. + 6 \, \mbf F_3\left(\ol{\bs \phi_1^{[1]}}, \ol{\bs \phi_1^{[1]}}, \mbf W_2^{[2]}\right) + 12 \, \mbf F_4\left(\bs \phi_1^{[1]}, \bs \phi_1^{[1]}, \ol{\bs \phi_1^{[1]}}, \ol{\bs \phi_1^{[1]}}\right)\right)
			\s 
			+ A^3 \bar B \, e^{2ikx} \left(2 \, \mbf F_2\left(\mbf W_2^{[2]}, \mbf W_0^{[1,1]}\right) + 2 \, \mbf F_2\left(\bs \phi_1^{[1]}, \mbf W_3^{[3]} + \mbf W_{1,2}^{[2,1]}\right) \right.
			\s 
			\left. + 3 \, \mbf F_3\left(\bs \phi_1^{[1]}, \bs \phi_1^{[1]}, \mbf W_0^{[1,1]} + 2 \, \mbf W_2^{[2]}\right) + 4 \, \mbf F_4\left(\bs \phi_1^{[1]}, \bs \phi_1^{[1]}, \bs \phi_1^{[1]}, \bs \phi_1^{[1]}\right)\right) + \ldots + c.c., \label{secondequationfourthorder}
		\end{multline}
		and
		\begin{multline}
			\partial_A \mbf u^{[2,0]} \, h_1^{[1,2]} + \partial_B \mbf u^{[0,2]} \, h_2^{[2,1]} + i\omega A \, \partial_A \mbf u^{[2,2]} - i\omega B \, \partial_B \mbf u^{[2,2]} + \mbox{conj.}
			\s 
			= \left(\jac \mbf f(\mbf 0) + \mathds D \, \partial_{xx}\right) \mbf u^{[2,2]}
			\s 
			+ |A|^2 \, |B|^2 \left(4 \, \mbf F_2\left(\mbf W_0^{[2]}, \mbf W_0^{[2]}\right) + \mbf F_2\left(\mbf W_2^{[1,1]}, \mbf W_2^{[1,1]}\right) + \mbf F_2\left(\mbf W_0^{[1,1]}, \ol{\mbf W_0^{[1,1]}}\right) \right.
			\s 
			+ 4 \, \mbf F_2\left(\bs \phi_1^{[1]}, \ol{\mbf W_1^{[2,1]}}\right) + 6 \, \mbf F_3\left(\bs \phi_1^{[1]}, \ol{\bs \phi_1^{[1]}}, 2 \, \mbf W_0^{[2]} + \mbf W_2^{[1,1]}\right) + 6 \, \mbf F_3\left(\bs \phi_1^{[1]}, \bs \phi_1^{[1]}, \ol{\mbf W_0^{[1,1]}}\right)
            \s
            \left. + 12 \, \mbf F_4\left(\bs \phi_1^{[1]}, \bs \phi_1^{[1]}, \ol{\bs \phi_1^{[1]}}, \ol{\bs \phi_1^{[1]}}\right)\right)
			\s
			+ |B|^2 \, A^2 \, e^{2ikx} \left(2 \,  \mbf F_2\left(\mbf W_0^{[1,1]}, \mbf W_2^{[1,1]}\right) + 4 \, \mbf F_2\left(\mbf W_0^{[2]}, \mbf W_2^{[2]}\right) + 2 \, \mbf F_2\left(\bs \phi_1^{[1]}, \mbf W_1^{[2,1]} + \mbf W_3^{[2,1]}\right)\right.
			\s 
			+ 2 \, \mbf F_2\left(\ol{\bs \phi_1^{[1]}}, \mbf W_{1,2}^{[2,1]}\right) + 6 \, \mbf F_3\left(\bs \phi_1^{[1]}, \bs \phi_1^{[1]}, \mbf W_0^{[2]} + \mbf W_2^{[1,1]}\right) + 6 \, \mbf F_3\left(\bs \phi_1^{[1]}, \ol{\bs \phi_1^{[1]}}, \mbf W_2^{[2]} + \mbf W_0^{[1,1]}\right)
            \\
            \left. + 12 \, \mbf F_4\left(\bs \phi_1^{[1]}, \bs \phi_1^{[1]}, \bs \phi_1^{[1]}, \ol{\bs \phi_1^{[1]}}\right)\right) + \ldots + c.c., \label{thirdequationfourthorder}
		\end{multline}
        where ``$\ldots$'' represents terms that will not be necessary to obtain what we want, and the other two equations correspond to complex conjugate versions of \eqref{firstequationfourthorder} and \eqref{secondequationfourthorder}. The solutions to these equations are given by
		\begin{align*}
			\mbf u^{[4,0]} &= |A|^4 \, \mbf W_0^{[4]} + |A|^2 \, A^2 \, e^{2ikx} \, \mbf W_2^{[4]} + \ldots + c.c.,
			\s 
			\mbf u^{[3,1]} &= |A|^2 \, A \bar B \, \mbf W_0^{[3,1]} + |A|^2 \, AB \, e^{2ikx} \, \mbf W_2^{[3,1]} + A^3 \bar B \, e^{2ikx} \, \mbf W_{2,2}^{[3,1]} + \ldots + c.c.,
			\s 
			\mbf u^{[2,2]} &= |A|^2 \, |B|^2 \, \mbf W_0^{[2,2]} + |B|^2 \, A^2 \, e^{2ikx} \, \mbf W_2^{[2,2]} + |A|^2 \, B^2 \, e^{2ikx} \, \ol{\mbf W_2^{[2,2]}} + \ldots + c.c.,
			\s 
			\mbf u^{[1,3]} &= |B|^2 \, \bar A B \, \ol{\mbf W_0^{[3,1]}} + |B|^2 \, AB \, e^{2ikx} \, \ol{\mbf W_2^{[3,1]}} + \bar A B^3 \, e^{2ikx} \, \ol{\mbf W_{2,2}^{[3,1]}} + \ldots + c.c.,
			\s 
			\mbf u^{[0,4]} &= |B|^4 \, \mbf W_0^{[4]} + |B|^2 \, B^2 \, e^{2ikx} \, \ol{\mbf W_2^{[4]}} + \ldots + c.c.,
		\end{align*}
        where ``$\ldots$'' stands for terms that are multiples of $e^{4ikx}$ and will not be necessary for the computation of the fifth-order coefficients so they are omitted.
        
		Thus, by replacing each term of $\mbf u^{[4, 0]}$, $\mbf u^{[3, 1]}$ and $\mbf u^{[2, 2]}$ into \eqref{firstequationfourthorder}, \eqref{secondequationfourthorder} and \eqref{thirdequationfourthorder}, respectively, we can see that
		\begin{multline*}
			\jac \mbf f(\mbf 0) \, \mbf W_0^{[4]} = 2 \, C_1^{[3,0]} \, \mbf W_0^{[2]} - 2 \, \mbf F_2\left(\mbf W_0^{[2]}, \mbf W_0^{[2]}\right) - \mbf F_2\left(\mbf W_2^{[2]}, \ol{\mbf W_2^{[2]}}\right) - 2 \, \mbf F_2\left(\bs \phi_1^{[1]}, \ol{\mbf W_1^{[3]}}\right)
			\s 
			- 3 \, \mbf F_3\left(\bs \phi_1^{[1]}, \bs \phi_1^{[1]}, \ol{\mbf W_2^{[2]}}\right) - 6 \, \mbf F_3\left(\bs \phi_1^{[1]}, \ol{\bs \phi_1^{[1]}}, \mbf W_0^{[2]}\right) - 3 \, \mbf F_4\left(\bs \phi_1^{[1]}, \bs \phi_1^{[1]}, \ol{\bs \phi_1^{[1]}}, \ol{\bs \phi_1^{[1]}}\right),
		\end{multline*}
		\begin{multline*}
			\left(\jac \mbf f(\mbf 0) - 4 \, k^2 \, \mathds D - 2 \, i\omega I\right) \mbf W_2^{[4]} = 2 \, C_1^{[3,0]} \, \mbf W_2^{[2]} - 4 \, \mbf F_2\left(\mbf W_0^{[2]}, \mbf W_2^{[2]}\right) - 2 \, \mbf F_2\left(\bs \phi_1^{[1]}, \mbf W_1^{[3]}\right)
			\s 
			- 2 \, \mbf F_2\left(\ol{\bs \phi_1^{[1]}}, \mbf W_3^{[3]}\right) - 6 \, \mbf F_3\left(\bs \phi_1^{[1]}, \bs \phi_1^{[1]}, \mbf W_0^{[2]}\right) - 6 \, \mbf F_3\left(\bs \phi_1^{[1]}, \ol{\bs \phi_1^{[1]}}, \mbf W_2^{[2]}\right) - 4 \, \mbf F_4\left(\bs \phi_1^{[1]}, \bs \phi_1^{[1]}, \bs \phi_1^{[1]}, \ol{\bs \phi_1^{[1]}}\right),
		\end{multline*}
		\begin{multline*}
			\left(\jac \mbf f(\mbf 0) - 2 \, i\omega I\right) \mbf W_0^{[3,1]} = \left(C_1^{[3,0]} + C_1^{[1,2]}\right) \mbf W_0^{[1,1]} - 4 \, \mbf F_2\left(\mbf W_0^{[2]}, \mbf W_0^{[1,1]}\right) - 2 \, \mbf F_2\left(\mbf W_2^{[2]}, \mbf W_2^{[1,1]}\right)
			\s 
			- 2 \, \mbf F_2\left(\bs \phi_1^{[1]}, \mbf W_1^{[3]} + \mbf W_1^{[2,1]}\right) - 2 \, \mbf F_2\left(\ol{\bs \phi_1^{[1]}}, \mbf W_{1,2}^{[2,1]}\right) - 3 \, \mbf F_3\left(\bs \phi_1^{[1]}, \bs \phi_1^{[1]}, 4 \, \mbf W_0^{[2]} + \mbf W_2^{[1,1]}\right)
			\s 
			- 6 \, \mbf F_3\left(\bs \phi_1^{[1]}, \ol{\bs \phi_1^{[1]}}, \mbf W_2^{[2]} + \mbf W_0^{[1,1]}\right) - 12 \, \mbf F_4\left(\bs \phi_1^{[1]}, \bs \phi_1^{[1]}, \bs \phi_1^{[1]}, \ol{\bs \phi_1^{[1]}}\right),
		\end{multline*}
		\begin{multline*}
			\left(\jac \mbf f(\mbf 0) - 4 \, k^2 \, \mathds D\right) \mbf W_2^{[3,1]} = \left(C_1^{[3,0]} + \ol{C_1^{[1,2]}}\right) \mbf W_2^{[1,1]} - 4 \, \mbf F_2\left(\mbf W_0^{[2]}, \mbf W_2^{[1,1]}\right) - 2 \, \mbf F_2\left(\mbf W_2^{[2]}, \ol{\mbf W_0^{[1,1]}}\right)
			\s 
			- 2 \, \mbf F_2\left(\bs \phi_1^{[1]}, \ol{\mbf W_1^{[2,1]}}\right) - 2 \, \mbf F_2\left(\ol{\bs \phi_1^{[1]}}, \mbf W_1^{[3]} + \mbf W_3^{[2,1]}\right) - 3 \, \mbf F_3\left(\bs \phi_1^{[1]}, \bs \phi_1^{[1]}, \ol{\mbf W_0^{[1,1]}}\right)
			\s 
			- 6 \, \mbf F_3\left(\bs \phi_1^{[1]}, \ol{\bs \phi_1^{[1]}}, 2 \, \mbf W_0^{[2]} + \mbf W_2^{[1,1]}\right) - 6 \, \mbf F_3\left(\ol{\bs \phi_1^{[1]}}, \ol{\bs \phi_1^{[1]}}, \mbf W_2^{[2]}\right) - 12 \, \mbf F_4\left(\bs \phi_1^{[1]}, \bs \phi_1^{[1]}, \ol{\bs \phi_1^{[1]}}, \ol{\bs \phi_1^{[1]}}\right).
		\end{multline*}
		\begin{multline*}
			\jac \mbf f(\mbf 0) \, \mbf W_0^{[2,2]} = 4 \, C_1^{[1,2]} \, \mbf W_0^{[2]} - 4 \, \mbf F_2\left(\mbf W_0^{[2]}, \mbf W_0^{[2]}\right) - \mbf F_2\left(\mbf W_2^{[1,1]}, \mbf W_2^{[1,1]}\right)
			\\
			- \mbf F_2\left(\mbf W_0^{[1,1]}, \ol{\mbf W_0^{[1,1]}}\right) - 4 \, \mbf F_2\left(\bs \phi_1^{[1]}, \ol{\mbf W_1^{[2,1]}}\right) - 6 \, \mbf F_3\left(\bs \phi_1^{[1]}, \ol{\bs \phi_1^{[1]}}, 2 \, \mbf W_0^{[2]} + \mbf W_2^{[1,1]}\right)
            \\
            - 6 \, \mbf F_3\left(\bs \phi_1^{[1]}, \bs \phi_1^{[1]}, \ol{\mbf W_0^{[1,1]}}\right) - 12 \, \mbf F_4\left(\bs \phi_1^{[1]}, \bs \phi_1^{[1]}, \ol{\bs \phi_1^{[1]}}, \ol{\bs \phi_1^{[1]}}\right),
		\end{multline*}
		and
		\begin{multline*}
			\left(\jac \mbf f(\mbf 0) - 4 \, k^2 \, \mathds D - 2 \, i\omega I\right) \mbf W_2^{[2,2]} = 2 \, C_1^{[1,2]} \, \mbf W_2^{[2]} - 2 \,  \mbf F_2\left(\mbf W_0^{[1,1]}, \mbf W_2^{[1,1]}\right) - 4 \, \mbf F_2\left(\mbf W_0^{[2]}, \mbf W_2^{[2]}\right)
			\s 
			- 2 \, \mbf F_2\left(\bs \phi_1^{[1]}, \mbf W_1^{[2,1]} + \mbf W_3^{[2,1]}\right) - 2 \, \mbf F_2\left(\ol{\bs \phi_1^{[1]}}, \mbf W_{1,2}^{[2,1]}\right) - 6 \, \mbf F_3\left(\bs \phi_1^{[1]}, \bs \phi_1^{[1]}, \mbf W_0^{[2]} + \mbf W_2^{[1,1]}\right)
            \\
            - 6 \, \mbf F_3\left(\bs \phi_1^{[1]}, \ol{\bs \phi_1^{[1]}}, \mbf W_2^{[2]} + \mbf W_0^{[1,1]}\right) - 12 \, \mbf F_4\left(\bs \phi_1^{[1]}, \bs \phi_1^{[1]}, \bs \phi_1^{[1]}, \ol{\bs \phi_1^{[1]}}\right).
		\end{multline*}
        The other vectors can be obtained from the equations as well but they will not affect terms at order 5, so we omit the equations they solve. 

        With this, we are ready to proceed to the final order we are interested in.
                
	\paragraph{Fifth order}
        We start the development of the equations we get at this order by stating a Lemma based on \cite{Kuznetsov}:
		\newtheorem{lemma}{Lemma}
		\begin{lemma} \label{simplificationorder5}
			The equation
			\begin{align*}
				\dot{z}&=\lambda z + C^{[2,1]} \, z^2 \bar z + C^{[5,0]} \, z^5+C^{[4,1]} \, z^4\bar z+C^{[3,2]} \, z^3\bar z^2 + C^{[2,3]} \, z^2\bar z^3+C^{[1,4]} \, z\bar z^4+C^{[0,5]} \, \bar z^5,
			\end{align*}
			where $\lambda=\lambda(\alpha)=\mu(\alpha)+i\omega(\alpha)$, $\mu(0)=0$, $\omega(0)=\omega_0>0$, and $C^{[j,l]}=C^{[j,l]}(\alpha$), can be transformed by an invertible parameter-dependent change of coordinate into an equation with only one quintic term, for all $|\alpha|$ sufficiently small.
		\end{lemma}
		\begin{proof}
            Consider the following locally invertible coordinate:
			\begin{align*}
				w=z- h^{[5,0]} \, z^5-h^{[4,1]} \, z^4\bar z-h^{[3,2]} \, z^3\bar z^2 - h^{[2,3]} \, z^2\bar z^3-h^{[1,4]} \, z\bar z^4-h^{[0,5]} \, \bar z^5.
			\end{align*}
			With this, we have
			\begin{align*}
				\dot w &= \begin{multlined}[t][10cm]
					\lambda z + C^{[2,1]} \, z^2 \bar z + C^{[5,0]} \, z^5+C^{[4,1]} \, z^4\bar z+C^{[3,2]} \, z^3\bar z^2 + C^{[2,3]} \, z^2\bar z^3+C^{[1,4]} \, z\bar z^4+C^{[0,5]} \, \bar z^5
					\s 
					- 5 \lambda \, h^{[5,0]} \, z^5-h^{[4,1]} \left(4\lambda \, z^4\bar z+\bar \lambda \, z^4\bar z\right)-h^{[3,2]} \left(3\lambda \, z^3\bar z^2 + 2 \bar \lambda \, z^3\bar z^2\right)
					\s 
					- h^{[2,3]} \left(2\lambda \, z^2\bar z^3 + 3 \bar \lambda \, z^2\bar z^3\right)-h^{[1,4]} \left(\lambda \, z\bar z^4+4\bar \lambda \, z\bar z^4\right)-5 \bar \lambda \, h^{[0,5]} \, \bar z^5 + \mathcal O\left(z^6\right)
				\end{multlined}
				\s
				&=\begin{multlined}[t][10cm]
					\lambda w + C^{[2,1]} \, z^2 \bar z + \left(C^{[5,0]}-4 \lambda \, h^{[5,0]}\right) \, z^5+\left(C^{[4,1]}-3\lambda \, h^{[4,1]}-\bar \lambda \, h^{[4,1]}\right) \, z^4\bar z
					\s 
					+\left(C^{[3,2]}-2\lambda \, h^{[3,2]}-2\bar \lambda \, h^{[3,2]}\right) \, z^3\bar z^2 + \left(C^{[2,3]}-\lambda \, h^{[2,3]}-3\bar \lambda \, h^{[2,3]}\right) \, z^2\bar z^3
					\s 
					+\left(C^{[1,4]}-4\bar \lambda \, h^{[1,4]}\right) \, z\bar z^4+\left(C^{[0,5]}+\lambda \, h^{[0,5]}-5\bar \lambda  \, h^{[0,5]}\right) \, \bar z^5 + \mathcal O\left(z^6\right)
				\end{multlined}
				\s 
				&=\begin{multlined}[t][10cm]
					\lambda w + C^{[2,1]} \, w^2 \bar w + \left(C^{[5,0]}-4 \lambda \, h^{[5,0]}\right) \, w^5+\left(C^{[4,1]}- \left(3\lambda +\bar \lambda \right) h^{[4,1]}\right) \, w^4\bar w
					\s 
					+\left(C^{[3,2]}-\left(2\lambda +2\bar \lambda\right) h^{[3,2]}\right) \, w^3\bar w^2 + \left(C^{[2,3]}-\left(\lambda +3\bar \lambda \right) h^{[2,3]}\right) \, w^2\bar w^3
					\s 
					+\left(C^{[1,4]}-4\bar \lambda \, h^{[1,4]}\right) \, w\bar w^4+\left(C^{[0,5]}+\left(\lambda -5\bar \lambda \right) h^{[0,5]}\right) \, \bar w^5 + \mathcal O\left(w^6\right).
				\end{multlined}
			\end{align*}
			Therefore, if we choose
			\begin{align*}
				h^{[5,0]} = \frac{C^{[5,0]}}{4\lambda}, \quad h^{[4,1]} = \frac{C^{[4,1]}}{3\lambda + \bar \lambda}, \quad h^{[2,3]} = \frac{C^{[2,3]}}{\lambda + 3\bar \lambda}, \quad h^{[1,4]} = \frac{C^{[1,4]}}{4\bar \lambda}, \quad h^{[0,5]} = \frac{C^{[0,5]}}{\lambda-5\bar \lambda},
			\end{align*}
			then our equation gets reduced to
			\begin{align*}
				\dot{w}&=\begin{multlined}[t][10cm]
					\lambda w + C^{[2,1]} \, w^2 \bar w +\left(C^{[3,2]}-\left(2\lambda +2\bar \lambda\right) h^{[3,2]}\right) \, w^3\bar w^2.
				\end{multlined}
			\end{align*}
			However, the last term cannot be reduced since $\lambda(0)+\bar \lambda(0)=0$. Therefore, we choose $h^{[3,2]}=0$, which leaves us the following equation:
			\begin{align*}
				\dot{w}&=\begin{multlined}[t][10cm]
					\lambda w + C^{[2,1]} \, w^2 \bar w +C^{[3,2]} \, w^3\bar w^2+\mathcal O\left(w^6\right).
				\end{multlined}
			\end{align*}
		\end{proof}
		Keeping this lemma in mind note that, at order 5, \eqref{eq:geneq} is given by
		\begin{multline*}
			\partial_A \mbf u^{[1,0]}\left(h_1^{[5,0]} + h_1^{[4,1]} + h_1^{[3,2]} + h_1^{[2,3]} + h_1^{[1,4]} + h_1^{[0,5]}\right)
			\s 
			+ \partial_B \mbf u^{[0,1]} \left(h_2^{[5,0]} + h_2^{[4,1]} + h_2^{[3,2]} + h_2^{[2,3]} + h_2^{[1,4]} + h_2^{[0,5]}\right)
			\s 
			+ \partial_A \mbf u^{[3,0]}\left(h_1^{[3,0]} + h_1^{[1,2]}\right) + \partial_B \mbf u^{[0,3]}\left(h_2^{[2,1]} + h_2^{[0,3]}\right) + \partial_A \mbf u^{[2,1]}\left(h_1^{[3,0]} + h_1^{[1,2]}\right)
			\s 
			+ \partial_B \mbf u^{[1,2]}\left(h_2^{[2,1]} + h_2^{[0,3]}\right) + \partial_A \mbf u^{[1,2]}\left(h_1^{[3,0]} + h_1^{[1,2]}\right) + \partial_B \mbf u^{[2,1]}\left(h_2^{[2,1]} + h_2^{[0,3]}\right)
			\s 
			+ \partial_A \mbf u^{[5,0]} \, h_1^{[1,0]} + \partial_B \mbf u^{[0,5]} \, h_2^{[0,1]} + \partial_A \mbf u^{[4,1]} \, h_1^{[1,0]} + \partial_B \mbf u^{[1,4]} \, h_2^{[0,1]}
			\s 
			+ \partial_A \mbf u^{[3,2]} \, h_1^{[1,0]} + \partial_B \mbf u^{[2,3]} \, h_2^{[0,1]} + \partial_A \mbf u^{[2,3]} \, h_1^{[1,0]} + \partial_B \mbf u^{[3,2]} \, h_2^{[0,1]}
			\s 
			+ \partial_A \mbf u^{[1,4]} \, h_1^{[1,0]} + \partial_B \mbf u^{[4,1]} \, h_2^{[0,1]} + \mbox{conj.}
			\s 
			= \left(\jac \mbf f(\mbf 0) + \mathds D \, \partial_{xx}\right)\left(\mbf u^{[5,0]} + \mbf u^{[4,1]} + \mbf u^{[3,2]} + \mbf u^{[2,3]} + \mbf u^{[1,4]} + \mbf u^{[0,5]}\right)
			\s 
			+ 2 \, \mbf F_2\left(\mbf u^{[1,0]} + \mbf u^{[0,1]}, \mbf u^{[4,0]} + \mbf u^{[3,1]} + \mbf u^{[2,2]} + \mbf u^{[1,3]} + \mbf u^{[0,4]}\right)
			\s 
			+ 2 \, \mbf F_2\left(\mbf u^{[2,0]} + \mbf u^{[1,1]} + \mbf u^{[0,2]}, \mbf u^{[3,0]} + \mbf u^{[2,1]} + \mbf u^{[1,2]} + \mbf u^{[0,3]}\right)
			\s 
			+ 3 \, \mbf F_3\left(\mbf u^{[1,0]} + \mbf u^{[0,1]}, \mbf u^{[2,0]} + \mbf u^{[1,1]} + \mbf u^{[0,2]}, \mbf u^{[2,0]} + \mbf u^{[1,1]} + \mbf u^{[0,2]}\right)
			\s 
			+ 3 \, \mbf F_3\left(\mbf u^{[1,0]} + \mbf u^{[0,1]}, \mbf u^{[1,0]} + \mbf u^{[0,1]}, \mbf u^{[3,0]} + \mbf u^{[2,1]} + \mbf u^{[1,2]} + \mbf u^{[0,3]}\right)
			\s 
			+ 4 \, \mbf F_4\left(\mbf u^{[1,0]} + \mbf u^{[0,1]}, \mbf u^{[1,0]} + \mbf u^{[0,1]}, \mbf u^{[1,0]} + \mbf u^{[0,1]}, \mbf u^{[2,0]} + \mbf u^{[1,1]} + \mbf u^{[0,2]}\right)
			\s 
			+ \mbf F_5\left(\mbf u^{[1,0]} + \mbf u^{[0,1]}, \mbf u^{[1,0]} + \mbf u^{[0,1]}, \mbf u^{[1,0]} + \mbf u^{[0,1]}, \mbf u^{[1,0]} + \mbf u^{[0,1]}, \mbf u^{[1,0]} + \mbf u^{[0,1]}\right).
		\end{multline*}
		Therefore, when splitting this equation according to the amplitudes, we get three equations given by:
		\begin{multline}
			\partial_A \mbf u^{[1,0]} \, h_1^{[5,0]} + \partial_B \mbf u^{[0,1]} \, h_2^{[5,0]} + \partial_A \mbf u^{[3,0]} \, h_1^{[3,0]} + i\omega A \, \partial_A \mbf u^{[5,0]} + \mbox{conj.}
			\s 
			= \left(\jac \mbf f(\mbf 0) + \mathds D \, \partial_{xx}\right) \mbf u^{[5,0]} + 2 \, \mbf F_2\left(\mbf u^{[1,0]}, \mbf u^{[4,0]}\right)
			\s 
			+ 2 \, \mbf F_2\left(\mbf u^{[2,0]}, \mbf u^{[3,0]}\right) + 3 \, \mbf F_3\left(\mbf u^{[1,0]}, \mbf u^{[2,0]}, \mbf u^{[2,0]}\right)
			\s
			+ 3 \, \mbf F_3\left(\mbf u^{[1,0]}, \mbf u^{[1,0]}, \mbf u^{[3,0]}\right) + 4 \, \mbf F_4\left(\mbf u^{[1,0]}, \mbf u^{[1,0]}, \mbf u^{[1,0]}, \mbf u^{[2,0]}\right)
			\s 
			+ \mbf F_5\left(\mbf u^{[1,0]}, \mbf u^{[1,0]}, \mbf u^{[1,0]}, \mbf u^{[1,0]}, \mbf u^{[1,0]}\right), \label{firstequationfifthorder}
		\end{multline}
		\begin{multline}
			\partial_A \mbf u^{[1,0]} \, h_1^{[4,1]} + \partial_B \mbf u^{[0,1]} \, h_2^{[4,1]} + \partial_A \mbf u^{[2,1]} \, h_1^{[3,0]} + \partial_B \mbf u^{[2,1]} \, h_2^{[2,1]}
			\s 
			+i\omega A \, \partial_A \mbf u^{[4,1]} - i\omega B \, \partial_B \mbf u^{[4,1]} + \mbox{conj.}
			\s 
			= \left(\jac \mbf f(\mbf 0) + \mathds D \, \partial_{xx}\right) \mbf u^{[4,1]} + 2 \, \mbf F_2\left(\mbf u^{[0,1]}, \mbf u^{[4,0]}\right)
			\s 
			+ 2 \, \mbf F_2\left(\mbf u^{[1,0]}, \mbf u^{[3,1]}\right) + 2 \, \mbf F_2\left(\mbf u^{[2,0]}, \mbf u^{[2,1]}\right) + 2 \, \mbf F_2\left(\mbf u^{[1,1]}, \mbf u^{[3,0]}\right)
			\s 
			+ 3 \, \mbf F_3\left(\mbf u^{[0,1]}, \mbf u^{[2,0]}, \mbf u^{[2,0]}\right) + 6 \, \mbf F_3\left(\mbf u^{[1,0]}, \mbf u^{[2,0]}, \mbf u^{[1,1]}\right)
			\s 
			+ 3 \, \mbf F_3\left(\mbf u^{[1,0]}, \mbf u^{[1,0]}, \mbf u^{[2,1]}\right) + 6 \, \mbf F_3\left(\mbf u^{[1,0]}, \mbf u^{[0,1]}, \mbf u^{[3,0]}\right)
			\s
			+ 4 \, \mbf F_4\left(\mbf u^{[1,0]}, \mbf u^{[1,0]}, \mbf u^{[1,0]}, \mbf u^{[1,1]}\right) + 12 \, \mbf F_4\left(\mbf u^{[1,0]}, \mbf u^{[1,0]}, \mbf u^{[0,1]}, \mbf u^{[2,0]}\right)
			\s 
			+ 5 \, \mbf F_5\left(\mbf u^{[1,0]}, \mbf u^{[1,0]}, \mbf u^{[1,0]}, \mbf u^{[1,0]}, \mbf u^{[0,1]}\right), \label{secondequationfifthorder}
		\end{multline}
		and
		\begin{multline}
			\partial_A \mbf u^{[1,0]} \, h_1^{[3,2]} + \partial_B \mbf u^{[0,1]} \, h_2^{[3,2]} + \partial_A \mbf u^{[3,0]} \, h_1^{[1,2]} + \partial_B \mbf u^{[1,2]} \, h_2^{[2,1]} + \partial_A \mbf u^{[1,2]} \, h_1^{[3,0]}
			\s 
			+ i\omega A \, \partial_A \mbf u^{[3,2]} - i\omega B \, \partial_B \mbf u^{[3,2]} + \mbox{conj.}
			\s 
			= \left(\jac \mbf f(\mbf 0) + \mathds D \, \partial_{xx}\right) \mbf u^{[3,2]} + 2 \, \mbf F_2\left(\mbf u^{[0,1]}, \mbf u^{[3,1]}\right)
			\s 
			+ 2 \, \mbf F_2\left(\mbf u^{[1,0]}, \mbf u^{[2,2]}\right) + 2 \, \mbf F_2\left(\mbf u^{[2,0]}, \mbf u^{[1,2]}\right) + 2 \, \mbf F_2\left(\mbf u^{[1,1]}, \mbf u^{[2,1]}\right)
			\s 
			+ 2 \, \mbf F_2\left(\mbf u^{[0,2]}, \mbf u^{[3,0]}\right) + 3 \, \mbf F_3\left(\mbf u^{[1,0]}, \mbf u^{[1,1]}, \mbf u^{[1,1]}\right) + 6 \, \mbf F_3\left(\mbf u^{[1,0]}, \mbf u^{[2,0]}, \mbf u^{[0,2]}\right)
			\s 
			+ 6 \, \mbf F_3\left(\mbf u^{[0,1]}, \mbf u^{[2,0]}, \mbf u^{[1,1]}\right) + 3 \, \mbf F_3\left(\mbf u^{[1,0]}, \mbf u^{[1,0]}, \mbf u^{[1,2]}\right) + 3 \, \mbf F_3\left(\mbf u^{[0,1]}, \mbf u^{[0,1]}, \mbf u^{[3,0]}\right)
			\s 
			+ 6 \, \mbf F_3\left(\mbf u^{[1,0]}, \mbf u^{[0,1]}, \mbf u^{[2,1]}\right) + 4 \, \mbf F_4\left(\mbf u^{[1,0]}, \mbf u^{[1,0]}, \mbf u^{[1,0]}, \mbf u^{[0,2]}\right)
			\s 
			+ 12 \, \mbf F_4\left(\mbf u^{[1,0]}, \mbf u^{[1,0]}, \mbf u^{[0,1]}, \mbf u^{[1,1]}\right) + 12 \, \mbf F_4\left(\mbf u^{[1,0]}, \mbf u^{[0,1]}, \mbf u^{[0,1]}, \mbf u^{[2,0]}\right)
			\s 
			+ 10 \, \mbf F_5\left(\mbf u^{[1,0]}, \mbf u^{[1,0]}, \mbf u^{[1,0]}, \mbf u^{[0,1]}, \mbf u^{[0,1]}\right). \label{thirdequationfifthorder}
		\end{multline}
		Here, we are not concerned about the solutions of these systems, but the fifth-order coefficients associated with the amplitudes. To determine those, we need to state a solvability condition that comes from the parts of the solutions that are explicitly written below:
		\begin{align*}
			\mbf u^{[5,0]} &= |A|^4 \, A \, e^{ikx} \, \mbf W_1^{[5]} + \ldots + c.c.,
			\s 
			\mbf u^{[4,1]} &= |A|^4 \, \bar B \, e^{-ikx} \, \mbf W_1^{[4,1]} + \ldots + c.c.,
			\s 
			\mbf u^{[3,2]} &= |A|^2 \, |B|^2 \, A \, e^{ikx} \, \mbf W_1^{[3,2]} + \ldots + c.c.,
		\end{align*}
		where
        \begin{multline*}
			\left(\jac \mbf f(\mbf 0) - k^2 \, \mathds D - i \omega I\right) \, \mbf W_1^{[5]} = \bs \phi_1^{[1]} \, \frac{h_1^{[5,0]}}{|A|^4 A} + \ol{\bs \phi_1^{[1]}} \, \frac{h_2^{[5,0]}}{|A|^4 A}
			\s 
			- \left(\vphantom{\frac{1}{2}} - \left(2 \, C_1^{[3,0]} + \ol{C_1^{[3,0]}}\right) \mbf W_1^{[3]} + 2 \, \mbf F_2\left(\bs \phi_1^{[1]}, \mbf W_0^{[4]} + \ol{\mbf W_0^{[4]}}\right)\right.
			\s 
			+ 2 \, \mbf F_2\left(\ol{\bs \phi_1^{[1]}}, \mbf W_2^{[4]}\right) + 4 \, \mbf F_2\left(\mbf W_0^{[2]}, \mbf W_1^{[3]}\right) + 2 \, \mbf F_2\left(\mbf W_2^{[2]}, \ol{\mbf W_1^{[3]}}\right) + 2 \, \mbf F_2\left(\ol{\mbf W_2^{[2]}}, \mbf W_3^{[3]}\right)
			\s 
			+ 12 \, \mbf F_3\left(\bs \phi_1^{[1]}, \mbf W_0^{[2]}, \mbf W_0^{[2]}\right) + 6 \, \mbf F_3\left(\bs \phi_1^{[1]}, \mbf W_2^{[2]}, \ol{\mbf W_2^{[2]}}\right) + 12 \, \mbf F_3\left(\ol{\bs \phi_1^{[1]}}, \mbf W_0^{[2]}, \mbf W_2^{[2]}\right)
			\s 
			+ 3 \, \mbf F_3\left(\bs \phi_1^{[1]}, \bs \phi_1^{[1]}, \ol{\mbf W_1^{[3]}}\right) + 6 \, \mbf F_3\left(\bs \phi_1^{[1]}, \ol{\bs \phi_1^{[1]}}, \mbf W_1^{[3]}\right) + 3 \, \mbf F_3\left(\ol{\bs \phi_1^{[1]}}, \ol{\bs \phi_1^{[1]}}, \mbf W_3^{[3]}\right)
			\s 
			+ 4 \, \mbf F_4\left(\bs \phi_1^{[1]}, \bs \phi_1^{[1]}, \bs \phi_1^{[1]}, \ol{\mbf W_2^{[2]}}\right) + 24 \, \mbf F_4\left(\bs \phi_1^{[1]}, \bs \phi_1^{[1]}, \ol{\bs \phi_1^{[1]}}, \mbf W_0^{[2]}\right)
			\s 
			\left. + 12 \, \mbf F_4\left(\bs \phi_1^{[1]}, \ol{\bs \phi_1^{[1]}}, \ol{\bs \phi_1^{[1]}}, \mbf W_2^{[2]}\right) + 10 \, \mbf F_5\left(\bs \phi_1^{[1]}, \bs \phi_1^{[1]}, \bs \phi_1^{[1]}, \ol{\bs \phi_1^{[1]}}, \ol{\bs \phi_1^{[1]}}\right)\right),
		\end{multline*}
        \begin{multline*}
			\left(\jac \mbf f(\mbf 0) - k^2 \, \mathds D - i\omega I\right) \, \mbf W_1^{[4,1]} = \ol{\bs \phi_1^{[1]}} \, \frac{\ol{h_1^{[4,1]}}}{|A|^4 \bar B} + \bs \phi_1^{[1]} \, \frac{\ol{h_2^{[4,1]}}}{|A|^4 \bar B}
			\s 
			- \left(\vphantom{\frac{1}{2}} - \left(C_1^{[1,2]} + C_1^{[3,0]} + \ol{C_1^{[3,0]}}\right) \mbf W_1^{[2,1]}\right.
			\s 
			+ 2 \, \mbf F_2\left(\bs \phi_1^{[1]},\mbf W_0^{[4]} + \ol{\mbf W_0^{[4]}} + \ol{\mbf W_2^{[3,1]}}\right) + 2 \, \mbf F_2\left(\ol{\bs \phi_1^{[1]}}, \mbf W_0^{[3,1]}\right) + 4 \, \mbf F_2\left(\mbf W_0^{[2]}, \mbf W_1^{[2,1]}\right)
			\s 
			+ 2 \, \mbf F_2\left(\ol{\mbf W_2^{[2]}}, \mbf W_{1,2}^{[2,1]}\right) + 2 \, \mbf F_2\left(\mbf W_2^{[2]}, \ol{\mbf W_3^{[2,1]}}\right) + 2 \, \mbf F_2\left(\mbf W_1^{[3]}, \mbf W_2^{[1,1]}\right) + 2 \, \mbf F_2\left(\ol{\mbf W_1^{[3]}}, \mbf W_0^{[1,1]}\right)
			\s 
			+ 12 \, \mbf F_3\left(\bs \phi_1^{[1]}, \mbf W_0^{[2]}, \mbf W_0^{[2]} + \mbf W_2^{[1,1]}\right) + 6 \, \mbf F_3\left(\bs \phi_1^{[1]}, \ol{\mbf W_2^{[2]}}, \mbf W_2^{[2]} + \mbf W_0^{[1,1]}\right)
			\s 
			+ 12 \, \mbf F_3\left(\ol{\bs \phi_1^{[1]}}, \mbf W_0^{[2]}, \mbf W_0^{[1,1]}\right) + 6 \, \mbf F_3\left(\ol{\bs \phi_1^{[1]}}, \mbf W_2^{[2]}, \mbf W_2^{[1,1]}\right) + 6 \, \mbf F_3\left(\bs \phi_1^{[1]}, \ol{\bs \phi_1^{[1]}}, \mbf W_1^{[3]} + \mbf W_1^{[2,1]}\right)
			\s 
			+ 3 \, \mbf F_3\left(\ol{\bs \phi_1^{[1]}}, \ol{\bs \phi_1^{[1]}}, \mbf W_{1,2}^{[2,1]}\right) + 3 \, \mbf F_3\left(\bs \phi_1^{[1]}, \bs \phi_1^{[1]}, 2 \, \ol{\mbf W_1^{[3]}} + \ol{\mbf W_3^{[2,1]}}\right)
			\s 
			+ 12 \, \mbf F_4\left(\bs \phi_1^{[1]}, \bs \phi_1^{[1]}, \ol{\bs \phi_1^{[1]}}, 4 \, \mbf W_0^{[2]} + \mbf W_2^{[1,1]}\right) + 12 \, \mbf F_4\left(\bs \phi_1^{[1]}, \bs \phi_1^{[1]}, \bs \phi_1^{[1]}, \ol{\mbf W_2^{[2]}}\right)
            \s
            \left. + 12 \, \mbf F_4\left(\bs \phi_1^{[1]}, \ol{\bs \phi_1^{[1]}}, \ol{\bs \phi_1^{[1]}}, \mbf W_2^{[2]} + \mbf W_0^{[1,1]}\right) + 30 \, \mbf F_5\left(\bs \phi_1^{[1]}, \bs \phi_1^{[1]}, \bs \phi_1^{[1]}, \ol{\bs \phi_1^{[1]}}, \ol{\bs \phi_1^{[1]}}\right)\right),
		\end{multline*}
		and
        \begin{multline*}
			\left(\jac \mbf f(\mbf 0) - k^2 \, \mathds D - i \omega I\right) \, \mbf W_1^{[3,2]} = \bs \phi_1^{[1]} \, \frac{h_1^{[3,2]}}{|A|^2 \, |B|^2 A} + \ol{\bs \phi_1^{[1]}} \,  \frac{h_2^{[3,2]}}{|A|^2 \, |B|^2 A}
			\s 
			- \left(-\left(C_1^{[3,0]} + C_1^{[1,2]} + \ol{C_1^{[1,2]}}\right) \mbf W_1^{[2,1]} - \left(2 \, C_1^{[1,2]} + \ol{C_1^{[1,2]}}\right) \mbf W_1^{[3]} \right.
			\s 
			+ 2 \, \mbf F_2\left(\bs \phi_1^{[1]}, \mbf W_0^{[2,2]} + \ol{\mbf W_0^{[2,2]}} + \mbf W_2^{[3,1]}\right) + 2 \, \mbf F_2\left(\ol{\bs \phi_1^{[1]}}, \mbf W_2^{[2,2]} + \mbf W_0^{[3,1]}\right) + 4 \, \mbf F_2\left(\mbf W_0^{[2]}, \mbf W_1^{[2,1]}\right)
			\s 
			+ 2 \, \mbf F_2\left(\ol{\mbf W_1^{[2,1]}}, \mbf W_2^{[2]} + \mbf W_0^{[1,1]}\right) + 2 \, \mbf F_2\left(\mbf W_2^{[1,1]}, \mbf W_1^{[2,1]} + \mbf W_3^{[2,1]}\right) + 2 \, \mbf F_2\left(\ol{\mbf W_0^{[1,1]}}, \mbf W_{1,2}^{[2,1]}\right)
            \s
            + 4 \, \mbf F_2\left(\mbf W_0^{[2]}, \mbf W_1^{[3]}\right) + 6 \, \mbf F_3\left(\bs \phi_1^{[1]}, \mbf W_2^{[1,1]}, \mbf W_2^{[1,1]}\right) + 6 \, \mbf F_3\left(\bs \phi_1^{[1]}, \ol{\mbf W_0^{[1,1]}}, \mbf W_2^{[2]} + \mbf W_0^{[1,1]}\right)
			\s 
			+ 6 \, \mbf F_3\left(\ol{\bs \phi_1^{[1]}}, 2 \, \mbf W_0^{[2]} + \mbf W_2^{[1,1]}, \mbf W_2^{[2]} + \mbf W_0^{[1,1]}\right) + 12 \, \mbf F_3\left(\bs \phi_1^{[1]}, \mbf W_0^{[2]}, 2 \, \mbf W_0^{[2]} + \mbf W_2^{[1,1]}\right)
			\s 
			+ 9 \, \mbf F_3\left(\bs \phi_1^{[1]}, \bs \phi_1^{[1]}, \ol{\mbf W_1^{[2,1]}}\right) + 6 \, \mbf F_3\left(\bs \phi_1^{[1]}, \ol{\bs \phi_1^{[1]}}, \mbf W_1^{[3]} + 2 \, \mbf W_1^{[2,1]} + \mbf W_3^{[2,1]}\right)
			\s 
			+ 6 \, \mbf F_3\left(\ol{\bs \phi_1^{[1]}}, \ol{\bs \phi_1^{[1]}}, \mbf W_{1,2}^{[2,1]}\right) + 36 \, \mbf F_4\left(\bs \phi_1^{[1]}, \bs \phi_1^{[1]}, \ol{\bs \phi_1^{[1]}}, 2 \, \mbf W_0^{[2]} + \mbf W_2^{[1,1]}\right)
			\s 
			+ 12 \, \mbf F_4\left(\bs \phi_1^{[1]}, \bs \phi_1^{[1]}, \bs \phi_1^{[1]}, \ol{\mbf W_0^{[1,1]}}\right) + 24 \, \mbf F_4\left(\bs \phi_1^{[1]}, \ol{\bs \phi_1^{[1]}}, \ol{\bs \phi_1^{[1]}}, \mbf W_2^{[2]} + \mbf W_0^{[1,1]}\right)
            \s
            \left. + 60 \, \mbf F_5\left(\bs \phi_1^{[1]}, \bs \phi_1^{[1]}, \bs \phi_1^{[1]}, \ol{\bs \phi_1^{[1]}}, \ol{\bs \phi_1^{[1]}}\right)\right).
		\end{multline*}
		With this, considering Lemma \ref{simplificationorder5} we can choose, for simplicity,
		\begin{align*}
			h_2^{[5,0]} = h_1^{[4,1]} = h_2^{[3,2]} = h_1^{[0,5]} = h_2^{[1,4]} = h_1^{[2,3]} = 0,
		\end{align*}
		and we have, again, applying the inner product with the vector function $\bs{\hat \psi}_{1,\pm}^{[1]}$, that $h_1^{[5,0]} = C_1^{[5,0]} \, |A|^4 \, A$, $\ol{h_2^{[4,1]}} = C_1^{[1,4]} \, |A|^4 \, \bar B$, $h_1^{[3, 2]} = C_1^{[3, 2]} \, |A|^2 \, |B|^2 \, A$, $h_2^{[0, 5]} = \ol{C_1^{[5, 0]}} \, |B|^4 \, B$, $h_1^{[1, 4]} = C_1^{[1, 4]} \, |B|^4 \, A$, and $h_2^{[2, 3]} = \ol{C_1^{[3, 2]}} \, |A|^2 \, |B|^2 \, B$, where
		\begin{multline*}
			C_1^{[5,0]} = \frac{1}{\bs \phi_1^{[1]}\cdot \bs \psi_1^{[1]}} \, \bs \psi_1^{[1]}\cdot\left( - \left(2 \, C_1^{[3,0]} + \ol{C_1^{[3,0]}}\right) \mbf W_1^{[3]} + 2 \, \mbf F_2\left(\bs \phi_1^{[1]}, \mbf W_0^{[4]} + \ol{\mbf W_0^{[4]}}\right)\right.
			\s 
			+ 2 \, \mbf F_2\left(\ol{\bs \phi_1^{[1]}}, \mbf W_2^{[4]}\right) + 4 \, \mbf F_2\left(\mbf W_0^{[2]}, \mbf W_1^{[3]}\right) + 2 \, \mbf F_2\left(\mbf W_2^{[2]}, \ol{\mbf W_1^{[3]}}\right) + 2 \, \mbf F_2\left(\ol{\mbf W_2^{[2]}}, \mbf W_3^{[3]}\right)
			\s 
			+ 12 \, \mbf F_3\left(\bs \phi_1^{[1]}, \mbf W_0^{[2]}, \mbf W_0^{[2]}\right) + 6 \, \mbf F_3\left(\bs \phi_1^{[1]}, \mbf W_2^{[2]}, \ol{\mbf W_2^{[2]}}\right) + 12 \, \mbf F_3\left(\ol{\bs \phi_1^{[1]}}, \mbf W_0^{[2]}, \mbf W_2^{[2]}\right)
			\s 
			+ 3 \, \mbf F_3\left(\bs \phi_1^{[1]}, \bs \phi_1^{[1]}, \ol{\mbf W_1^{[3]}}\right) + 6 \, \mbf F_3\left(\bs \phi_1^{[1]}, \ol{\bs \phi_1^{[1]}}, \mbf W_1^{[3]}\right) + 3 \, \mbf F_3\left(\ol{\bs \phi_1^{[1]}}, \ol{\bs \phi_1^{[1]}}, \mbf W_3^{[3]}\right)
			\s 
			+ 4 \, \mbf F_4\left(\bs \phi_1^{[1]}, \bs \phi_1^{[1]}, \bs \phi_1^{[1]}, \ol{\mbf W_2^{[2]}}\right) + 24 \, \mbf F_4\left(\bs \phi_1^{[1]}, \bs \phi_1^{[1]}, \ol{\bs \phi_1^{[1]}}, \mbf W_0^{[2]}\right)
			\s 
			\left. + 12 \, \mbf F_4\left(\bs \phi_1^{[1]}, \ol{\bs \phi_1^{[1]}}, \ol{\bs \phi_1^{[1]}}, \mbf W_2^{[2]}\right) + 10 \, \mbf F_5\left(\bs \phi_1^{[1]}, \bs \phi_1^{[1]}, \bs \phi_1^{[1]}, \ol{\bs \phi_1^{[1]}}, \ol{\bs \phi_1^{[1]}}\right)\right),
		\end{multline*}
		\begin{multline*}
			C_1^{[1,4]} = \frac{1}{\bs \phi_1^{[1]} \cdot \bs \psi_1^{[1]}} \, \bs \psi_1^{[1]} \cdot \left(- \left(C_1^{[1,2]} + C_1^{[3,0]} + \ol{C_1^{[3,0]}}\right) \mbf W_1^{[2,1]} \vphantom{\frac{1}{4}}\right.
			\s 
			2 \, \mbf F_2\left(\bs \phi_1^{[1]}, \mbf W_0^{[4]} + \ol{\mbf W_0^{[4]}} + \ol{\mbf W_2^{[3,1]}}\right) + 2 \, \mbf F_2\left(\ol{\bs \phi_1^{[1]}}, \mbf W_0^{[3,1]}\right) + 4 \, \mbf F_2\left(\mbf W_0^{[2]}, \mbf W_1^{[2,1]}\right)
			\s 
			+ 2 \, \mbf F_2\left(\ol{\mbf W_2^{[2]}}, \mbf W_{1,2}^{[2,1]}\right) + 2 \, \mbf F_2\left(\mbf W_2^{[2]}, \ol{\mbf W_3^{[2,1]}}\right) + 2 \, \mbf F_2\left(\mbf W_1^{[3]}, \mbf W_2^{[1,1]}\right) + 2 \, \mbf F_2\left(\ol{\mbf W_1^{[3]}}, \mbf W_0^{[1,1]}\right)
			\s 
			+ 12 \, \mbf F_3\left(\bs \phi_1^{[1]}, \mbf W_0^{[2]}, \mbf W_0^{[2]} + \mbf W_2^{[1,1]}\right) + 6 \, \mbf F_3\left(\bs \phi_1^{[1]}, \ol{\mbf W_2^{[2]}}, \mbf W_2^{[2]} + \mbf W_0^{[1,1]}\right)
			\s 
			+ 12 \, \mbf F_3\left(\ol{\bs \phi_1^{[1]}}, \mbf W_0^{[2]}, \mbf W_0^{[1,1]}\right) + 6 \, \mbf F_3\left(\ol{\bs \phi_1^{[1]}}, \mbf W_2^{[2]}, \mbf W_2^{[1,1]}\right) + 6 \, \mbf F_3\left(\bs \phi_1^{[1]}, \ol{\bs \phi_1^{[1]}}, \mbf W_1^{[3]} + \mbf W_1^{[2,1]}\right)
			\s 
			+ 3 \, \mbf F_3\left(\ol{\bs \phi_1^{[1]}}, \ol{\bs \phi_1^{[1]}}, \mbf W_{1,2}^{[2,1]}\right) + 3 \, \mbf F_3\left(\bs \phi_1^{[1]}, \bs \phi_1^{[1]}, 2 \, \ol{\mbf W_1^{[3]}} + \ol{\mbf W_3^{[2,1]}}\right)
			\s 
			+ 12 \, \mbf F_4\left(\bs \phi_1^{[1]}, \bs \phi_1^{[1]}, \ol{\bs \phi_1^{[1]}}, 4 \, \mbf W_0^{[2]} + \mbf W_2^{[1,1]}\right) + 12 \, \mbf F_4\left(\bs \phi_1^{[1]}, \bs \phi_1^{[1]}, \bs \phi_1^{[1]}, \ol{\mbf W_2^{[2]}}\right)
            \s
            \left. + 12 \, \mbf F_4\left(\bs \phi_1^{[1]}, \ol{\bs \phi_1^{[1]}}, \ol{\bs \phi_1^{[1]}}, \mbf W_2^{[2]} + \mbf W_0^{[1,1]}\right) + 30 \, \mbf F_5\left(\bs \phi_1^{[1]}, \bs \phi_1^{[1]}, \bs \phi_1^{[1]}, \ol{\bs \phi_1^{[1]}}, \ol{\bs \phi_1^{[1]}}\right)\right),
		\end{multline*}
		and
		\begin{multline*}
			C_1^{[3,2]} = \frac{1}{\bs \phi_1^{[1]}\cdot \bs \psi_1^{[1]}} \, \bs \psi_1^{[1]} \cdot \left(- \left(C_1^{[3,0]} + C_1^{[1,2]} + \ol{C_1^{[1,2]}}\right) \mbf W_1^{[2,1]} - \left(2 \, C_1^{[1,2]} + \ol{C_1^{[1,2]}}\right) \mbf W_1^{[3]} \right.
			\s
			+ 2 \, \mbf F_2\left(\bs \phi_1^{[1]}, \mbf W_0^{[2,2]} + \ol{\mbf W_0^{[2,2]}} + \mbf W_2^{[3,1]}\right) + 2 \, \mbf F_2\left(\ol{\bs \phi_1^{[1]}}, \mbf W_2^{[2,2]} + \mbf W_0^{[3,1]}\right) + 4 \, \mbf F_2\left(\mbf W_0^{[2]}, \mbf W_1^{[2,1]}\right)
			\s 
			+ 2 \, \mbf F_2\left(\ol{\mbf W_1^{[2,1]}}, \mbf W_2^{[2]} + \mbf W_0^{[1,1]}\right) + 2 \, \mbf F_2\left(\mbf W_2^{[1,1]}, \mbf W_1^{[2,1]} + \mbf W_3^{[2,1]}\right) + 2 \, \mbf F_2\left(\ol{\mbf W_0^{[1,1]}}, \mbf W_{1,2}^{[2,1]}\right)
            \s
            + 4 \, \mbf F_2\left(\mbf W_0^{[2]}, \mbf W_1^{[3]}\right) + 6 \, \mbf F_3\left(\bs \phi_1^{[1]}, \mbf W_2^{[1,1]}, \mbf W_2^{[1,1]}\right) + 6 \, \mbf F_3\left(\bs \phi_1^{[1]}, \ol{\mbf W_0^{[1,1]}}, \mbf W_2^{[2]} + \mbf W_0^{[1,1]}\right)
			\s 
			+ 6 \, \mbf F_3\left(\ol{\bs \phi_1^{[1]}}, 2 \, \mbf W_0^{[2]} + \mbf W_2^{[1,1]}, \mbf W_2^{[2]} + \mbf W_0^{[1,1]}\right) + 12 \, \mbf F_3\left(\bs \phi_1^{[1]}, \mbf W_0^{[2]}, 2 \, \mbf W_0^{[2]} + \mbf W_2^{[1,1]}\right)
			\s 
			+ 9 \, \mbf F_3\left(\bs \phi_1^{[1]}, \bs \phi_1^{[1]}, \ol{\mbf W_1^{[2,1]}}\right) + 6 \, \mbf F_3\left(\bs \phi_1^{[1]}, \ol{\bs \phi_1^{[1]}}, \mbf W_1^{[3]} + 2 \, \mbf W_1^{[2,1]} + \mbf W_3^{[2,1]}\right)
			\s 
			+ 6 \, \mbf F_3\left(\ol{\bs \phi_1^{[1]}}, \ol{\bs \phi_1^{[1]}}, \mbf W_{1,2}^{[2,1]}\right) + 36 \, \mbf F_4\left(\bs \phi_1^{[1]}, \bs \phi_1^{[1]}, \ol{\bs \phi_1^{[1]}}, 2 \, \mbf W_0^{[2]} + \mbf W_2^{[1,1]}\right)
			\s 
			+ 12 \, \mbf F_4\left(\bs \phi_1^{[1]}, \bs \phi_1^{[1]}, \bs \phi_1^{[1]}, \ol{\mbf W_0^{[1,1]}}\right) + 24 \, \mbf F_4\left(\bs \phi_1^{[1]}, \ol{\bs \phi_1^{[1]}}, \ol{\bs \phi_1^{[1]}}, \mbf W_2^{[2]} + \mbf W_0^{[1,1]}\right)
            \s
            \left. + 60 \, \mbf F_5\left(\bs \phi_1^{[1]}, \bs \phi_1^{[1]}, \bs \phi_1^{[1]}, \ol{\bs \phi_1^{[1]}}, \ol{\bs \phi_1^{[1]}}\right)\right).
		\end{multline*}
\end{document}